\documentclass[format=acmsmall, review=false]{acmart}
\usepackage{acm-ec-26}

\usepackage{auxiliary/packages}
\usepackage{auxiliary/mymacros}
\usepackage{auxiliary/filecommands}

\setcitestyle{acmnumeric}

\title[Tractable General Equilibrium]{Tractable General Equilibrium}

\author{Denizalp Goktas}
\affiliation{
  \institution{Cornell Tech}
  \city{New York}
  \country{USA}
}
\email{dg776@cornell.edu}
\author{Kl\'ara Chura}
\affiliation{
  \institution{Brown University}
  \city{Providence}
  \country{USA}
}
\email{klara\_chura@brown.edu}
\author{Emin Alp Guneri}
\affiliation{
  \institution{Brown University}
  \city{Providence}
  \country{USA}
}
\email{emin\_alp\_guneri@brown.edu}
\author{Amy Greenwald}
\affiliation{
  \institution{Brown University}
  \city{Providence}
  \country{USA}
}
\email{amy@brown.edu}
   
\begin{abstract}
A key challenge in Walrasian economies is identifying price-adjustment processes that converge to (Walrasian) equilibrium. One such process, tâtonnement, is an auction-like algorithm first proposed in 1874 by Léon Walras. While tâtonnement is known to converge to equilibrium in economies satisfying the Weak Axiom of Revealed Preferences (WARP), the process fails to converge in a pathological Walrasian economy known as the Scarf economy. To address this issue, we analyze Walrasian economies using variational inequalities (VIs).
We study the class of mirror extragradient algorithms, which under Bregman continuity converge to an $\varepsilon$-approximate solution of any VI satisfying the Minty condition in $O\left(\nicefrac{1}{\varepsilon^2}\right)$ iterations. Applying the mirror extragradient algorithm to suitable VIs, we obtain a class of tâtonnement-like processes, which we call mirror extratâtonnement. We can then establish the analogous polynomial-time convergence of mirror extratâtonnement in WARP economies assuming bounded elasticity of demand, generalizing known results for economies that satisfy weak gross substitutes, as the corresponding VI satisfies the Minty condition. For the Scarf economy, we secure the Minty condition by expanding the search space to the unit box, and we require Bregman continuity along the path of extrat\^atonnement; under these modifications, our theory extends to this pathological economy. We validate our approach in experiments on large Arrow-Debreu economies, including the Scarf economy, demonstrating convergence in all cases within the bounds given by our theoretical guarantees. This work thus provides one resolution to the challenge set by Herbert Scarf’s 50-year-old agenda in general equilibrium theory, namely, providing "a general method for the explicit numerical solution of the neoclassical model."
\end{abstract}

\begin{document}

\begin{titlepage}

\maketitle

\vspace{1cm}
\setcounter{tocdepth}{2} 

\end{titlepage}

\section{Introduction}
Walrasian economies, introduced by French economist L\'eon Walras in 1874, are a broad mathematical framework for modeling any economic system governed by supply and demand \cite{walras}. A Walrasian economy consists of a finite set of commodities, characterized by an excess demand function that maps values for commodities, called \mydef{prices}, to positive (resp. negative) quantities of each commodity demanded (resp. supplied) in excess. Walras proposed a steady-state solution of his economy namely a collection of per-commodity prices which is \mydef{feasible}, i.e., there is no excess demand for any commodity, and for which \mydef{Walras' law} holds, i.e., the value of the excess demand is equal to 0. We call such a solution a Walrasian (or competitive) equilibrium.

Walras did not establish conditions ensuring the existence of an equilibrium, leaving the question unresolved until the 1950s \cite{arrow-debreu}, but argued
that his economy would settle at a Walrasian equilibrium via a \mydef{price-adjustment process} (i.e., any process that generates a sequence of prices based on prior prices and associated excess demands), known as \mydef{t\^atonnement}, which mimics the behavior of the \mydef{law of supply and demand}, updating prices at a rate equal to the excess demand \cite{walras, uzawa1960walras, arrow-hurwicz}. To motivate the relevance of t\^atonnement to real-world economies, Walras argued that t\^atonnement is a \mydef{natural} price-adjustment process, in the following sense: if each commodity is owned by a different seller, then each seller can update the price of its commodity, using only information about the excess demand of its commodity, without coordinating with other sellers. It is thus plausible that t\^atonnement could explain the movement of prices in real-world economies.

Nearly half a century after Walras' initial foray into general equilibrium analysis, a group of academics brought together by the Cowles Commission in 1939 reinitiated a study of Walras' economic model with the aim of bringing rigorous mathematics to the analysis of markets. One of the earliest and most important outputs of this collaborative effort was the introduction of a broad and well-justified class of Walrasian economies known as \mydef{competitive economies} \cite{arrow-debreu}, for which the existence of Walrasian equilibrium was established by a novel application of fixed point theorems to economics. With the question of existence thus resolved, the field subsequently turned its focus to questions on the \mydef{stability} of Walrasian equilibrium, i.e., which price-adjustment processes can settle at a Walrasian equilibrium and under which assumptions~\cite{uzawa1960walras,balasko1975some,arrow-hurwicz, cole2008fast, cheung2018dynamics,fisher-tatonnement, jain2005market, codenotti2005market, codenotti2006leontief, chen2009spending}?

Most relevant work on stability has been concerned with the convergence properties of t\^atonnement.
Beyond Walras' justification for t\^atonnement's relevance to real-world economies, research on t\^atonnement in the post-World-War-II economics literature is motivated by the fact that it can be understood as a plausible explanation of how prices move in real-world markets \cite{gillen2020divergence}. Hence, if one could prove that t\^atonnement is a \mydef{universal} price-adjustment process, i.e., one that converges to a Walrasian equilibrium in all competitive economies, then perhaps it would be justifiable to claim real-world economies would also eventually settle at a Walrasian equilibrium.

In 1958, \citet{arrow-hurwicz} 
established the convergence of a continuous-time variant of t\^atonnement in Walrasian economies with an excess demand function satisfying the weak axiom of revealed preferences (WARP) \cite{afriat1967construction}, which among others, includes Walrasian economies satisfying \mydef{gross substitutes (GS)} (i.e., the excess demand of any commodity can only increase when the price of another commodity increases, fixing all other prices) \cite{arrow1959stability, arrow1960competitive}. This result was complemented by \citeauthor{nikaido1960stability}'s (\citeyear{nikaido1960stability}) result on the convergence of a discrete-time variant of t\^atonnement in Walrasian economies satisfying WARP---albeit without any non-asymptotic convergence guarantees. These initial results sparked hopes that t\^atonnement could be a universal price-adjustment process.

Furthermore, as Walrasian equilibria in general cannot be expressed in closed form, these results ignited further interest in discovering algorithms to compute a Walrasian equilibrium, as t\^atonnement could be implemented on a computer to obtain numerical approximations of Walrasian equilibria in Walrasian economies.
Indeed, these early results on the stability of t\^atonnement inspired a new line of work on \mydef{applied general equilibrium} \cite{scarf1967computation, scarf1967approximation, scarf1973book, scarf1982computation} initiated by Herbert Scarf \cite{scarf-eaves}, whose goal was to establish ``a general method for the explicit numerical solution of the neoclassical [Walrasian economy] model'' \cite{scarf1973book}. The motivation behind this research agenda was a desire to predict the impact of economic policy on an economy by estimating the parameters of a parametric Walrasian economy from empirical data, and then running a comparative static analysis to compare the numerical solution of the Walrasian economy before and after the implementation of the policy.

Soon after this research agenda was initiated, Scarf dashed all hopes that t\^atonnement could be a universal price-adjustment process by exhibiting a competitive economy with only three commodities and an excess demand function generated by three consumers with Leontief preferences, i.e., \mydef{the Scarf economy}, for which the sequence of prices generated by a continuous-time variant of t\^atonnement cycles ad infinitum around the unique Walrasian equilibrium \cite{scarf1960instable}. Even more disheartening, the prices generated by discrete-time variants of t\^atonnement, when applied to the Scarf economy, spiral farther and farther away from the Walrasian equilibrium.

Scarf's negative result seems to have discouraged further research by economists on the stability of Walrasian equilibrium \cite{fisher1975stability}.
Despite research on this question coming to a near halt, one positive outcome was achieved, on the convergence of a non-t\^atonnement update rule known as \mydef{Smale's process} \cite{herings1997globally, kamiya1990globally, van1987convergent, smale1976convergent}, which updates prices at the rate of the product of the excess demand and the inverse of its Jacobian, to a Walrasian equilibrium in competitive economies that have an excess demand with a non-singular Jacobian, including Scarf economies. Unfortunately, this convergence result for Smale's process comes with two caveats: 1) Smale's process is not a ``natural" price-adjustment process, as it updates the price of each commodity using information not only about the excess demand of the commodity but also the derivative of the excess demand function with respect to each commodity in the economy, 2) convergence of discrete time-variants of Smale's process requires the excess demand to satisfy the law of supply and demand, which even Walrasian economies that satisfy the GS or WARP conditions need not satisfy. 

Nearly half a century after these seminal analyses of competitive economies, research on the stability and efficient computation of Walrasian equilibrium is once again coming to the fore, motivated by applications of algorithms to compute Walrasian equilibrium in dynamic stochastic general equilibrium models in macroeconomics \cite{geanakoplos1990introduction,sargent2000recursive,taylor1999handbook,FernandezVillaverde2023CompMethodsMacro}, and the use of algorithms such as t\^atonnement to solve models of transactions on crypotocurrency blockchains \cite{leonardos2021dynamical, liu2022empirical, reijsbergen2021transaction} and load balancing over networks \cite{jain2013constrained}.
In contrast to the prior literature on the stability of t\^atonnement, which was primarily concerned with proving asymptotic convergence of price-adjustment processes to Walrasian equilibria, this line of work is also concerned with obtaining non-asymptotic convergence rates, and hence computing approximate Walrasian equilibria in polynomial time.

The first result on this question is due to \citet{codenotti2005market}, who introduced a discrete-time version of t\^atonnement, and showed that in exchange economies that satisfy \mydef{weak gross substitutes (WGS)} (i.e., the excess demand of any commodity \emph{weakly} increases if the price of any other commodity increases, fixing all other prices), the t\^atonnement process converges to an approximate Walrasian equilibrium in a number of steps polynomial in the inverse of the approximation factor and the size of the problem.
Soon after this positive result appeared, \citet{papadimitriou2010impossibility} argued that no natural (i.e., decentralized)%
\footnote{\citet{papadimitriou2010impossibility} deemed a price-adjustment process ``natural" if the update to commodity $\good$'s price depends only on the history of commodity $\good$'s prices and excess demands, but not on the prices or excess demands of any other commodities.} price-adjustment process based on excess demand can converge to a Walrasian equilibrium in polynomial time in general in competitive economies, thus ruling out the possibility of Smale's process (and many others) justifying the notion of Walrasian equilibrium.
Nonetheless, further study of the convergence of price-adjustment processes such as t\^atonnement under stronger assumptions, or in simpler models than full-blown Arrow-Debreu competitive economies, continued, as these processes are being deployed in practice \cite{jain2013constrained, leonardos2021dynamical, liu2022empirical, reijsbergen2021transaction}.%
\footnote{We refer the reader to \Cref{sec_app:related} for additional related work on algorithms for Walrasian Economies and VIs.}

\subsection{A Tractable Variational Inequality Framework for Walrasian Economies}

The intent of this paper is to bridge the gap between the practical relevance of t\^atonnement and \citeauthor{papadimitriou2010impossibility}' \citep{papadimitriou2010impossibility} impossibility result.
To do so, we employ the variational inequality (VI) optimization framework \cite{lions1967variational}, exploiting a well-known connection between (Walrasian) VIs and Walrasian economies, 
namely that the set of Walrasian equilibria of any Walrasian economy can be characterized as the solution set of an associated complementarity problem (i.e., a VI where the constraint set is the positive orthant) \cite{dafermos1990exchange}.
Within the VI framework, we study the class of \mydef{mirror extragradient} algorithms \cite{nemirovski2004prox},%
\footnote{While \citet{nemirovski2004prox} introduced this algorithm as the ``prox-method,'' the name ``mirror extragradient''
is synonymous.}
which is known \citep{zhang2023mirror} \amy{correct reference?} to converge in polynomial time to a strong solution in VIs that satisfy a computational tractability condition known as the Minty condition \cite{minty1967generalization} 
and a generalization of Lipschitz continuity known as Bregman (or relative \cite{lu2019relative}) continuity.
Applied to Walrasian VIs, 
mirror extragradient corresponds to a novel, natural price-adjustment process we call \mydef{mirror extrat\^atonnement}.


\amy{The Minty condition can be interpreted as an analogue of the gradient dominance (Polyak–\L ojasiewicz) condition in non-convex optimization, as both rule out spurious stationary points by ensuring that every stationary point is globally optimal.}



Our focus in this paper is on \mydef{balanced economies}, a class of Walrasian economies that includes among others Arrow and Debreu's competitive economies \cite{arrow-debreu}.
In a balanced economy, excess demand is a homogeneous correspondence of degree $0$, and weak Walras' Law is satisfied (i.e., the value of the excess demand is less than or equal to $0$).
While it is known \cite{dafermos1990exchange} that the set of Walrasian equilibria of balanced economies is equal to the set of strong solutions of the \mydef{Walrasian VI} $(\R_+^\numgoods,-\excessset)$, where $\numgoods$ is the number of goods and $\excessset$ is the excess demand correspondence,
we observe that this connection extends to the modified Walrasian VI $([0,1]^\numgoods, -\excessset)$, which we call the \mydef{(Walrasian) box VI}, and to the modified Walrasian VI $(\simplex[\numgoods], -\excessset)$, which we call the \mydef{(Walrasian) simplex VI}.
\citet{papadimitriou2010impossibility} restrict their attention to price-adjustment processes that generate prices in the unit simplex.
The simplex VI, however, can fail to satisfy the Minty condition, and accordingly, can be computationally intractable. 
On the other hand, by relaxing this requirement, so that prices instead lie in the unit box, the Minty condition is satisfied, with $\zeros[m]$ as the Minty solution.
Consequently, if the excess demand correspondence can be shown to satisfy something like Bregman continuity, we can apply mirror extrat\^atonnement to solve it.

As it turns out, the box VI is Bregman discontinuous at one point in its search space, namely the Minty solution, i.e., when all prices are $0$. 
It is thus impossible to simultaneously ensure both the Minty condition and Bregman continuity 
of excess demand on the unit box in general.
Moreover, the only balanced economies with Lipschitz-continuous excess demand on the unit box are those with a constant excess demand.
It is therefore not possible to obtain any meaningful polynomial-time convergence results for mirror extrat\^atonnement applied to the box VI without further assumptions. 
Nevertheless, 
in our experiments, we observe the fast convergence of mirror extrat\^atonnement in a large class of competitive economies, 
including very large instances with Leontief consumers, for which the computation of a Walrasian equilibrium is known to be PPAD-complete \cite{codenotti2006leontief, deng2008computation}. 
To explain this empirical observation, we introduce a \mydef{pathwise Bregman continuity} assumption, which requires  Bregman-continuous excess demand only along the sequence of prices generated by mirror extrat\^atonnement.
We prove that this condition is sufficient to guarantee polynomial-time convergence to a Walrasian equilibrium in balanced and thus also competitive economies.

\subsection{Technical Contributions}

\paragraph{Variational Inequalities}
Our first contribution is an analysis of the class of constrained
mirror extragradient algorithms, a generalization of Korpelevich's extragradient method \cite{korpelevich1976extragradient} for solving VIs. For VIs that satisfy the Minty condition and are Bregman continuous, we establish best-iterate convergence to an $\varepsilon$-strong solution in $O(\nicefrac{1}{\varepsilon^2})$ evaluations of the VI's optimality operator (\Cref{thm:mirror_extragradient_global_convergence}). Our result generalizes the results and proof techniques of \citet{huang2023beyond} for the extragradient method, and extends the convergence results of \citet{zhang2023mirror} for the unconstrained mirror extragradient method to constrained domains. In addition, 
we establish suitable conditions for the local convergence of the mirror extragradient algorithm to an $\varepsilon$-strong solution of any Bregman-continuous VI that does \emph{not\/} satisfy the Minty condition---to the best of our knowledge, the first result of its kind (\Cref{thm:vi_mirror_extragrad_local}).

\vspace{-2mm}
\paragraph{Walrasian Economies}
It is known that the set of Walrasian equilibria of any Walrasian economy can be characterized as the solution set of an associated complementarity problem (i.e., a VI where the constraint set is the positive orthant) \cite{dafermos1990exchange}.
For balanced economies, however, we provide the first computationally tractable characterization of Walrasian equilibria as the set of scalar multiples ($\geq 1$) of the set of strong solutions of a box VI.
As the box trivially satisfies the Minty condition, we apply the mirror extragradient method to obtain a novel natural price-adjustment process we call mirror extrat\^atonnement (\Cref{alg:mirror_extratatonnement}), and establish its polynomial-time convergence in all balanced economies that satisfy pathwise Bregman continuity (\Cref{thm:bregman_mirror_exta_tatonn_convergence}).


Next, we apply these insights to the Scarf economy.
First, we observe that the Scarf economy is balanced (\Cref{lemma:Scarf_balanced}).
Then, we prove polynomial-time convergence of mirror extrat\^atonnement to the unique Walrasian equilibrium of the Scarf economy (\Cref{thm:scarf_convergence}).
As such, the mirror extrat\^atonnement process is the first discrete-time \emph{natural\/} price-adjustment process known to converge in the Scarf economy.

\amy{Rather than bound prices away from zero, we could instead assume bounded excess demand. The intent of either assumption is to ensure Bregman continuity.}
\klara{We can say that in practice, Scarf satisfies pathwise Bregman, hence in practice convergence is guaranteed.}
\klara{Note to self: rewrite proof of Scarf's full Bregman on subbox as modified Scarf's Bregman on box.}

While the pathwise Bregman continuity assumption provides intuition for the fast convergence of mirror extrat\^atonnement in practice, it is hard to verify this assumption analytically, in advance of running the process.
We obtain stronger theoretical results by restricting our attention to competitive economies with \mydef{bounded elasticity of excess demand} (i.e., the change in excess demand as prices change is bounded across all price changes)\sklara{}{, bounded aggregate demand, and bounded aggregate supply} that are \mydef{variationally stable} \cite{mertikopoulos_learning_2019} on the unit simplex (i.e., simplex VIs for which the Minty condition holds) \sklara{}{in \Cref{thm:mirror_extratatonn_var_stable}}. 
We show that under these assumptions, the excess demand is Bregman continuous, thus providing the first polynomial-time convergence result for a natural price-adjustment process in this broad class of Walrasian economies, which includes economies that satisfy weak GS \sklara{}{\cite{codenotti2005market}}, and more generally, WARP.

Finally, we run experiments (\Cref{sec:experiments}) which demonstrate that mirror extrat\^atonnement converges to a Walrasian equilibrium 
in a variety of large (600, 800, or 1000 consumers and 500 or 1000 commodities)
competitive economies, including some of which are known to be PPAD-complete (e.g., Leontief economies).
The convergence rate matches our theory because we first conduct a search for an appropriate step size, meaning one that ensures pathwise Bregman continuity.
In particular, convergence is fast in all economies except those with linear consumers, 
where our convergence bounds are weak due to the magnitude of the pathwise Bregman continuity coefficient.
We conclude with an additional set of experiments using step sizes that violate pathwise Bregman continuity, but we nonetheless observe a similar pattern of convergence.


\section{Preliminaries}
\paragraph{Notation.} 
We use caligraphic uppercase letters to denote sets (e.g., $\calX$), bold uppercase letters to denote matrices (e.g., $\allocation$), bold lowercase letters to denote vectors (e.g., $\price$), lowercase letters to denote scalar quantities (e.g., $x$). 
We denote the $i$th row vector of a matrix (e.g., $\allocation$) by the corresponding bold lowercase letter with subscript $i$ (e.g., $\allocation[\buyer])$. 
Similarly, we denote the $j$th entry of a vector (e.g., $\price$ or $\allocation[\buyer]$) by the corresponding lowercase letter with subscript $j$ (e.g., $\price[\good]$ or $\allocation[\buyer][\good]$).
We denote functions by a letter determined by the value of the function, e.g., $f$ if the mapping is scalar valued, $\f$ if the mapping is vector valued, and $\calF$ if the mapping is set valued (i.e., $\calF$ is a correspondence).
If a correspondence $\calF$ happens to be singleton valued, we overload notation and denote it by $\f$.
We denote the set $\left\{1, \hdots, n\right\}$ by $[n]$, the set of natural numbers by $\N$, and the set of real numbers by $\R$. 
We denote the positive and strictly positive elements of a set using a $+$ or $++$ subscript, respectively (e.g., $\R_+$ and $\R_{++}$). 

For any $n \in \N$, we denote the  $n$-dimensional vector of zeros and ones by $\zeros[n]$ and $\ones[n]$, respectively, and the $i^{th}$ basis vector in $\R^n$ by $\basis[i]$.
We let $\simplex[n] = \{\x \in \R_+^n \mid \sum_{i = 1}^n x_i = 1 \}$ denote the unit simplex in $\R^n$.
Unless otherwise noted, the norm notation denotes the 2-norm, i.e., $\| \cdot \| \doteq \| \cdot \|_2$, and we denote the Euclidean projection operator onto a set $C$ by $\project[C]$, i.e., $\project[C](\x) \doteq \argmin_{\y \in C} \left\|\x - \y \right\|^2$. Given a metric space $(\metricspace, \metric)$ and $\varepsilon \geq 0$, we write $\closedball[\varepsilon][\var] = \{ \var[\prime] \in \metricspace \mid \metric(\var, \var[\prime]) \leq \varepsilon \}$ to denote the closed $\varepsilon$-ball centered at $\var \in \metricspace$.
The multiplication of a scalar and a set is defined as the Minkowksi product, i.e., for all $a \in \R$ and $\set \subseteq \R^\numgoods$, we define $a \set \doteq \{a \vartuple \mid \vartuple \in \set \}$.

\paragraph{Functions. }
Given a Euclidean vector space $\set \subseteq \R^n$, we define its dual space $\set^*$ as the set of all linear maps $\vioper: \set \to \R^d$.
Let $(\set, \norm_{\set})$ and $(\otherset, \norm_{\otherset})$ be normed spaces. 
A function $\vioper: \set \to \otherset$ is \mydef{continuous} if for all sequences $\left\{ \var[(n)] \right\}_{n \in \N}$ s.t.\ $\var[(n)] \to \var \in \set$, it holds that $\obj(\var[(n)]) \to \obj(\var)$.
Given $\lipschitz \geq 0$, $\vioper$ is said to be $\lipschitz$-\mydef{Lipschitz continuous} on $\calA \subseteq \set$ iff for all $\x_1, \x_2 \in \calA, \left\| \obj(\x_1) - \obj(\x_2) \right\|_{\otherset} \leq \lipschitz \left\| \x_1 - \x_2 \right\|_{\set}$. 
If $f$ is differentiable, then it is $\lipschitz$-\mydef{Lipschitz smooth} if its gradient is $\lipschitz$-Lipschitz continuous.
A function $\obj: \set \to \R$ is \mydef{convex} iff for all $\lambda \in [0,1]$ and $\vartuple, \vartuple[][\prime] \in  \set$, 
$
\obj (\lambda \vartuple + (1-\lambda) \vartuple[][\prime]) \leq \lambda \obj(\vartuple) + (1-\lambda) \obj(\vartuple[][\prime])  \enspace .
$
Given $\sconvex \geq 0$, $\obj$ is $\sconvex$-strongly-convex\sklara{,}{} iff $\vartuple \mapsto \obj(\vartuple) - \nicefrac{\sconvex}{2}\| \vartuple\|^2$ is convex.

\paragraph{Correspondences.} 
Let $(\set, \innerprod)$ be an inner product space. 
A correspondence $\relation: \set \rightrightarrows \set^*$ is said to be \mydef{upper hemicontinuous} if for any sequence $\{(\var[(n)], y^{(n)}) \}_{n \in \N} \subset \set \times \set^*$ that converges to $(\var, y)$
with $y^{(n)} \in \relation(\var[(n)])$, for all $n \in \N$, it also holds that $y \in \relation(\var)$. 
$\relation$ is \mydef{continuous} if for any sequence $\left\{ \var[(n)] \right\}_{n \in \N} \subset \set$ that converges to $\var$
it also holds that $\relation(\var[(n)]) \to \relation(\var)$. 
$\relation$ is said to be \mydef{closed-valued} (resp.\ \mydef{compact-valued} / \mydef{convex-valued} / \mydef{singleton-valued}) iff for all $\var \in \set$, $\relation(\var)$ is closed (resp.\ compact / convex / a singleton).
$\relation$ is \mydef{monotone} iff for all $\var, \var[\prime] \in \set$ and $\othervar \in \relation(\var),  \othervar[\prime] \in  \relation(\var[\prime])$,
$
    \left< \othervar[\prime] - \othervar, \var[\prime] - \var \right> \geq 0.
$
$\relation$ is \mydef{pseudomonotone} iff for all $\var, \var[\prime] \in \set$, $\othervar \in \relation(\var)$, and $\othervar[\prime] \in  \relation(\var[\prime])$, 
$
    \left< \othervar, \var[\prime] - \var \right> \geq 0 \, \text{ implies } \left< \othervar[\prime], \var[\prime] - \var \right> \geq 0
$.
$\relation$ is \mydef{quasimonotone} iff for all $\var, \var[\prime] \in \set$, $\othervar \in \relation(\var)$, and $\othervar[\prime] \in  \relation(\var[\prime])$, 
$
    \left< \othervar, \var[\prime] - \var \right> > 0 \, \text{ implies } \left< \othervar[\prime], \var[\prime] - \var \right> \geq 0
$. 
We note the following relationship between these notions of monotonicity: $\mathrm{monotone} \implies \mathrm{pseudomonotone} \implies \mathrm{quasimonotone}$.

\paragraph{Bregman divergence}
 Given a set $\set$ and a \mydef{kernel function}  $\kernel: \set \to \R$, the \mydef{Bregman divergence} $\divergence[\kernel]: \set \times \set \to \R$ associated with $\kernel$ is defined as
$
    \divergence[\kernel](\vartuple, \othervartuple) \doteq \kernel(\vartuple) - \kernel(\othervartuple) - \innerprod[{\grad \kernel(\othervartuple)}][{\vartuple - \othervartuple}].
$
If $\kernel$ is convex, then the Bregman divergence is non-negative.
If $\kernel$ is strictly convex, then $\divergence[\kernel](\vartuple, \othervartuple) = 0$ iff $\vartuple = \othervartuple$.
If $\kernel$ is $\sconvex$-strongly convex, then for all $\vartuple, \othervartuple \in \set$, we have $\divergence[\kernel](\vartuple, \othervartuple) \geq \nicefrac{\sconvex}{2} \| \vartuple - \othervartuple\|^2$. 
When the kernel function is chosen s.t.\ $\kernel(\vartuple) \doteq \nicefrac{1}{2}\norm[\vartuple]^2$, then the Bregman divergence corresponds to Euclidean square distance, i.e., $\divergence[\kernel][\vartuple][{\othervartuple}] \doteq \frac{1}{2} \norm[{\vartuple - \othervartuple}]^2$.
%
Given a modulus of continuity $\lsmooth \geq 0$ and a kernel function $\kernel$, a function $\vioper: \set \to \otherset$ is said to be \mydef{$(\lsmooth, \kernel)$-Bregman-continuous} (or relatively continuous \cite{lu2019relative}) on $\calA \subseteq \set$ iff for all $\vartuple, \othervartuple \in \calA$, $\nicefrac{1}{2} \norm[\vioper(\vartuple) - \vioper(\othervartuple)]^2 \leq \lsmooth^2 \divergence[\kernel][{\vartuple}][{\othervartuple}]$.

\section{Variational Inequalities}
\label{section:vis}
Consider an inner product space $(\universe, \innerprod)$. A 
\mydef{variational inequality (VI)}, denoted $(\set, \vioperset)$, 
%
%
comprises a \mydef{constraint set} $\set \subseteq \universe$ and an \mydef{optimality operator} $\vioperset: \universe \rightrightarrows \universe^*$. For notational convenience, for any $\vartuple \in \set$, we denote any arbitrary element of $\vioperset(\vartuple)$ by $\vioper(\vartuple)$, and denote the variational inequality by $(\set, \vioper)$ when $\vioperset$ is singleton-valued.%
\footnote{In the literature, a nomenclatural distinction is sometimes made between a \mydef{variational inequality} where the optimality operator is a function and a \mydef{generalized variational inequality}, where it is a correspondence. Since we make this distinction through our choice of notation, we forego the ``generalized'' nomenclature.}

\subsection{Solution Concepts and their Properties}

The canonical solution concept for VIs is the strong, or {\mydef{Stampacchia}} \cite{lions1967variational}, solution. In practice, it is not possible to compute an exact strong solution 
to an arbitrary VI $(\set, \vioperset)$,%
\footnote{The exact strong solution may be irrational, which is computationally intractable (see \citet{anagnostides2025polynomialtimealgorithmvariationalinequalities}).}
so we resort to an approximate solution.
Given an \mydef{approximation parameter} $\vepsilon \geq 0$, an $\vepsilon$-\mydef{strong} (or $\vepsilon$-\mydef{Stampacchia}) \mydef{solution} of the VI $(\set, \vioperset)$ is an $\vartuple[][*] \in \set$ s.t.\ for all $\vartuple \in \set$, there exists an $\vioper(\vartuple[][*]) \in \vioperset(\vartuple[][*])$ s.t.\ $\innerprod[{\vioper(\vartuple[][*])}][{\vartuple[][*]  - \vartuple }] \leq \vepsilon$.
A $0$-strong solution is simply called a \mydef{strong solution}. 
We denote the set of $\varepsilon$-strong (resp.\ the set of strong) solutions of a VI $(\set, \vioperset)$ by $\svi[\varepsilon](\set, \vioperset)$ (resp.\ $\svi(\set, \vioperset)$).

Strong solutions are known to exist in a broad of class of VIs known as continuous. 
A \mydef{continuous} VI is a VI $(\set, \vioperset)$ s.t.\ $\set$ is non-empty, compact, and convex and $\vioperset$ is  upper hemicontinuous, non-empty-, compact-, and convex-valued.
One can establish the existence of a strong solution in a continuous VI by defining a mapping whose fixed points correspond to the strong solutions of the VI, and then invoking the Glicksberg-Kakutani fixed point theorem. We refer the reader to Theorem 2.2.1 of \citet{facchinei2003finite}.

An alternative but related solution to a VI is the weak (or Minty) solution \cite{minty1967generalization}, for which we can once again define an approximate variant for computational purposes. Given a VI $(\set, \vioperset)$ and an \mydef{approximation parameter} $\vepsilon \geq 0$, an \mydef{$\vepsilon$-weak (or Minty) solution} is an $\vartuple[][*] \in \set$ s.t.\ for all $\vartuple \in \set, \vioper(\vartuple) \in \vioperset(\vartuple)$, it holds that
$\innerprod[{\vioper(\vartuple)}][{\vartuple[][*] - \vartuple}] \leq \vepsilon$.
A $0$-weak solution to the VI is simply called a \mydef{weak solution}. 
We denote the set of $\varepsilon$-weak (resp.\ the set of weak) solutions a VI $(\set, \vioperset)$ by $\mvi[\varepsilon](\set, \vioperset)$ (resp.\ $\mvi(\set, \vioperset)$).

In continuous VIs, $\svi$ is a refinement of $\mvi$, so that any weak solution is also strong.
Surprisingly, an $\vepsilon$-weak-solution is not guaranteed to be $\vepsilon$-strong.
If the optimality operator $\vioperset$ is monotone though, the set of strong and weak solutions coincide. 
Moreover, any $\vepsilon$-strong solution is also $\vepsilon$-weak, but not vice versa.

A VI $(\set, \vioperset)$ is \{ \mydef{monotone}, \mydef{pseudomonotone}, \mydef{quasimonotone} \} iff the optimality operator $\vioperset$ is \{ monotone, pseudomonotone, quasimonotone \}.

Unlike strong solutions, the existence of weak solutions is not guaranteed in continuous VIs.
A VI $(\set, \vioperset)$ is said to satisfy the \mydef{Minty condition} iff the set of weak solutions is non-empty, i.e., $\mvi(\set, \vioperset) \neq \emptyset$.
With these definitions in place, we summarize the following known properties of the solution sets of VIs. 


\begin{remark}[Solution Set Properties]
\label{remark:sol_set_propert}
Assuming $\varepsilon \geq 0$, the following implications hold:

    \begin{itemize}
        \item $(\set, \vioperset)$ is continuous implies $\svi(\set, \vioperset) \neq \emptyset$ (Theorem 2.2.1 of \citet{facchinei2003finite}))
        
        \item $(\set, \vioperset)$ is continuous implies $\mvi(\set, \vioperset) \subseteq \svi(\set, \vioperset)$ 
        
        \item $(\set, \vioperset)$ is monotone implies $\svi[\varepsilon](\set, \vioperset) \subseteq \mvi[\varepsilon](\set, \vioperset)$
        
        \item $(\set, \vioperset)$ is pseudomonotone implies $\svi(\set, \vioperset) \subseteq \mvi(\set, \vioperset)$
        
        \item $(\set, \vioperset)$ is quasimonotone with $\set$ non-empty and compact implies Minty's condition, i.e., $\mvi(\set, \vioperset) \neq \emptyset$ (Lemma 3.1 of \cite{he2017solvability})
        
        \item If $\svi(\set, \vioperset) \neq \emptyset$, then monotone implies pseudomonotone, which in turn implies Minty's condition, i.e., $\mvi(\set, \vioperset) \neq \emptyset$ 
    \end{itemize}
    
\end{remark}

\subsection{First-Order Methods}
We now turn our attention to the computation of solutions to variational inequalities.
For simplicity, we restrict our attention to VIs $(\set, \vioperset)$ in which $\vioperset$ is singleton-valued, which we denote as $(\set, \vioper)$. 
In future work, analogous results to those described here could be developed for the more general non-singleton-valued VI setting. 
\klara{In the application to general equilibrium theory, set-values seem only to arise at indifference points, whose set has zero-measure. So, smoothing would appear a more natural and productive approach.} \klara{I've made edits assuming that results for correspondences might nonetheless find applications to general equilibrium search.} \klara{Since a correspondence, albeit unsolved, arises in the case of linear utilities, I propose keeping this as is.}

We consider first-order methods for computing strong solutions of VIs.
Given a VI $(\set, \vioper)$, and an initial iterate $\vartuple[][][0] \in \set$, a \mydef{first-order method} $\kordermethod: \bigcup_{\numhorizons \geq 1} (\set \times \set^*)^{\numhorizons} \to \set$ consists of an update function that generates the sequence of iterates $\left\{ \vartuple[][][\numhorizon] \right\}_{\numhorizon \in \N}$, 
given by $\vartuple[][][\numhorizon + 1] \doteq \kordermethod \left( \left\{ (\vartuple[][][i], \vioper(\vartuple[][][i])) \right\}_{i=0}^{\numhorizon} \right)$.
When $\kordermethod$ depends solely on the last item in the sequence, we simply write $ \vartuple[][][\numhorizon + 1] = \kordermethod(\vartuple[][][\numhorizon], \vioper(\vartuple[][][\numhorizon]))$. 

A common assumption that is used to establish polynomial-time convergence to strong solutions of VIs is Lipschitz continuity. 
Given a modulus of continuity $\lsmooth \geq 0$, a $\lsmooth$-\mydef{Lipschitz-continuous} VI is a VI $(\set, \vioper)$ s.t. $\set$ is non-empty, compact, and convex and $\vioper$ is $\lsmooth$-Lipschitz-continuous.
As is standard in the literature (see, for instance, \citet{cai2022tight}), the computational complexity measures in this section consider the number of evaluations of the optimality operator $\vioper$ as the unit of account.

\subsubsection{Mirror Gradient Algorithm}

The canonical class of first-order methods for VIs is the class of \mydef{mirror gradient algorithms} \cite{nemirovskij1983problem}. 
These algorithms are parameterized by a kernel function $\kernel: \set \to \R$, which induces a Bregman divergence $\divergence[\kernel]: \set \times \set \to \R$ that in turn defines the update function $\kordermethod^{\mathrm{MG}}(\othervartuple,  \vioper(\othervartuple)) \doteq \argmin\limits_{\vartuple \in \set} \left\{ \innerprod[{\vioper(\othervartuple)}][{ \vartuple - \othervartuple}] + \frac{1}{2 \learnrate[][ ]} \divergence[\kernel][\vartuple][{\othervartuple}]\right\} 
$. When the kernel function is chosen s.t.\ $\kernel(\vartuple) \doteq \frac{1}{2}\norm[\vartuple]^2$, the mirror gradient method $\kordermethod^{\mathrm{MG}}$ reduces to the well-known \mydef{projected gradient} method $\kordermethod^{\mathrm{PG}}$  \cite{cauchy1847methode}, i.e., $\kordermethod^{\mathrm{PG}}(\othervartuple,  \vioper(\othervartuple)) \doteq \proj[\outerset] \left[ \othervartuple - \learnrate[ ][ ] \vioper(\othervartuple) \right]$.

The average of the iterates of the mirror gradient method converges asymptotically to a strong solution in monotone and Lipschitz-continuous VIs \cite{nemirovskij1983problem}.
In general, however, it is only possible to prove polynomial-time (i.e., non-asymptotic) computation of an $\varepsilon$-weak solution in such VIs, which does not necessarily imply convergence to an $\varepsilon$-strong solution (see, for instance, Proposition~8 and Appendix~D of \citet{liu2021first}).
Even more troublesome, the sequence of iterates generated by the mirror gradient method is not guaranteed to converge at all, nor does averaging the iterates imply polynomial-time computation of an $\varepsilon$-strong solution, as shown by \Cref{example:monotone_non_convergence} in Appendix~\ref{sec_app:vi_examples}.


\subsubsection{Global Convergence of the Mirror Extragradient Algorithm}
%
%
Mirror extragradient (\Cref{alg:VI_mirror_extragrad}, \cite{nemirovski2004prox}) generalizes the well-known extragradient algorithm \cite{korpelevich1976extragradient}, which is known to asymptotically converge to a strong solution under reasonable assumptions%
\footnote{A closed convex constraint set, $L$-Lipschitzness of the operator, and a non-empty set of strong solutions.}
\cite{popov1980modification}, and allows for the polynomial-time computation of an $\varepsilon$-strong solution \cite{nemirovski2004prox, golowich2020eglast, cai2022tight}.
Similar to the class of mirror gradient methods, these algorithms are parameterized by a kernel function $\kernel: \set \to \R$, which induces a Bregman divergence $\divergence[\kernel]: \set \times \set \to \R$ that defines the update function $\kordermethod^{\mathrm{MEG}}(\othervartuple,  \vioper(\othervartuple)) \doteq  \kordermethod^{\mathrm{MG}}\left(\othervartuple,  \vioper(\kordermethod^{\mathrm{MG}}( \othervartuple, \vioper(\othervartuple)))\right)$. 
Moreover, when the kernel function for mirror extragradient is chosen s.t.\ $\kernel(\vartuple) \doteq \frac{1}{2}\norm[\vartuple]^2$, it reduces to Korpelevich's projected extragradient algorithm \cite{korpelevich1976extragradient}.

\begin{figure}
\begin{minipage}{0.65\textwidth}
\begin{algorithm}[H]
\caption{Mirror Extragradient Algorithm}
\label{alg:VI_mirror_extragrad}
\begin{flushleft}
\textbf{Input:} $\set, \vioper,  \kernel, \numhorizons, \learnrate[ ][ ], \vartuple[][][0]$\\
\textbf{Output:} $\{\vartuple[][][\numhorizon + 0.5], \vartuple[][][\numhorizon + 1]\}_{\numhorizon \in [\numhorizons]}$
\end{flushleft}
\begin{algorithmic}[1]
\For {$\numhorizon = 1, \hdots, \numhorizons $}
    \State $\vartuple[][][\numhorizon + 0.5] \gets 
    \argmin\limits_{\vartuple \in \set} \left\{ \innerprod[{\vioper(\vartuple[][(\numhorizon)])}][{ \vartuple - \vartuple[][(\numhorizon)]}] + \frac{1}{2 \learnrate[][ ]} \divergence[\kernel][\vartuple][{\vartuple[][(\numhorizon)]}]\right\}  $
    \State $\vartuple[][][\numhorizon +1] \gets 
    \argmin\limits_{\vartuple \in \set} \left\{ \innerprod[{\vioper(\vartuple[][(\numhorizon + 0.5)])}][{ \vartuple - \vartuple[][(\numhorizon)]}] + \frac{1}{2 \learnrate[][ ]} \divergence[\kernel][\vartuple][{\vartuple[][(\numhorizon)]}]\right\}$
\EndFor
\State \Return $\{\vartuple[][][\numhorizon + 0.5], \vartuple[][][\numhorizon + 1]\}_{\numhorizon \in [\numhorizons]}$
\end{algorithmic}
\end{algorithm}
\end{minipage}
\end{figure}

A seminal result by \citet{nemirovski2004prox} shows that the average of the iterates output by the extragradient algorithm is an $\varepsilon$-strong solution for any monotone VI with a Lipschitz-continuous optimality operator when the algorithm is run for $O(\nicefrac{1}{\varepsilon})$ iterations.
In the same setting, \citet{golowich2020eglast} and \citet{ cai2022tight} show best-iterate convergence to an $\varepsilon$-strong solution in $O(\nicefrac{1}{\varepsilon^2})$ iterations.
Later, \citet{huang2023beyond} extended this polynomial-time computation result to VIs which satisfy the weaker Minty condition rather than the monotonicity assumption.

We extend \citeauthor{huang2023beyond}'s result to mirror extragradient algorithm with the following theorem.
Our result states that an $\varepsilon$-strong solution of a VI $(\set, \vioper)$ can be found in polynomial time if $(\set, \vioper)$ satisfies the Minty condition and is \mydef{pathwise Bregman\sklara{}{-}continuous} over the outputs of the mirror extragradient method $\left\{ \vartuple[][][\numhorizon + 0.5], \vartuple[][][\numhorizon] \right\}_{\numhorizon \in [\numhorizons]}$, i.e., there exists $\lsmooth \geq 0$, s.t.\ for all $\numhorizon \in [\numhorizons]$, 
$\frac{1}{2} \norm[\vioper({\vartuple[][][t+0.5]}) - \vioper({\vartuple[][][t]})]^2 \leq \lsmooth^2 \divergence[\kernel][{\vartuple[][][t+0.5]}][{\vartuple[][][t]}]$.
As we show in \Cref{section:walrasian_economies}, this weaker pathwise Bregman continuity condition can be useful in the analysis of price-adjustment processes, and has indeed found such applications in the past (see, for example, \citet{fisher-tatonnement}).%
\footnote{All proofs omitted from this section can be found in Appendix~\ref{sec_app:vis}.}

\begin{restatable}[Mirror Extragradient Method Convergence]{theorem}{thmmirrorextragradglobal}
\label{thm:mirror_extragradient_global_convergence}
    Let $(\set, \vioper)$ be a continuous 
    VI satisfying the Minty condition and let $\kernel$ be a $1$-strongly-convex and $\kernelsmooth$-Lipschitz-smooth kernel function.%
    \footnote{The assumption that $\kernel$ is $1$-strongly-convex is without loss of generality since any $\mu$-strongly-convex kernel $\kernel^\prime$ can be converted to a $1$-strongly-convex kernel $\frac{1}{\mu} \kernel^\prime$.}
    %
    Imagine running the mirror extragradient algorithm (\Cref{alg:VI_mirror_extragrad}) on the VI $(\set, \vioper)$, with kernel function $\kernel$, time horizon $\numhorizons \in \N$, and initial iterate $\vartuple[][][0] \in \set$.
    Let $\left\{ \vartuple[][][\numhorizon + 0.5], \vartuple[][][\numhorizon + 1] \right\}_{\numhorizon \in [\numhorizons]}$ be the sequence of outputs generated, and assume there exists a step size $\learnrate[ ][ ] > 0$ and a corresponding $\lsmooth \in (0, \frac{1}{\sqrt{2}\learnrate[ ][ ]}]$ s.t.\ $\frac{1}{2} \norm[\vioper({\vartuple[][][k+0.5]}) - \vioper({\vartuple[][][k])}]^2 \leq \lsmooth^2 \divergence[\kernel][{\vartuple[][][k+0.5]}][{\vartuple[][][k]}]$ for all $k \in [\numhorizons]$.
    Then, the following bound holds:
    \begin{align}
    \label{eq:estrong_exists}
     \min_{k = 0, \hdots, \numhorizons} \max_{\vartuple \in \set} \ \langle \vioper(\vartuple[][][k + 0.5]), \vartuple[][][k + 0.5] - \vartuple \rangle \leq \, \frac{2 (1 + \kernelsmooth) \diam(\set)}{\learnrate[ ][ ]} \frac{\sqrt{\divergence[\kernel][{\vartuple[][*]}][{\vartuple[][][0]}]}}{\sqrt{\numhorizons}} \enspace ,
    \end{align}
    where $\vartuple[][*] \in \mvi(\set, \vioper)$ is a weak solution of $(\set, \vioper)$.

    In addition, for all $\varepsilon > 0$, if there exist time horizon $\altnumhorizons[\varepsilon] \gg O(\nicefrac{1}{\varepsilon^2})$ and step size
    $\learnrate[\varepsilon][ ] > 0$ with corresponding $\lsmooth[\varepsilon] \in (0, \frac{1}{\sqrt{2}\learnrate[\varepsilon][ ]}]$ s.t.\ $\frac{1}{2} \norm[\vioper({\vartuple[][][k+0.5]}) - \vioper({\vartuple[][][k])}]^2 \leq \lsmooth[\varepsilon]^2 \divergence[\kernel][{\vartuple[][][k+0.5]}][{\vartuple[][][k]}]$ for all $k \in [\altnumhorizons[\varepsilon]]$, then there exists a choice of time horizon $\numhorizons[\varepsilon] \ \dotin \ O \left( \nicefrac{\kernelsmooth^2\diam(\set)^2\divergence[\kernel][{\vartuple[][*]}][{\vartuple[][][0]}]}{\learnrate[\varepsilon][ ]^2\varepsilon^2} \right)$
    s.t.\ $\bestiter[\vartuple][{\numhorizons[\varepsilon]}] \in \argmin_{\vartuple[][][k+0.5] : k = 0, \hdots, \numhorizons[\varepsilon]} \divergence[\kernel] (\vartuple[][][k+0.5], \vartuple[][][k])$ is an $\varepsilon$-strong solution of $(\set, \vioper)$.
%
\end{restatable}

\klara{The $\eta,\lambda,\tau$ triple is specific to $\varepsilon$ (written out in my email from last night, validating Alp's experiments).}
\klara{The main point is $\altnumhorizons[\varepsilon] \gg O \left( 1 / \varepsilon^2 \right)$ while $\numhorizons[\varepsilon] \in O \left( 1 /\varepsilon^2 \right)$, treating the learning rate as a constant (because we assume it fixed when looking for our $\varepsilon$-eqm). We just need to convey this meaning, because an exact mathematical formulation, e.g. spelling out the big-O, would make the proof much, much longer.}

\amy{DISCUSS REMARK TOGETHER!}

\begin{remark}
Here are a few remarks intended to help the reader interpret this theorem.
1.~Even if the weak solution $\vartuple[][*]$ is also a strong solution, as in the continuous VIs we study, mirror extragradient may still converge to a strong solution other than $\vartuple[][*]$.
2.~While the right-hand side of \Cref{eq:estrong_exists} in \Cref{thm:mirror_extragradient_global_convergence} includes terms denoting divergence ($\divergence[\kernel]$) and distance ($\kernelsmooth$, $\diam(\set)$, and $\learnrate[ ][ ]$), the left-hand side only captures how ``strong'' the best iterate is, rather than what separates it from a strong solution.
3.~Since obtaining the best iterate according to the bound in \Cref{eq:estrong_exists} would require solving $\numhorizons$ optimization problems, $\bestiter[\vartuple][{\numhorizons[\varepsilon]}]$ is a more practical alternative.
\end{remark}

\klara{I'm still not super confident about the typing (in a programming-languages sense) of this whole thing.}

\subsubsection{Local Convergence of the Mirror Extragradient Algorithm}

Beyond VIs for which the Minty condition holds, it seems unlikely that we can devise a first-order method which, given an arbitrary initial iterate, converges to a strong solution, as the computation of an $\vepsilon$-strong solution for Lipschitz-continuous VIs is in general a PPAD-complete problem \cite{kapron2024computational}.%
\footnote{See also \Cref{example:non_convergence_non_minty} in Appendix~\ref{sec_app:vi_examples} for an example involving 
a strictly concave function, for which both the mirror gradient and mirror extragradient methods diverge to infinity when initialized at a negative-valued point.}
As a result, 
we investigate the conditions that guarantee local convergence of mirror extragradient to a (global) strong solution. \klara{Ahhh, it's the continuous setting, so the local Minty solution is itself the sought global strong solution\ldots}

Given a VI $(\set, \vioperset)$, and a \mydef{locality parameter} $\vdelta \geq 0$, a \mydef{$\vdelta$-local weak solution} of the VI is an $\vartuple[][*] \in \set$ s.t.\ for all $\vartuple \in \set \cap \closedball[\vdelta][{\vartuple[][*]}]$ and $\vioper(\vartuple) \in \vioperset(\vartuple)$, it holds that $\innerprod[{\vioper(\vartuple[][])}][{\vartuple - \vartuple[][*]}] \leq 0$  \cite{aussel2024variational} 
We denote the set of $\vdelta$-local weak solutions of a VI $(\set, \vioperset)$ by $\lmvi[][\vdelta] (\set, \vioperset)$. 
With these definitions in place, we now show that we can guarantee local convergence by assuming the algorithm is initialized close enough to a local weak (or Minty) solution.
As we will not be applying this result in this paper, for simplicity, we present it under the assumption of Lipschitz continuity.

\klara{What are some applications?} \amy{agreed, even future applications?} \klara{The Minty non-existence example in the ``GAES, revisited'' notes would be one such case. The outer two of the three stationary points there are local Minty.}


\begin{restatable}[Mirror Extragradient Method Local Convergence]{theorem}{thmvimirrorextragradlocal}
\label{thm:vi_mirror_extragrad_local}
    Let $(\set, \vioper)$ be a $\lsmooth$-Lipschitz-continuous VI, let $\kernel$ be a $1$-strongly-convex and $\kernelsmooth$-Lipschitz-smooth kernel function, and let $\learnrate[ ][ ] \in  \left(0, \frac{1}{\sqrt{2}\lsmooth}\right]$.
    Assume $\|\vioper\|_\infty \leq \lipschitz < \infty$, and that for some $\vartuple[][*] \in \lmvi[][\delta] (\set, \vioper)$ $\delta$-local weak solution, the initial iterate $\vartuple[][][0] \in \set$ is chosen s.t.\ $\sqrt{2 \divergence[\kernel][{\vartuple[][*]}][{\vartuple[][][0]}]} \leq \delta - \learnrate[ ][ ] \lipschitz$.
     Imagine running the mirror extragradient algorithm (\Cref{alg:VI_mirror_extragrad}) on the VI $(\set, \vioper)$, with kernel function $\kernel$, time horizon $\numhorizons \in \N$, initial iterate $\vartuple[][][0] \in \set$, and step size $\learnrate[ ][ ]$.
     Let $\left\{ \vartuple[][][\numhorizon + 0.5], \vartuple[][][\numhorizon + 1] \right\}_{\numhorizon \in [\numhorizons]}$ be the sequence of outputs generated. 
     Then, the following bound holds:
    \begin{align*}
     \min_{k = 0, \hdots, \numhorizons} \max_{\vartuple \in \set} \ \langle \vioper (\vartuple[][][k + 0.5]), \vartuple[][][k + 0.5] - \vartuple \rangle \leq \frac{\sqrt{2} (1 + \kernelsmooth) \diam(\set)}{\learnrate[ ][ ]} \frac{\delta}{\sqrt{\numhorizons}} \enspace .
    \end{align*}
    In addition, there exists time horizon 
    $\numhorizons[\varepsilon] \ \dotin \ O(\nicefrac{1}{\varepsilon^2})$ s.t.\ $\bestiter[\vartuple][{\numhorizons[\varepsilon]}] \in \argmin_{\vartuple[][][k+0.5] : k = 0, \hdots, \numhorizons} \divergence[\kernel] (\vartuple[][][k+0.5], \vartuple[][][k])$ is an $\varepsilon$-strong solution of $(\set, \vioper)$ and $\vartuple[][**] \doteq \lim_{\numhorizon \to \infty} \vartuple[][][\numhorizon +0.5] = \lim_{\numhorizon \to \infty} \vartuple[][][\numhorizon]$ is a strong solution of $(\set, \vioper)$.
\end{restatable}

We now turn our attention to an application of Theorem~\ref{thm:mirror_extragradient_global_convergence} to Walrasian economies.

\section{Walrasian economies}
\label{section:walrasian_economies}

A \mydef{Walrasian economy} $(\numgoods, \excessset)$ consists of $\numgoods \in \N$ \mydef{commodities},%
\footnote{The ``commodity'' terminology is used here in the tradition of \citet{arrow-debreu}, and refers to any raw, intermediate, or finished products, as well as labor and services.}
with any quantity of each commodity being exchangeable for a quantity of another. 
The exchange process is governed by a valuation system called \mydef{prices}, modeled as a vector $\price \in \R^\numgoods_+$ s.t.\ $\price[\good] \geq 0$ is the price of commodity $\good \in \goods$.
Prices $\price \in \R^\numgoods_+$ allow the exchange of $x \in \R_+$ units of any commodity $\good \in \goods$ for $x \left( \nicefrac{\price[\good]}{\price[k]} \right)$ units of some other commodity $k \in \goods$.
For example, if $j$ represents apples whose price is \$1 and $k$ represents bananas whose price is \$2, then one apple can be exchanged for half of a banana.

A
Walrasian economy is characterized by an \mydef{excess demand correspondence} $\excessset: \R^\numgoods_+ \rightrightarrows \R^\numgoods$, which, for any price vector $\price \in \R^\numgoods_+$ outputs a set of \mydef{excess demands} $\excessset(\price) \subseteq \R^\numgoods$, with each excess demand denoted by $\excess(\price) \in \excessset(\price)$. 
That is, $\excess[\good](\price) > 0$ denotes the number of units of commodity $\good \in \goods$ \mydef{demanded in excess} (i.e., more units of $\good$ are bought than sold), while $\excess[\good](\price) < 0$ denotes the number of units of commodity $\good \in \goods$ \mydef{supplied in excess} (i.e., more units of $\good$ are sold than bought).
%

The canonical solution concept for Walrasian economies is the Walrasian equilibrium \cite{walras}.
%
%
Given an approximation parameter $\varepsilon \geq 0$, we say that  a price vector $\price[][][*] \in \R^\numgoods_+$ is an $\varepsilon$-\mydef{Walrasian} \mydef{equilibrium} iff there exists an excess demand $\excess(\price[][][*]) \in \excessset(\price[][][*])$ s.t.\  for all commodities $\good \in \goods$, it holds that $\excess[\good](\price[][][*]) \leq \varepsilon$ \mydef{($\varepsilon$-Feasibility)} and $\price[][][*] \cdot \excess(\price[][][*]) \in [-\varepsilon, \varepsilon]$ \mydef{($\varepsilon$-Walras' law)}.
We denote the set of $\varepsilon$-Walrasian equilibria of a Walrasian economy $(\numgoods, \excessset)$ by $\we[\varepsilon](\numgoods, \excessset)$.
A $0$-Walrasian equilibrium, i.e., one that is ($0$-)feasible and satisfies ($0$-)Walras Law, is simply called a \mydef{Walrasian equilibrium}, and denoted by $\we(\numgoods, \excessset)$.

\subsection{Walrasian Economies and Variational Inequalities}

With these definitions in place, we now present the fundamental relationship between Walrasian economies and VIs.
The following theorem, due to \citet{dafermos1990exchange}, is to the best of our knowledge the first result exposing the connection between VIs and Walrasian equilibria.%
\footnote{\citet{jofre_viecon_2007} enhance this formulation by incorporating Lagrange multipliers for budget constraints, thereby facilitating equilibrium existence proofs for a broader class of economies.}
It states that the problem of computing a Walrasian equilibrium is equivalent to the problem of computing a strong solution of the Walrasian VI, whose constraint set is given by the positive orthant, an instance of the class of VIs known as \mydef{complementarity problems} \cite{cottle1968complementary}.
For completeness, we include its proof, as well as all other omitted proofs of results in this section in Appendix~\ref{sec_app:walrasian}.

\begin{restatable}[Walrasian economies as complementarity problems]{theorem}{thmweequalsvi}\label{thm:we_equal_svi}
    The set of Walrasian equilibria of any Walrasian economy $(\numgoods, \excessset)$ is equal to the set of strong solutions of the Walrasian VI $(\R^\numgoods_+, -\excessset)$, i.e., $\we(\numgoods, \excessset) = \svi(\R^\numgoods_+, -\excessset)$.
\end{restatable}

By \Cref{thm:we_equal_svi}, we can view any Walrasian equilibrium computation problem as a strong VI computation problem.
Nevertheless, because the domain of prices is unbounded (i.e., $\R^\numgoods_+$), and because we seek existence and convergence results, we restrict the class of Walrasian economies we study.
To this end, we introduce two important classes of Walrasian economies. 
The first of these is balanced economies; and the second, a further restriction, we call competitive economies.

\subsubsection{Balanced Economies}

A \mydef{balanced economy} is a Walrasian economy $(\numgoods, \excessset)$ whose excess demand correspondence satisfies \mydef{(Homogeneity of degree $0$)} for all $\lambda >0$, $\excessset(\lambda \price) = \excessset(\price)$ and \mydef{(Weak Walras' Law)} for all $\price \in \R^\numgoods_+$ and $\excess(\price) \in \excessset(\price)$, $\price \cdot \excess(\price) \leq 0$.
Intuitively, homogeneity requires that commodity prices have no absolute meaning of their own, but meaning only relative to the prices of other commodities, while weak Walras' law requires budget balance (demand expenditure does not exceed supply revenue). 
While homogeneity of degree 0 is a standard assumption, weak Walras' law is significantly weaker than standard assumptions in the literature \cite{arrow-hurwicz, debreu1974excess}.

We now provide a novel characterization of Walrasian equilibrium prices in balanced economies as the \mydef{(Walrasian) box VI}, defined over $[0, 1]^\numgoods$ rather than $\R^\numgoods_+$, which facilitates polynomial-time computation of Walrasian equilibria.
The computational guarantees of our algorithms depend on the diameter of the constraint set of the VIs, which is hereby made finite. 


\begin{restatable}[Balanced economies as VIs]{theorem}{thmwebalancedequalsvi}
\label{thm:we_balanced_equal_svi}
    For any balanced economy $(\numgoods, \excessset)$, the set of Walrasian equilibria is equal to the set of scalar multiples ($\geq 1$) of the set of strong solutions of the box VI $([0, 1]^\numgoods, -\excessset)$, i.e., $\we(\numgoods, \excessset) = \bigcup_{\lambda \geq 1} \lambda \svi([0, 1]^\numgoods, -\excessset)$.
 \klara{It is not a cone. This is the most precise alternative I could come up with.}
\end{restatable}

In the sequel, we make use of the following lemma.

\begin{restatable}[$\varepsilon$-strong solution $\implies$ $\varepsilon$-Walrasian equilibrium]{lemma}{lemmaapproxsvieqapproxwe}\label{lemma:approx_svi_eq_approx_we}
    For any balanced economy $(\numgoods, \excessset)$, any $\varepsilon$-strong solution of the box VI $([0, 1]^\numgoods, -\excessset)$ is an $\varepsilon$-Walrasian equilibrium of $(\numgoods, \excessset)$. 
\end{restatable}

We now turn our attention to the question of existence of a Walrasian equilibrium.
In balanced economies, the existence of Walrasian equilibrium follows as a corollary of the existence of strong solutions to continuous VIs (Theorem 2.2.1 of \citet{facchinei2003finite}), under the assumption that the excess demand correspondence $\excessset$ is non-empty-, compact-, convex-valued, and upper hemicontinuous on $[0,1]^\numgoods$. \sklara{}{Requiring upper hemicontinuity on $[0,1]^\numgoods$ only is motivated by the observation that in balanced economies, assuming $\excessset$ is upper hemicontinuous on $\R^\numgoods_+$ is too restrictive, because any correspondence which is homogeneous of degree $0$ and continuous on the entirety of its domain is constant.}%
\footnote{
In more stylized applications such as Arrow-Debreu competitive economies \cite{arrow-debreu}, the excess demand correspondence is defined so as to be continuous only on the interior of the unit simplex,
as the excess demand for a commodity can be infinite if its price is 0. 
{An alternative modeling choice, which also does not modify the set of Walrasian equilibria, is} to restrict the excess demand for a commodity to be bounded by the total amount of the commodity that can be ever supplied in the economy.
\citet{arrow-debreu} take exactly this approach in Section 3 of their paper when proving their seminal Walrasian equilibrium existence result, and it is also the approach we take in \Cref{lemma:ad_economies_are_comp_bounded} (Appendix~\ref{sec_app:ad_comp}) to prove that any Arrow-Debreu competitive economy can be represented as a continuous competitive economy with bounded excess demand.
}
This Walrasian equilibrium can be trivial, however, i.e., $\price[][][*] = \zeros[\numgoods]$ is an equilibrium. 
To establish the existence of a non-trivial Walrasian equilibrium, we further restrict our attention to a subset of balanced Walrasian economies. 
We choose to study a canonical subset \cite{debreu1974excess, sonnenschein1972market}, which we call \mydef{competitive economies}.


\subsubsection{Competitive Economies}

A \mydef{competitive economy} is a balanced economy $(\numgoods, \excessset)$ whose excess demand correspondence satisfies \mydef{Non-Satiation}, which we define as follows: for all $\price \in \R^\numgoods_+$ and $\excess(\price) \in \excessset(\price)$, $\excess(\price) \leq \zeros[\numgoods]$ 
implies $\price \cdot \excess(\price) = 0$. 
Intuitively, this condition requires that whenever all commodities are supplied in excess, it must be that the economy has exhausted its purchasing power. 
As such, the excess demand is non-satiated, in the sense that the economy cannot demand more of any commodity, not because its supply is insufficient, but rather because it cannot afford it.
The canonical example of a competitive economy is the \mydef{Arrow-Debreu competitive economy} \cite{arrow-debreu} (see, \Cref{lemma:ad_economies_are_comp_bounded}, Appendix~\ref{sec_app:ad_comp}).

In competitive economies, we can alternatively characterize the set of Walrasian equilibria as the Walrasian VI with constraint set $\simplex[\numgoods]$, 
yielding the \mydef{(Walrasian) simplex VI}. 
This VI is better suited to proving existence, because it is not necessary to rule out the trivial equilibrium.


\begin{restatable}[Competitive economies as VIs]{theorem}{thmwecompequalsvi}
\label{thm:we_comp_equal_svi}
     For any competitive economy $(\numgoods, \excessset)$, the set of Walrasian equilibria is equal to the strictly positive cone generated by the set of strong solutions of the simplex VI $(\simplex[\numgoods], -\excessset)$, i.e., $\we(\numgoods, \excessset) = \bigcup_{\lambda > 0} \lambda \svi(\simplex[\numgoods], -\excessset)$.
\end{restatable}

To establish the existence of a Walrasian equilibrium in competitive economies, 
we assume continuity of the excess demand, which necessitates the following definition.
A \mydef{continuous economy} is a Walrasian economy $(\numgoods, \excessset)$ whose excess demand correspondence $\excessset$ is upper hemicontinuous on $\simplex[\numgoods]$, non-empty-, compact-, and convex-valued.
Requiring upper hemicontinuity on $\simplex[\numgoods]$ only is motivated as before.
Intuitively, continuous economies are those economies in which relative changes in
prices lead to well-behaved changes in excess demands.

\klara{Assuming $\excessset$ is upper hemicontinuous on $\R^\numgoods_+$ is too restrictive, so 1) requiring upper hemicontinuity on $[0,1]^\numgoods$ only is enough for a balanced economy and 2) requiring upper hemicontinuity on $\simplex[\numgoods]$ only is enough for a competitive economy. These allow us to prove that a strong solution exists, but since we're on the unit box, this solution might be trivial. Now we move to competitive economies, again make some minimum viable assumptions; these allow us to prove that a strong solution exists, but since we're now on the simplex, any strong solution is provably non-trivial.}

Next, we leverage the fact that a strong solution is guaranteed to exist in continuous VIs (see \Cref{remark:sol_set_propert}) to establish the existence of a Walrasian equilibrium in continuous competitive economies.
We use this fact to provide an alterative proof---one that is consistent with our theory---of the existence of a Walrasian equilibrium in Arrow-Debreu competitive economies, as they are indeed continuous competitive economies (see Appendix~\ref{sec_app:ad_comp}, \Cref{lemma:ad_economies_are_comp_bounded}).


\begin{restatable}[Existence of Walrasian Equilibrium]{theorem}{thmexistencewe}\label{thm:existence_we}
    The set of non-trivial (i.e., non-zero) Walrasian equilibria of any continuous competitive economy $(\numgoods, \excessset)$ is non-empty, i.e., $\we(\numgoods, \excessset) \setminus \{ \zeros[\numgoods] \} \neq \emptyset$.
\end{restatable}



With our characterization of Walrasian equilibria as strong solutions of VIs complete, we now turn our attention to solving the Walrasian, box, and simplex VIs.

\subsection{Price Adjustment Processes for Walrasian Equilibrium}

In this section, we develop algorithms with polynomial-time 
convergence guarantees. 
Specifically, we consider (first-order) price-adjustment processes, as defined by \citet{papadimitriou2010impossibility}.
Unless otherwise noted, we assume the excess demand correspondence is singleton-valued.

Given a Walrasian economy $(\numgoods, \excess)$ and an initial iterate $\price[][0] \in \R^\numgoods_+$, a \mydef{(first-order) price-adjustment process} $\priceupdate$ consists of an update function $\priceupdate: \bigcup_{\numhorizons \geq 1} (\R^\numgoods_+ \times \R^\numgoods)^{\numhorizons} \to \R^\numgoods_+$ that generates the sequence of iterates $\left\{ \price[][\numhorizon] \right\}_{\numhorizon \in \N}$ given by $\price[][\numhorizon + 1] \doteq \priceupdate \left( \left\{ (\price[][i], \excess(\price[][i])) \right\}_{i=0}^{\numhorizon} \right)$.
As is standard in the literature (see, for instance, \citet{papadimitriou2010impossibility}), the computational complexity measures in this section consider the number of evaluations of the excess demand $\excess$ as the unit of account.

An important class of price-adjustment processes are \mydef{natural} price-adjustment processes. 
These are processes where the price of each commodity is updated using only information about the past prices of the commodity itself and its excess demand.
They are natural in the following sense: if each commodity is sold by exactly one fictional seller,%
\footnote{All commodities are assumed to be unique in some way.}
that seller can update the price of her commodity in a decentralized fashion, that is, without having to coordinate with other sellers.

Formally, given a Walrasian economy $(\numgoods, \excess)$ and an initial price vector $\price[][0] \in \R^\numgoods_+$, a price-adjustment process $\priceupdate$ is said to be \mydef{natural} if it can be written as as $\priceupdate \doteq (\priceupdate[1], \hdots, \priceupdate[\numgoods])$, where, for all commodities $\good \in \goods$, $\priceupdate[\good]:
\bigcup_{\numhorizons \geq 1} (\R_+ \times \R)^{\numhorizons} \to \R_+$ is given by, for all $\numhorizon = 0, 1, \hdots$, $\price[\good][\numhorizon + 1] \doteq \priceupdate[\good] \left( \left\{ \price[\good][k], \excess[\good] (\price[][k]) \right\}_{k = 0}^{\numhorizon} \right)$.
The canonical natural price-adjustment processes are \emph{t\^atonnement processes} \cite{walras, arrow-hurwicz}, i.e., processes $\priceupdate \doteq (\priceupdate[1], \hdots, \priceupdate[\numgoods])$ s.t.\ for all $\good \in \goods$ and $\numhorizon \in \N_{++}$, there exists a function $g: \R_+ \times \R \to \R$ that satisfies $\priceupdate[\good]\left( \left\{\price[\good][k], \excess[\good] (\price[][k])\right\}_{k = 0}^{\numhorizon} \right) \doteq g(\price[\good][\numhorizon], \excess[\good] (\price[][\numhorizon]))$.%
\footnote{Traditionally \klara{Should we cite this?} \amy{YES!!!} \klara{Might need Deni's input.}, the function $\tatonnfunc$ is further restricted to be sign preserving, i.e., for all $\price[ ] \in \R_+$ and $\excess[ ] \in \R$, $\sign(g(\price[ ], \excess[ ])) = \sign(\excess[ ])$. With this restriction in place, a t\^atonnement process can be seen a mathematical model of the law supply and demand, which stipulates that the price of any commodity that is demanded (resp.\ supplied) in excess will increase (resp.\ decrease) \cite{walras, arrow-hurwicz}.}
Observe that $\priceupdate[\good]$ depends only on commodity $j$'s price $\price[\good]$ and excess demand $\excess[\good]$, not on $\price$ and $\excess$.%
\footnote{As $\excess$ is a function of $\price$, there is an indirect coupling of commodity prices, in spite of the informational restrictions of natural price-adjustment processes.}

\subsubsection{Variationally Stable Walrasian Economies}

To characterize the class of Walrasian economies for which a Walrasian equilibrium can be computed using our VI characterizations, we now introduce the class of variationally stable economies. 
A Walrasian economy $(\numgoods, \excessset)$, is said to be \mydef{variationally stable} \cite{mertikopoulos_learning_2019} on a general price space $\pricespace \subseteq \R^\numgoods_+$ iff there exists a $\price[][][*] \in \pricespace$ s.t.\ for all prices $\price \in \pricespace$ and $\excess(\price) \in \excessset(\price)$,
$
        \innerprod[{ \excess(\price)}][{\price[][][*] - \price}] \geq 0
$.
In other words, a Walrasian economy $(\numgoods, \excessset)$ is variationally stable on $\pricespace$ iff the VI $(\pricespace, -\excessset)$ satisfies the Minty condition.%

The next lemma is key to our story: the box VI $([0, 1]^\numgoods, -\excessset)$ associated with balanced economy $(\numgoods, \excessset)$ satisfies the Minty condition, i.e., any balanced economy 
is variationally stable on $[0, 1]^\numgoods$.


\begin{restatable}[Balanced Economies are Variationally Stable on the Unit Box]{lemma}{lemmabalancedisminty}\label{lemma:balanced_is_minty}
     Any balanced economy $(\numgoods, \excessset)$ is variationally stable on $[0, 1]^\numgoods$.
     In particular, letting $\price[][][*] \doteq \zeros[\numgoods]$, for all prices $\price \in [0, 1]^\numgoods$ and $\excess(\price) \in \excessset(\price)$, it holds that $\innerprod[{ \excess(\price)}][{\price[][][*] - \price}] \geq 0$.
\end{restatable}

While \Cref{lemma:balanced_is_minty} may seem trivial, it suggests that in balanced economies, which include among others Arrow-Debreu competitive economies, first-order methods for the box VI $([0, 1]^\numgoods, -\excessset)$ are guaranteed to converge to a strong solution under suitable continuity assumptions---possibly even one that is not the zero-price vector. 

\subsubsection{Mirror Extrat\^atonnement Process in Balanced Economies}

We now turn our attention to solving the box VI $([0, 1]^\numgoods, -\excessset)$---or rather, the box VI $([0, 1]^\numgoods, -\excess)$, since we assume for our algorithms that the excess demand is singleton-valued---and hence computing a Walrasian equilibrium with the mirror extragradient method. Solving the VI $(\pricespace, -\excess)$ for a general price space $\pricespace$ with the mirror extragradient method gives rise to a family of price-adjustment processes parameterized by a kernel function $\kernel$, which we call \mydef{mirror extrat\^atonnement} (\Cref{alg:mirror_extratatonnement}).
The name is justified because, as we show next, these processes are natural.%
\footnote{A similar observation was previously made by \citet{fisher-tatonnement} for a smaller class of Walrasian economies known as convex potential markets.}

\begin{figure}
\begin{minipage}{0.67\textwidth}
\begin{algorithm}[H]
\caption{Mirror Extrat\^atonnement}
\label{alg:mirror_extratatonnement}
\begin{flushleft}
\textbf{Input:} $\numgoods, \excess, \numhorizons, \learnrate[ ][ ], \kernel, \pricespace, \price[][0]$\\
\textbf{Output:} $\{\price[][\numhorizon]\}_{\numhorizon \in [\numhorizons]}$
\end{flushleft}
\begin{algorithmic}[1]
\For {$\numhorizon = 1, \hdots, \numhorizons $}
    \State $\price[][\numhorizon + 0.5] \gets 
    \argmin\limits_{\price \in \pricespace} \left\{ \innerprod[{\excess(\price[][\numhorizon])}][{\price[][\numhorizon] - \price}] + \frac{1}{2 \learnrate[][ ]} \divergence[\kernel][\price][{\price[][\numhorizon]}]\right\}  $
    \State $\price[][\numhorizon +1] \gets 
    \argmin\limits_{\price \in \pricespace} \left\{ \innerprod[{\excess(\price[][\numhorizon + 0.5])}][{ \price[][\numhorizon] - \price}] + \frac{1}{2 \learnrate[][ ]} \divergence[\kernel][\price][{\price[][\numhorizon]}]\right\}$
\EndFor
\State \Return $\{\price[][\numhorizon + 0.5], \price[][\numhorizon + 1]\}_{\numhorizon \in [\numhorizons]}$
\end{algorithmic}
\end{algorithm}
\end{minipage}
\end{figure}

    


\begin{remark}[Mirror extrat\^atonnement is a natural price-adjustment process]
    For the price space $\pricespace \doteq [0, 1]^\numgoods$ and any choice of kernel function s.t.\ $\kernel(\price) \doteq \sum_{\good \in \goods} \kernel[\good] (\price[\good])$ for some $\{\kernel[\good]:\R^\numgoods \to \R\}_{\good \in \goods}$, the mirror extrat\^atonnement updates can be written as follows: for all commodities $\good \in \goods$ and iterations $\numhorizon \in \N$,
    \begin{align*}
        &\price[\good][\numhorizon + 0.5] \gets     \argmin\limits_{\price[\good] \in [0, 1]} \left\{ \excess[\good] (\price[][\numhorizon])(\price[\good][\numhorizon] - \price[\good]) + \frac{1}{2 \learnrate[][ ]} \divergence[{\kernel[\good]}][{\price[\good]}][{\price[\good][\numhorizon]}] \right\} \\
        &\price[\good][\numhorizon +1] \gets \argmin\limits_{\price[\good] \in [0, 1]} \left\{ \excess[\good] (\price[][\numhorizon + 0.5])(\price[\good][\numhorizon] - \price[\good]) + \frac{1}{2 \learnrate[][ ]} \divergence[{\kernel[\good]}][{\price[\good]}][{\price[\good][\numhorizon]}]\right\}
    \end{align*}
    Now, multiplying the indices of the sequence of price iterates by 2, mirror extrat\^atonnement (\Cref{alg:mirror_extratatonnement}) can be interpreted as a natural price-adjustment process which on odd time steps applies a t\^atonnement update using the current time step's prices, and on even time steps applies a t\^atonnement update using both the current and the previous time step's prices.
%
\end{remark}

With the mirror extrat\^atonnement process and \Cref{lemma:balanced_is_minty} in hand, we apply \Cref{thm:mirror_extragradient_global_convergence} to characterize the convergence of the mirror extrat\^atonnement process (\Cref{alg:mirror_extratatonnement}).  

\begin{restatable}[Convergence of Mirror Extrat\^atonnement]{theorem}{thmmirrorextratatonnconvergence}
\label{thm:mirror_extra_tatonn_convergence}
    Let $(\numgoods, \excess)$ be a balanced economy. 
    Consider the mirror extrat\^atonnement process run on $(\numgoods, \excess)$, with a $1$-strongly-convex and $\kernelsmooth$-Lipschitz-smooth kernel function $\kernel$, any time horizon $\numhorizons \ \dotin \ \N$, any step size $\learnrate[ ][ ] > 0$, a price space $\pricespace \doteq [0, 1]^\numgoods$, and any initial price vector $\price \in \interior(\pricespace)$, and let $\{\price[][\numhorizon], \price[][\numhorizon + 0.5] \}_{\numhorizon}$ be the sequence of prices generated. Suppose that there exists $\lsmooth \in (0, \frac{1}{\sqrt{2}\learnrate[ ][ ]}]$, s.t.\ $\frac{1}{2}\|\excess(\price[][k+0.5]) - \excess(\price[][k])\|^2 \leq \lsmooth^2 \divergence[\kernel][{\price[][k+0.5]}][{\price[][k]}]$ \samy{}{for all $k \in [\numhorizons]$}. 
    Then, the following bound holds:
    \begin{align}
    \label{eq:weq_exists}
     \min_{k = 0, \hdots, \numhorizons} \max_{\price \in \pricespace} \ \langle \excess(\price[][][k + 0.5]), \price[][][k + 0.5] - \price \rangle \leq \, \frac{2 (1 + \kernelsmooth) \numgoods}{\learnrate[ ][ ]} \frac{\sqrt{\divergence[\kernel][{\zeros[\numgoods]}][{\price[][][0]}]}}{\sqrt{\numhorizons}} \enspace.
    \end{align}
    
    In addition, for all $\varepsilon > 0$, if there exist time horizon $\altnumhorizons[\varepsilon] \gg O(\nicefrac{1}{\varepsilon^2})$ and step size $\learnrate[\varepsilon ][ ] > 0$ with corresponding $\lsmooth[\varepsilon] \in (0, \frac{1}{\sqrt{2} \learnrate[\varepsilon ][ ]}]$ s.t.\ $\frac{1}{2} \norm[\excess({\price[][][k+0.5]}) - \excess({\price[][][k])}]^2 \leq \lsmooth[\varepsilon]^2 \divergence[\kernel][{\price[][][k+0.5]}][{\price[][][k]}]$ for all $k \in [\altnumhorizons[\varepsilon]]$, then there exists a choice of time horizon $\numhorizons[\varepsilon] \ \dotin \ O \left( \nicefrac{\kernelsmooth^2\numgoods^2\divergence[\kernel][{\zeros[\numgoods]}][{\price[][][0]}]}{\learnrate[\varepsilon ][ ]^2\varepsilon^2} \right)$
\end{restatable}

The convergence guarantee provided by \Cref{thm:mirror_extra_tatonn_convergence} 
requires pathwise Bregman continuity over price trajectories generated by the mirror extrat\^atonnement process.
This statement is purposeful, as it is not possible to guarantee Bregman continuity of an excess demand function that is homogeneous of degree $\alpha > 0$,
unless this excess demand is constant and thus trivially Bregman continuous for any Bregman divergence $\divergence[\kernel]$.%
\footnote{\label{footnote:lipschitz_box} 
Suppose that $\excess$ is $\lsmooth$-Lipschitz-continuous on $[0, 1]^\numgoods$. By homogeneity of degree $0$, we have, for all $\alpha > 0$ and $\price, \otherprice \in [0, 1]^\numgoods$, $\norm[{\excess(\price) - \excess(\otherprice)}] = \norm[{\excess(\alpha\price) - \excess(\alpha \otherprice)}] \leq \lambda \alpha \norm[{\price - \otherprice}]$. Hence, taking $\alpha \to 0$, we have, for all  $\price, \otherprice \in [0, 1]^\numgoods$, $\excess(\otherprice) = \excess(\price)$.}
Moreover, pathwise Bregman continuity is achievable by mirror extrat\^atonnement, as we empirically verify for all Walrasian economies simulated in the experiments reported in \Cref{sec:experiments}.


\klara{Reviewer \#1377C questions the relevance of pathwise, seemingly critiquing our exploratory, non-deductivist flow. Can we respond to this anyhow?}

For choices of kernel functions $\kernel$ s.t.\ the associated Bregman divergence $\divergence[\kernel]$ is \emph{not\/} homogeneous of degree $\alpha > 0$ (i.e., for all $\price, \otherprice \in \R^\numgoods_+$ and $\alpha, \lambda > 0$, $\divergence[\kernel] (\lambda \price, \lambda\otherprice) \neq \lambda^\alpha \divergence[\kernel] (\price, \otherprice)$), we define a novel class of economies, which seems likely to capture a broader family of Walrasian economies.%
\footnote{We leave this question open for future work.}
Given a modulus of continuity $\lsmooth > 0$ and a kernel function $\kernel: \pricespace \to \R$, a \mydef{$(\lsmooth, \kernel)$-Bregman-continuous economy} on $\pricespace \subseteq \R^\numgoods_+$ is a Walrasian economy $(\numgoods, \excess)$ whose excess demand $\excess$ is $(\lsmooth, \kernel)$-Bregman-continuous on $\pricespace$.%
\footnote{Bregman-continuous functions have been introduced in recent years in the optimization literature and have been shown to contain a large number of important function classes which are not continuous (see, for instance, \citet{lu2019relative}). Note that when the kernel function $\kernel$ is chosen to be $\kernel(\price) \doteq \frac{1}{2}\|\price\|^2$, $\lsmooth$-Bregman continuity reduces to $\lsmooth$-Lipschitz continuity. Further, the literature on algorithmic general equilibrium theory has considered variants of Bregman continuity to prove the polynomial-time convergence of the mirror t\^atonnement process to Walrasian equilibria in restricted classes of Walrasian economies (see, for instance \citet{fisher-tatonnement} and \citet{cheung2018dynamics}). As such, Bregman continuity seems a natural assumption to prove the convergence of algorithms to a Walrasian equilibrium.} 
We now note the following corollary of \Cref{thm:mirror_extra_tatonn_convergence}.

\begin{corollary}[Convergence of mirror extrat\^atonnement under Bregman continuity]
\label{thm:bregman_mirror_exta_tatonn_convergence}
    Let $\kernel$ be a $1$-strongly-convex and $\kernelsmooth$-Lipschitz-smooth kernel function and let $(\numgoods, \excess)$ be a balanced economy that is $(\lsmooth, \kernel)$-Bregman-continuous on $[0, 1]^\numgoods$.
    Imagine running the mirror extrat\^atonnement process on $(\numgoods, \excess)$, with kernel function $\kernel$, time horizon $\numhorizons \in \N$, price space $\pricespace \doteq [0, 1]^\numgoods$, initial price vector $\price[][0] \in \interior (\pricespace)$, and step size $\learnrate[ ][ ] \in (0, \frac{1}{\sqrt{2}\lsmooth}]$.
    Let $\left\{ \price[][\numhorizon], \price[][\numhorizon + 0.5] \right\}_{\numhorizon \in [\numhorizons]}$ be the sequence of prices generated.
    Then, for all $\varepsilon > 0$, there exists a corresponding choice of time horizon 
    $\numhorizons \ \dotin \ O \left( \frac{\kernelsmooth^2 \numgoods^2 \divergence[\kernel] (\zeros[\numgoods], \price[][0])}{\learnrate[ ][ ]^2 \varepsilon^2} \right)$
    s.t.\ $\bestiter[{\price}][\numhorizons] \ \dotin \ \argmin_{\price[][][(k+0.5)] : k = 0, \hdots, \numhorizons} \divergence[\kernel] (\price[][k+0.5], \price[][k])$ is an $\varepsilon$-Walrasian equilibrium of $(\numgoods, \excess)$.
    Furthermore, $\price[][][*] \doteq \lim_{\numhorizon \to \infty} \price[][\numhorizon+0.5] = \lim_{\numhorizon \to \infty} \price[][\numhorizon]$ is a Walrasian equilibrium of $(\numgoods, \excess)$.
\end{corollary}

\subsubsection{Mirror Extrat\^atonnement Process in Elastic Economies}

While the convergence result established in \Cref{thm:bregman_mirror_exta_tatonn_convergence} seems promising, it is not immediately clear what types of excess demand functions satisfy Bregman continuity.
To characterize the Bregman continuity properties of Walrasian economies, we introduce additional economic parameters, which have been used extensively in the analysis of algorithms for computing Walrasian equilibria (e.g., \citet{cole2008fast}).


Given $\lelastic \geq 0$, an economy $(\numgoods, \demandfunc, \supplyfunc)$ is said to be \mydef{$\lelastic$-elastic on $\pricespace \subseteq \R^\numgoods_+ $}
if it is a Walrasian economy $(\numgoods, \excess)$ with an \mydef{aggregate demand function} $\demandfunc: \R^\numgoods_+ \to \R^\numgoods_+$ and \mydef{aggregate supply function} $\supplyfunc: \R^\numgoods_+ \to \R^\numgoods_+$ s.t.\ $\excess(\price) \doteq \demandfunc(\price) - \supplyfunc(\price)$, for which the following two bounds hold:
\begin{align*}
        &\frac{\norm[\demandfunc(\otherprice)-\demandfunc(\price)]}{\norm[\demandfunc(\price)]} \leq \lelastic \, \frac{\norm[\otherprice-\price]}{\norm[\price]_\infty}
        \quad \quad \quad
        \frac{\norm[\supplyfunc(\otherprice)-\supplyfunc(\price)]}{\norm[\supplyfunc(\price)]} \leq \lelastic \, \frac{\norm[\otherprice-\price]}{\norm[\price]_\infty}.
\end{align*}

The next lemma provides conditions that ensure that an $\lelastic$-economy satisfies Bregman continuity.

\begin{restatable}[Bregman continuity bound for elastic economies]{lemma}{lemmabregmancontelastic}\label{lemma:bregman_cont_elastic}
    For $(\numgoods, \demandfunc, \supplyfunc)$ an $\lelastic$-elastic economy on $\pricespace \subseteq \R^\numgoods_+\setminus \{ \zeros[\numgoods] \}$ and for any $1$-strongly-convex kernel function $\kernel: \R^\numgoods_+ \to \R$, 
    it holds that:
    \begin{align*}
        \nicefrac{1}{2} \norm[{\excess(\otherprice) - \excess(\price)}]^2 \leq \left(\frac{\lelastic \left(\norm[{\demandfunc(\price)}] + \norm[{\supplyfunc(\price)}] \right)}{\norm[{\price}]_\infty} \right)^2 \divergence[\kernel] (\otherprice, \price) \quad \forall \price, \otherprice \in \pricespace.
    \end{align*}
\end{restatable}


\Cref{lemma:bregman_cont_elastic} suggests that upper bounding the excess demand and lower bounding prices away from $\zeros[\numgoods]$ is sufficient to ensure the $(\lsmooth, \kernel)$-Bregman continuity of the excess demand, assuming a non-zero and finite Bregman constant
$\lsmooth$.


While it is not possible to ensure that prices remain bounded away from $\zeros[\numgoods]$ when running the mirror extrat\^atonnement process with price space $[0, 1]^\numgoods$, we can instead choose the price space $\simplex[\numgoods]$, in which case we can obtain a Bregman continuity bound from \Cref{lemma:bregman_cont_elastic}.
%
%
The price space $\simplex[\numgoods]$, however, does not include the zero vector $\zeros[\numgoods]$, which trivially ensures that balanced economies are variationally stable on $[0, 1]^\numgoods$.
In this sense, the restriction of the price space to $\simplex[\numgoods]$ can effectively destabilize an economy, potentially making the computation of a Walrasian equilibrium intractable. 
To overcome this challenge, we focus our attention on the class of competitive economies that are indeed variationally stable on $\simplex[\numgoods]$. 
This class includes among others Walrasian economies whose excess demand satisfies the weak axiom of revealed preferences (WARP) (\Cref{def:warp}, Appendix~\ref{sec_app:var_stable_classes}), the weak gross substitutes (WGS) condition (\Cref{def:wgs}, Appendix~\ref{sec_app:var_stable_classes}), and the law of supply and demand (\Cref{def:law_of_supply_and_demand}, Appendix~\ref{sec_app:var_stable_classes}).

To use \Cref{lemma:bregman_cont_elastic}, we also have to ensure that the excess demand of the economy is bounded, which we achieve by bounding each consumer's consumption space for a commodity by the maximum aggregate supply of that commodity. 
Formally,
given $\lbounded \geq 0$, a \mydef{$\lbounded$-bounded economy} $(\numgoods, \demandfunc, \supplyfunc)$ is a Walrasian economy $(\numgoods, \excess)$ that consists of an \mydef{aggregate demand function} $\demandfunc: \R^\numgoods_+ \to \R^\numgoods_+$ and an \mydef{aggregate supply function} $\supplyfunc: \R^\numgoods_+ \to \R^\numgoods_+$ s.t.\ $\excess(\price) \doteq \demandfunc(\price) - \supplyfunc(\price)$,
$\|\demandfunc\|_\infty \leq \lbounded$, and
$\|\supplyfunc\|_\infty \leq \lbounded$. 
In \Cref{lemma:ad_economies_are_comp_bounded} (Appendix~\ref{sec_app:ad_comp}), we prove that any Arrow-Debreu competitive economy \cite{arrow-debreu} can be represented as a bounded continuous competitive economy; as such, this assumption is mild.

With these definitions in place, we can finally apply \Cref{lemma:bregman_cont_elastic} to derive polynomial-time convergence of the mirror extrat\^atonnement process as a reinterpretation of \Cref{thm:mirror_extra_tatonn_convergence}.

Going beyond variationally stable competitive economies on $\simplex[\numgoods]$, the local convergence behavior of mirror extrat\^atonnement can similarly be obtained by applying \Cref{thm:vi_mirror_extragrad_local}, i.e., replacing the assumption that the competitive economy is variationally stable with the assumption that the initial price iterate starts close enough to a 
local weak solution of the simplex VI $(\simplex[\numgoods], -\excess)$.

\begin{restatable}[{Mirror Extrat\^atonnement convergence on the unit simplex}]{theorem}{thmmirrorextratatonnvarstable}
\label{thm:mirror_extratatonn_var_stable}
    Let $(\numgoods, \demandfunc, \supplyfunc)$ be a $\lbounded$-bounded balanced economy that is $\lelastic$-elastic and variationally stable on $\simplex[\numgoods]$, and let $\price[][][*] \in \simplex[\numgoods]$ be a price vector that makes it such. 
    In addition, let $\kernel$ be a $1$-strongly-convex and $\kernelsmooth$-Lipschitz-smooth kernel function.
    Imagine running the mirror extrat\^atonnement process on $(\numgoods, \demandfunc, \supplyfunc)$, with kernel function $\kernel$, step size $\learnrate[ ][ ] \in (0, \frac{1}{2\sqrt{2}\numgoods \lelastic \lbounded}]$, time horizon $\numhorizon \in \N$, price space $\pricespace \doteq \simplex[\numgoods]$, and initial price vector $\price[][0] \in \interior (\pricespace)$,
    Let $\left\{ \price[][\numhorizon], \price[][\numhorizon + 0.5] \right\}_{\numhorizon}$ be the sequence of prices generated. 
    Then, the following bound holds:
    \begin{align*}
        \min_{k = 0, \hdots, \numhorizons} \max_{\price \in \simplex} \langle \excess(\price[][k+0.5]),  \price - \price[][k+0.5] \rangle \leq  \frac{2 \sqrt{2}(1 + \kernelsmooth)}{\learnrate[ ][ ]} \frac{\sqrt{\divergence[\kernel][{\price[][][*]}][{\price[][0]}]}}{\sqrt{\numhorizons}}.
    \end{align*}
    Furthermore,
    $\price[][][**] \doteq \lim_{\numhorizon \to \infty} \price[][\numhorizon+0.5] = \lim_{\numhorizon \to \infty} \price[][\numhorizon]$
    is a Walrasian equilibrium of $(\numgoods, \demandfunc, \supplyfunc)$.
\end{restatable}

\subsubsection{Mirror Extrat\^atonnement in Scarf Economies}

One of the earliest negative 
results in the literature is the example of a Walrasian economy provided by Herbert Scarf, in which continuous-time t\^atonnement cycles around the unique Walrasian equilibrium, while discrete-time variants spiral away from the equilibrium from any initial non-equilibrium price vector \cite{scarf1960instable}---this, all when the algorithm is run on the simplex. 
Formally, a \mydef{Scarf economy} is a Walrasian economy $(3, \excess^{\mathrm{scarf}})$ with three goods, for which the excess demand is the following singleton-valued function:
    \begin{align*}
        \excess^{\mathrm{scarf}}(\price) \doteq \left(
            \frac{\price[1]}{\price[1] + \price[2]} + \frac{\price[3]}{\price[1] + \price[3]} - 1,
            \frac{\price[1]}{\price[1] + \price[2]} + \frac{\price[2]}{\price[2] + \price[3]} - 1,
            \frac{\price[2]}{\price[2] + \price[3]} + \frac{\price[3]}{\price[1] + \price[3]} - 1
            \right) \enspace.
    \end{align*}

As we show in \Cref{lemma:scarf_var_stable_breg_cont} (Appendix~\ref{sec_app:walrasian}), the Scarf economy is 
Bregman-continuous on any price space bounded away from $0$ in every coordinate. 
We thus obtain the following 
result.

\klara{Is the Scarf economy operator integrable? More generally, is the Leontief-consumer economy operator integrable? Is there any informal notion of ``duality'' between the Walrasian and Arrow-Debreu VIs, where the former breaks as $\rho \to 1$ and the latter breaks as $\rho \to -\infty$?}
%

\klara{I would LOVE to use Scarf's Bregmanness on $\left[ \underline{p},1 \right]^3$ and, arguing that the path stays on this sub-box, get pathwise Bregman on $[0,1]^3$ (with zero!) \emph{ex ante} rather than \emph{ex post}.}
    
    
\begin{restatable}[Convergence of Mirror extrat\^atonnement in Scarf Economies]{corollary}{scarfconvergence}
\label{thm:scarf_convergence}
    Let $\kernel$ be a $1$-strongly-convex and $\kernelsmooth$-Lipschitz-smooth kernel function.
    Imagine running the mirror extrat\^atonnement process on the Scarf economy $\left(3, \excess^{\mathrm{scarf}}\right)$ on price space $\pricespace \doteq [0, 1]^3$ with kernel function $\kernel$ and initial price vector $\price[][0] \in \interior(\pricespace)$.
    Let $\left\{ \price[][\numhorizon], \price[][\numhorizon + 0.5] \right\}_{\numhorizon \in [\numhorizons]}$ \amy{why is this index $t$ and the next one $k$?} be the sequence of prices generated after $\numhorizons \in \N$ time steps, and assume there exists a step size $\learnrate[ ][ ] > 0$ and corresponding $\lsmooth \in (0, \frac{1}{\sqrt{2}\learnrate[ ][ ]}]$ s.t.\ $\frac{1}{2} \norm[{\excess^{\mathrm{scarf}}(\price[][k+0.5]) - \excess^{\mathrm{scarf}}(\price[][k])}]^2 \leq \lsmooth^2 \divergence[\kernel][{\price[][k+0.5]}][{\price[][k]}]$ for all $k \in [\numhorizons]$.
    
    In addition, for all $\varepsilon > 0$, if there exist time horizon $\altnumhorizons[\varepsilon] \gg O(\nicefrac{1}{\varepsilon^2})$ and step size $\learnrate[\varepsilon ][ ] > 0$ with corresponding $\lsmooth[\varepsilon] \in (0, \frac{1}{\sqrt{2}\learnrate[\varepsilon ][ ]}]$ s.t.\ $\frac{1}{2} \norm[\excess({\price[][][k+0.5]}) - \excess({\price[][][k])}]^2 \leq \lsmooth[\varepsilon]^2 \divergence[\kernel][{\price[][][k+0.5]}][{\price[][][k]}]$ for all $k \in [\altnumhorizons[\varepsilon]]$, then there exists a choice of time horizon $\numhorizons[\varepsilon] \ \dotin \ O \left( \nicefrac{\kernelsmooth^2\numgoods^2\divergence[\kernel][{\zeros[\numgoods]}][{\price[][][0]}]}{\learnrate[\varepsilon ][ ]^2\varepsilon^2} \right)\leq \altnumhorizons[\varepsilon]$ s.t.\ $\bestiter[\price][{\numhorizons[\varepsilon]}] \in \argmin_{\price[][][k+0.5] : k = 0, \hdots, \numhorizons[\varepsilon]} \divergence[\kernel] (\price[][][k+0.5], \price[][][k])$ is an $\varepsilon$-Walrasian equilibrium of $(3, \excess^{\mathrm{scarf}})$.
\end{restatable}

\amy{should we just say something here about how we have to initialize away from $\zeros$ for Bregman to hold. and then when we do, we seem to get convergence, even when $\eta$ and $\lambda$ are not in exactly this relationship.}

\section{Experiments}
\label{sec:experiments}
We describe two sets of experiments in this section.
First, we run 
mirror t\^atonnement and extrat\^atonnement, both with kernel function $\kernel(\price) \doteq \nicefrac{1}{2} \norm[\price]^2$ in an attempt to solve the Scarf economy.
The goal is to illustrate the differing convergence behavior between these two price-adjustment processes. 
Second, we run mirror extrat\^atonnement, again with kernel function $\kernel(\price) \doteq \nicefrac{1}{2} \norm[\price]^2$, on various Arrow-Debreu exchange economies \cite{arrow-debreu}.
This time, the goal is to demonstrate that 
extrat\^atonnement can efficiently solve reasonably large Walrasian economies in practice. 





\amy{are the results robust? do they hold up for various choices of step size?}

In \Cref{fig:scarf_phase_portraits}, we visualize the trajectories generated by the t\^atonnement and extrat\^atonnement processes on the Scarf economy.
As is well established {\cite{gillen2020divergence}, the sequence of prices generated by t\^atonnement, even when initialized very close to the equilibrium prices $(\nicefrac{1}{3}, \nicefrac{1}{3}, \nicefrac{1}{3})$ spirals away from these prices and ultimately lands at $(0, 0, 1)$, which is \emph{not\/} a Walrasian equilibrium.%
\footnote{Prices do not converge here in the technical sense.  Rather, further updates attempt to move prices outside the simplex, after which they are repeatedly projected back to this point.}
In contrast, the prices generated by the mirror extrat\^atonnement process spiral inwards towards the equilibrium prices, even when initialized far away from them.
An informal rationale for this behavior is as follows: 
The continuous-time variant of t\^atonnement is known to cycle around the equilibrium prices \cite{scarf1960instable}.
If we interpret the discrete-time t\^atonnement (resp. mirror extrat\^atonnement) process as an explicit (resp. implicit) discretization \cite{butcher2008numerical} of the continuous-time t\^atonnement dynamics, it is well understood that explicit (resp.\ implicit) discretization methods tend to be unstable (resp.\ stable) when continuous-time dynamics cycle.

\begin{figure}[h!]
    \centering
    
    \begin{subfigure}{0.46\textwidth}
        \centering
        \includegraphics[width=\linewidth, height=60mm, keepaspectratio]{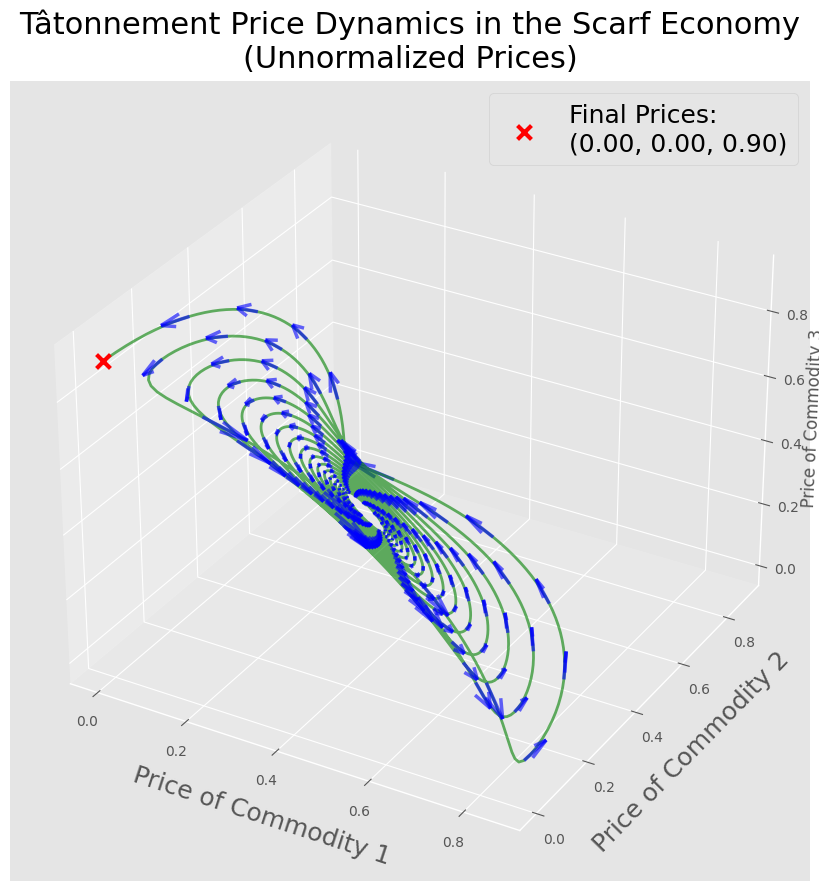}
    \end{subfigure}
    \hfill
    \begin{subfigure}{0.46\textwidth}
        \centering
        \includegraphics[width=\linewidth, height=60mm, keepaspectratio]{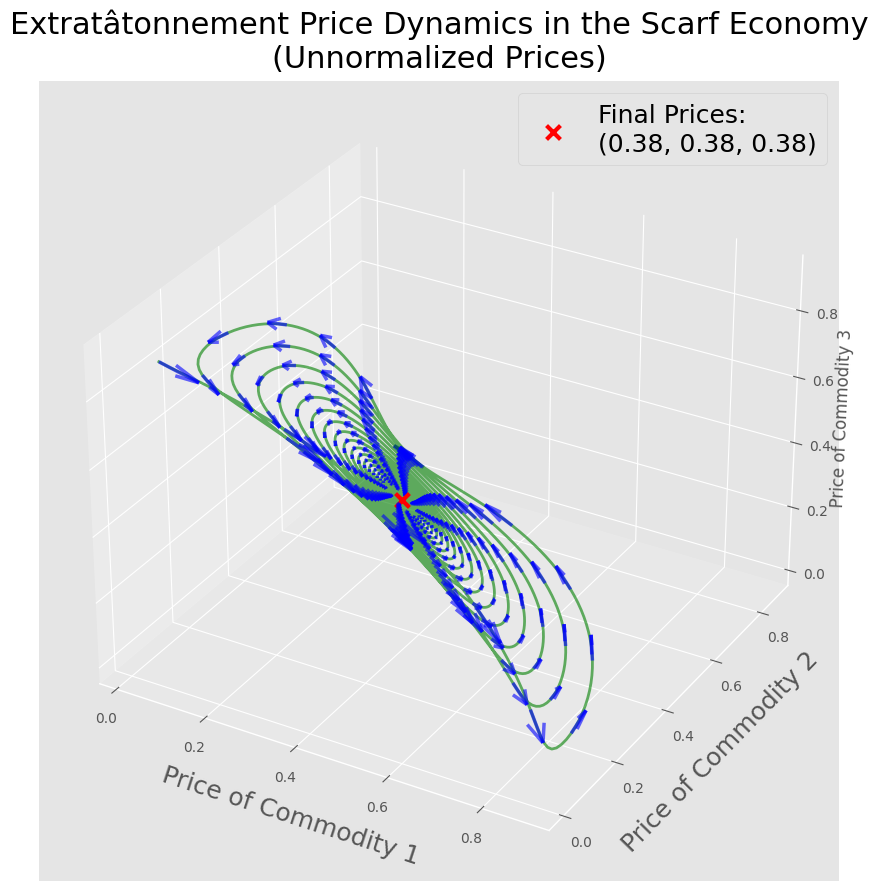}
    \end{subfigure}

    \caption{Phase portraits of t\^atonnement and extrat\^atonnement for the Scarf economy. All experiments were run on the unit box with step size $\eta = \nicefrac{7}{180}$. (Unnormalized) price plots show the price trajectories generated by t\^atonnement and extrat\^atonnement in the Scarf economy. Price plots post-hoc normalized to lie on the unit simplex appear in \Cref{sec_ap:experiments}. They are more triangular, but otherwise not qualitatively different.}
    \label{fig:scarf_phase_portraits}
\end{figure}

   \begin{table}[H]
   \caption{Arrow-Debreu exchange economy experimental setups}\label{table:exp_summary}
   \begin{center}
    \begin{tabular}{p{0.5cm} | p{1cm} | p{1.5cm} p{1.5cm} p{1.5cm} p{1.5cm} p{1.5cm}}
    \hline
    \makecell[l]{Exp \\ No.} & 
    \makecell[l]{Num. \\ Comm.} & 
    \makecell[l]{Num. \\ Linear} & 
    \makecell[l]{Num. \\ Cobb-\\ Douglas} & 
    \makecell[l]{Num. \\  CES  \\ $\rho \! \in \! \!(\!0, 1\!)$} & 
    \makecell[l]{Num. \\  CES \\ $\rho < 0$} & 
    \makecell[l]{Num. \\ Leontief} \\ \hline \hline
    1 & 500 & 0 & 0 & 0 & 0 & 600 \\
    2 & 500 & 0 & 0 & 0 & 600 & 0\\
    3 & 500 & 0 & 0 & 600 & 0 & 0\\
    4 & 500 & 0 & 600 & 0 & 0 & 0\\
    5 & 500 & 600 & 0 & 0 & 0 & 0\\
    6 & 1000 & 200 & 200 & 200 & 200 & 200\\
    7 & 1000 & 0 & 200 & 200 & 200 & 200\\ \hline
    \end{tabular}
   \end{center}
\end{table}

An Arrow-Debreu exchange economy $(\numbuyers, \numcommods, \consumptions, \consendow, \util)$ comprises $\numcommods \in \N$ commodities and $\numbuyers \in \N$ consumers, each $\consumer \in [\numconsumers]$ with a consumption space $\consumptions[\buyer]$, an endowment of commodities  $\consendow[\consumer] \in \R^\numcommods_+$, and a utility function $\util[\consumer]: \consumptions[\buyer] \to \R$. 
Such an economy can be represented as a bounded continuous competitive economy $(\numgoods, \excessset)$ where the excess demand correspondence is given by $\excessset(\price) \doteq \sum_{\player \in \players} \argmax\limits_{\consumption[\consumer] \in \consumptions[\consumer]: \consumption[\consumer] \cdot \price \leq \consendow[\consumer] \cdot \price}  \util[\consumer](\consumption[\consumer]) - \sum_{\consumer \in \consumers} \consendow[\consumer]$.%
\footnote{We refer the reader to Appendix~\ref{sec_app:ad_comp} on additional background and definitions on Arrow-Debreu exchange economies.} 

In our second set of experiments, we run extrat\^atonnement on Arrow-Debreu exchange economies, using the following utility function classes:
1.~linear: $\util[\buyer](\allocation[\buyer]) = \sum_{\good \in \goods} \valuation[\buyer][\good] \allocation[\buyer][\good]$; 
2.~Cobb-Douglas:  $\util[\buyer](\allocation[\buyer]) = \prod_{\good \in \goods} \allocation[\buyer][\good][][{\valuation[\buyer][\good]}]$; 
3.~Leontief:  $\util[\buyer](\allocation[\buyer]) = \min_{\good \in \goods} \left\{ \nicefrac{\allocation[\buyer][\good]}{\valuation[\buyer][\good]}\right\}$; 
and 4.~CES: $\util[\buyer](\allocation[\buyer]) = \sqrt[{\rho_\buyer}]{ \sum_{\good \in \goods} \valuation[\buyer][\good] \allocation[\buyer][\good][][{\rho_\buyer}]}$,
with each utility function parameterized by $\valuation[\buyer] \in \mathbb{R}_+^{\numbuyers}$, where every $\valuation[\buyer][\good]$ quantifies the value of commodity $\good$ to consumer $\consumer$.
These values are drawn randomly from a uniform distribution.

We study seven different reasonably-sized Arrow-Debreu exchange economies, with 500 or 1000 commodities, and 600, 800, or 1000 consumers.
Their precise makeup is shown in \Cref{table:exp_summary}.%
\footnote{For reproducibility, our code is available and ready to run on \coderepo.
In addition, we include all details of our experimental setup in Appendix~\ref{sec_ap:experiments}.}%
$^,$\footnote{Among these utility functions, linear and Leontief utilities could cause the demand of a consumer of one or more commodities to approach $\infty$ when the prices of these commodities approach $0$.
We thus bound the per-consumer per-commodity demand to be between 0 and 
the maximum supply of each commodity.
Bounding the consumer demand in this manner does not modify the set of Walrasian equillibria \cite{arrow-debreu}.}

\klara{The more linear the consumers get ($\rho \to 1$), the larger the coefficient $\lambda$ for the Bregman continuity of the excess demand (the correspondence remains upper hemicontinuous even for $\rho=1$). This affects the entire price space in a ``cobweb'' fashion, and a successful descent algorithm must step over these ``Bregman-discontinuities'' rather than circumvent them, which is what our story seems to suggest.} 

\begin{remark}
Even for very small $\epsilon > 0$, the approximation of linear utilities by CES utilities with $\rho=1-\epsilon$ is continuous, and is thus covered by our theory.
Nonetheless, the corresponding Bregman continuity coefficient $\lsmooth$ may become too large for the bound guaranteed by \Cref{thm:mirror_extratatonn_var_stable} to be meaningful. 
Even in the pathwise case of \Cref{thm:mirror_extra_tatonn_convergence}, where $\lsmooth$ remains finite in the limit as $\rho\to 1$ because the step size $\learnrate[ ][ ]$ is bounded away from zero, $\lsmooth$ may still become orders of magnitude larger for linear utilities than for other utility classes. 
See \Cref{rmk:linear_practice} in Appendix~\ref{sec_app:ad_comp} for an example. \klara{Is the placement okay like this?}
\end{remark}

\Cref{fig:experiment_plots_alt} depicts the convergence of extrat\^atonnement in all seven economies.
Moreover, the convergence rates (loosely; explanation forthcoming) respect our bounds (see \Cref{table:convergence_bounds_alt}).
Polynomial-time convergence is perhaps surprising, in light of \citeauthor{papadimitriou2010impossibility}'s observation that solving Arrow-Debreu economies is PPAD-complete (\citeyear{papadimitriou2010impossibility}).
But we are not contradicting their claim.
For one thing, the space of prices we consider is the unit box, not the unit simplex as in \citeauthor{papadimitriou2010impossibility}; thus, Minty is necessarily satisfied.
Additionally, we achieve our results by searching for a suitable step size, 
in consideration of pathwise Bregman continuity.

\klara{Convergence requires 1) Bregman continuity and 2) Minty. Both of these assumptions narrow down the class of solvable AD economies. Bregman continuity can be obtained a) pathwise ex-post through a step-size search or b) globally by proving Bregman continuity directly or c) globally by showing bounded excess demand and bounded elasticity of excess demand (these jointly imply global Bregman continuity). Minty can be obtained a) trivially for balanced economies by searching on the unit box or b) by (possibly more elaborate proof) for any smaller class of economies on the unit simplex. Papadimitriou and Yannakakis operate on the unit simplex. AD economies with $\rho<0$ do not necessarily satisfy Minty on the simplex, \sklara{though Scarf does}{including Scarf}. \amy{WRONG!?} AD economies are not necessarily globally Bregman-continuous on the unit simplex, though we can artificially bound excess demand and elasticity to get Bregman-continuity. Hardness likely lies in the class of AD economies where (NOT Bregman) OR (NOT Minty)}

\begin{figure}[htbp]
    \centering
    \begin{subfigure}{0.3\textwidth}
        \centering
        \includegraphics[height=30mm, width=45mm]{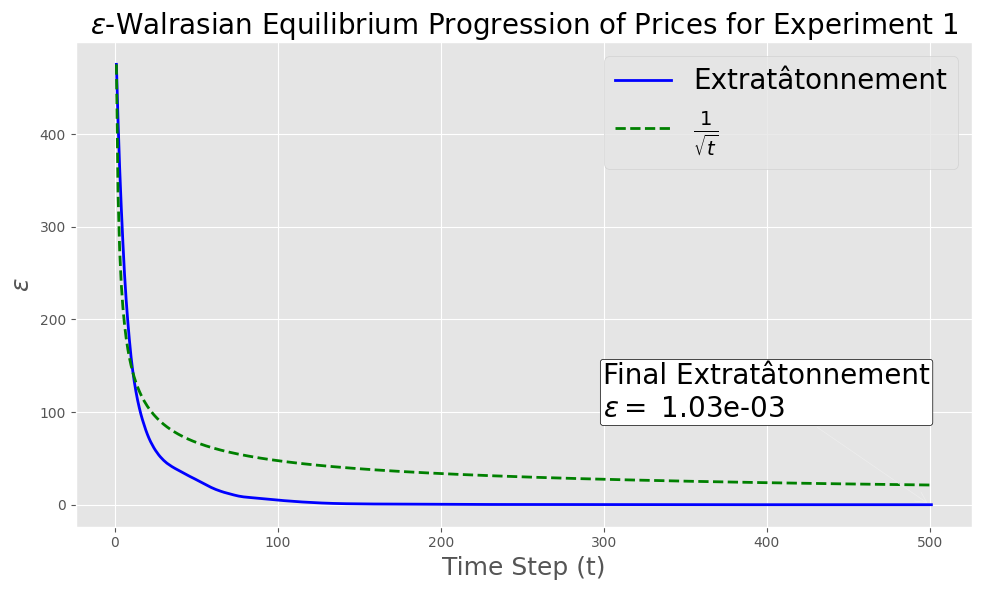}
    \end{subfigure}
    \hfill
    \begin{subfigure}{0.3\textwidth}
        \centering
        \includegraphics[height=30mm, width=45mm]{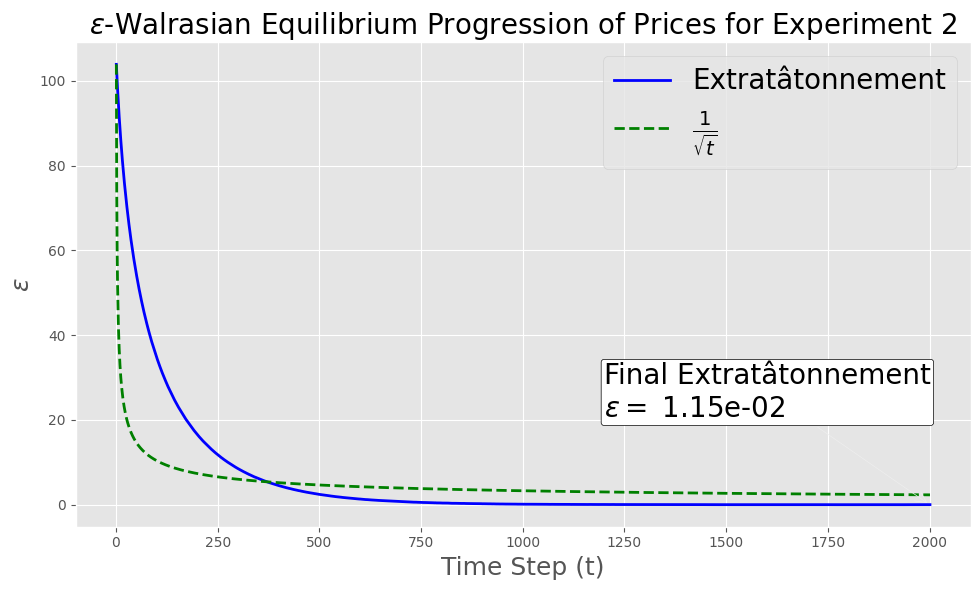}
    \end{subfigure}
    \hfill
    \begin{subfigure}{0.3\textwidth}
        \centering
        \includegraphics[height=30mm, width=45mm]{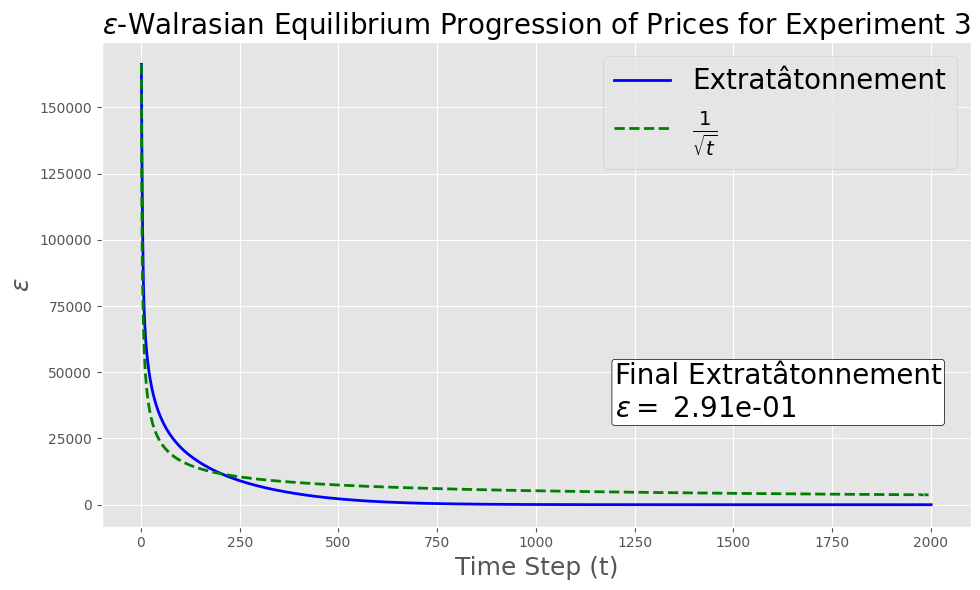}
    \end{subfigure}

    \vspace{0.3cm}

    \begin{subfigure}{0.3\textwidth}
        \centering
        \includegraphics[height=30mm, width=45mm]{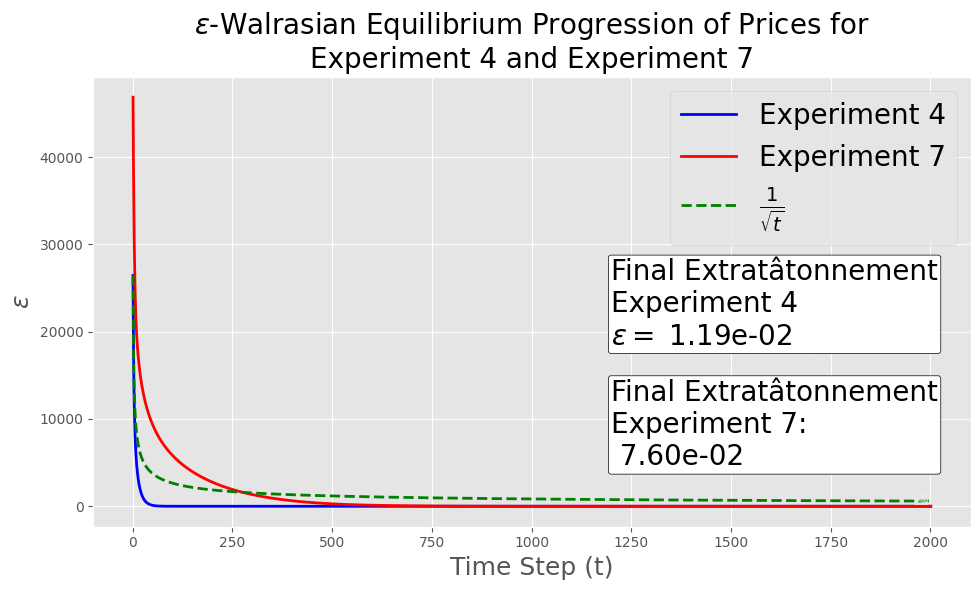}
    \end{subfigure}
    \hfill
    \begin{subfigure}{0.3\textwidth}
        \centering
        \includegraphics[height=30mm, width=45mm]{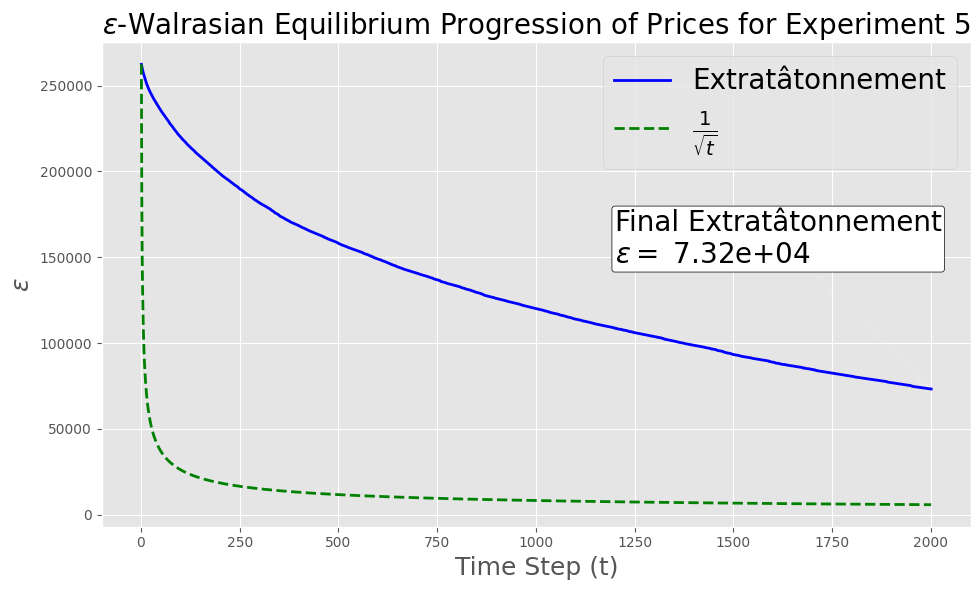}
    \end{subfigure}
    \hfill
    \begin{subfigure}{0.3\textwidth}
        \centering
        \includegraphics[height=30mm, width=45mm]{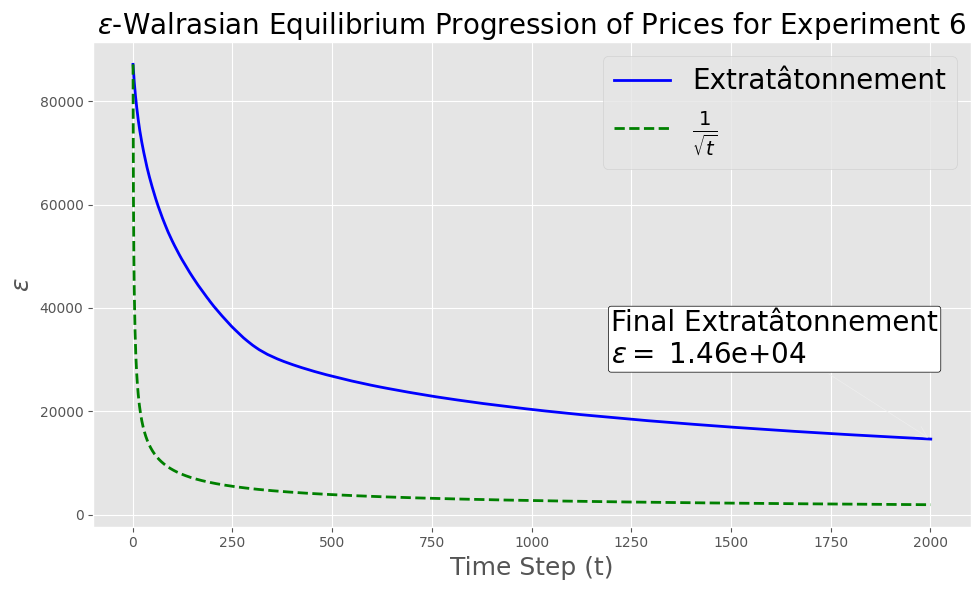}
    \end{subfigure}

    \caption{Per-iteration $\varepsilon$-Walrasian deviations in economies 1--7 using step sizes that correspond to the minimum best-iterate Walrasian deviation such that pathwise Bregman continuity is satisfied. Observe that $\varepsilon$ follows a $\nicefrac{1}{\sqrt{t}}$ trajectory in all economies except 5 and 6, in line with our theoretical results. Convergence is slower in 5 and 6, as these markets are inhabited by linear consumers, for which the convergence bounds guaranteed by our theory are much weaker due to the magnitude of the pathwise Bregman continuity coefficient.}
    \label{fig:experiment_plots_alt}
\end{figure}

\begin{table}[hbtp]
   \caption{Loose upper bounds (based on \Cref{thm:mirror_extragradient_global_convergence}) and observed Walrasian deviations
   in economies 1--7 using step sizes that minimize best-iterate Walrasian deviations while satisfying pathwise Bregman continuity. ``Minimum deviation'' corresponds to this minimum best-iterate Walrasian deviation, whereas ``final deviation'' corresponds to the Walrasian deviation of the last iterate. All Walrasian deviation values fall well below the corresponding loose theoretical upper bounds. The last iterate's Walrasian deviation was the same as best iterate's in economies 5 and 6, while in economy 7, the best iterate's Walrasian deviation was only marginally smaller than the last iterate's.}
   \label{table:convergence_bounds_alt}
    \begin{center}
        \renewcommand{\arraystretch}{1.1}
        \begin{tabular}{| m{.5cm} | m{3cm} m{3cm} m{2.5cm} m{2cm} |}
            \hline
            No. & Loose Upper Bound  & Minimum Deviation & Final Deviation & Step Size \\ 
            \hline 
            \hline
            1 & $2.22 \times 10^{1}$ & $7.66 \times 10^{-4}$ & $1.03 \times 10^{-3}$ & $2.30 \times 10^{0}$ \\
            2 & $3.91 \times 10^{2}$ & $6.98 \times 10^{-6}$ & $1.15 \times 10^{-2}$ & $6.54 \times 10^{-1}$ \\
            3 & $2.64 \times 10^{5}$ & $2.90 \times 10^{-1}$ & $2.91 \times 10^{-1}$ & $9.68 \times 10^{-4}$ \\
            4 & $4.50 \times 10^{3}$ & $1.19 \times 10^{-2}$ & $1.19 \times 10^{-2}$ & $5.68 \times 10^{-2}$ \\
            5 & $8.22 \times 10^{6}$ & $7.32 \times 10^{4}$ & $7.32 \times 10^{4}$ & $3.11 \times 10^{-5}$ \\
            6 & $3.57 \times 10^{6}$ & $1.46 \times 10^{4}$ & $1.46 \times 10^{4}$ & $7.17 \times 10^{-5}$ \\
            7 & $6.63 \times 10^{4}$ & $7.60 \times 10^{-2}$ & $7.60 \times 10^{-2}$ & $3.86 \times 10^{-3}$ \\
            \hline
        \end{tabular}
        \renewcommand{\arraystretch}{1.0}
    \end{center}
\end{table}

Specifically, we conduct a grid search over the step-size parameter for each economy to find appropriate settings, by uniformly discretizing an interval $[\eta_{\text{min}}, \eta_{\text{max}}]$ into 200 grid points.
We then record the best-iterate Walrasian deviation $(\varepsilon)$ for each of these points and the maximum Lipschitz coefficient (see~\Cref{table:convergence_bounds_alt}, columns 2 and 3).
The results of these grid searches for all seven economies are shown in \Cref{fig:ad_grid_search}.
In all cases, the minimum, best-iterate Walrasian deviation (indicated by a red X) was achieved at a step size at which pathwise Bregman continuity does \emph{not\/} hold.%
\footnote{We tested for pathwise-Bregman continuity using a fixed horizon, namely $500$ for economy 1, and $2000$ for economies 2 through 7.
It is conceivable that the property in fact holds for an alternative choice of $\tau$.}
The shaded regions of these plots indicate step sizes where pathwise Bregman continuity \emph{does\/} hold, as does a loose upper bound on the best-iterate Walrasian deviation (described below).


For our first set experiments with Arrow-Debreu exchange economies,
we choose as our step size that for which the best-iterate Walrasian deviation is minimized, \emph{among step sizes at which pathwise Bregman continuity is satisfied}---points marked by blue circles in the plots.
This choice ensures that our theory applies.
Having taken this choice, \Cref{fig:experiment_plots_alt} depicts the progression of the Walrasian deviations at prices generated during each iteration of extrat\^atonnement in all seven Arrow-Debreu economies.
We observe rapid convergence in those economies that lack linear consumers (1-4 and 7); but in economies with linear consumers (5 and 6), where results in the literature attest to the difficulty of convergence \citep{cole2019balancing}, 
convergence is notably slower. 
To nonetheless claim convergence in economies with linear consumers, we use \Cref{thm:mirror_extragradient_global_convergence} to compute a loose upper bound on the best-iterate Walrasian deviations. 
As we do not know $x^*$,
we replace the divergence from $x^*$ to $x^{(0)}$ with the maximal distance from $x^{(0)}$ to any point in the price space.
Since our price space is a box, this maximal distance is achieved by picking the farthest corner of the box from $x^{(0)}$. 
This upper bound is looser than the one presented in \Cref{thm:mirror_extragradient_global_convergence}; hence, violating this upper bound would certainly imply non-convergence.
The results presented in \Cref{table:convergence_bounds_alt} (column 1) show that no such violations are encountered, even in economies inhabited by linear consumers. 

\Cref{table:convergence_bounds_alt} (column 4) also lists the step sizes used in each experiment.
It is interesting to note the relationship between the choice of step size and the reported deviations.
For example, in economy 3 (CES with $\rho \in (0,1)$, where the step size is very small, convergence is likewise slower.
On the other hand, in economy 1 (Leontief), where the step size is large, convergence is fast.


\begin{figure}[htbp]
    \centering
    
    \begin{subfigure}[ht]{0.46\textwidth}
        \centering
        \includegraphics[width=\linewidth, height=60mm, keepaspectratio]{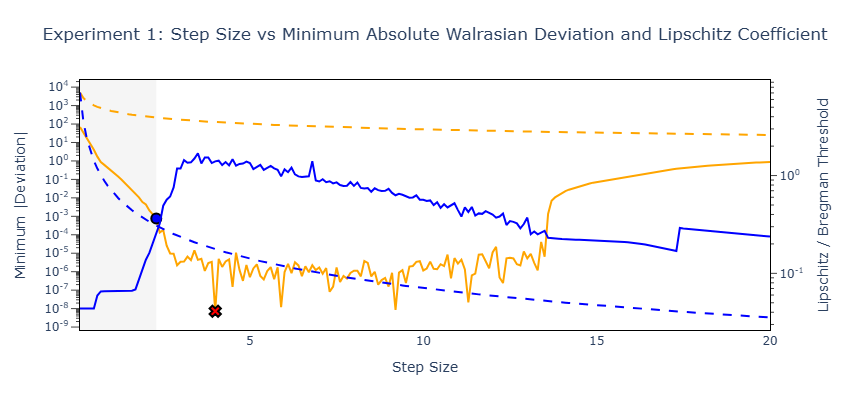}
    \end{subfigure}
    \hfill
    \begin{subfigure}[ht]{0.46\textwidth}
        \centering
        \includegraphics[width=\linewidth, height=60mm, keepaspectratio]{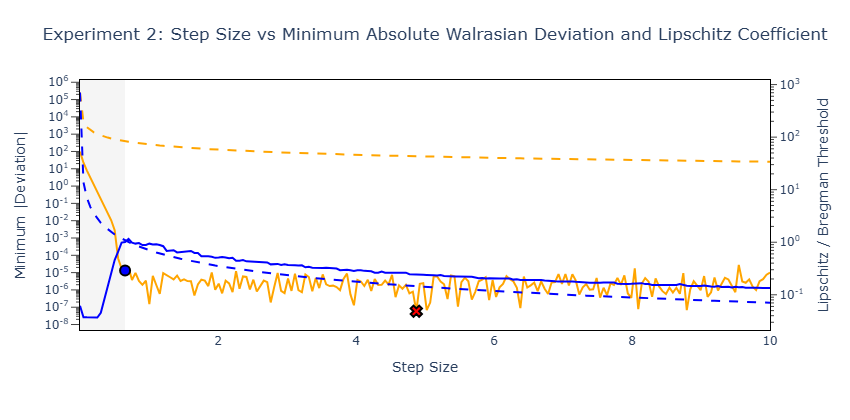}
    \end{subfigure}
    \begin{subfigure}[ht]{0.46\textwidth}
        \centering
        \includegraphics[width=\linewidth, height=60mm, keepaspectratio]{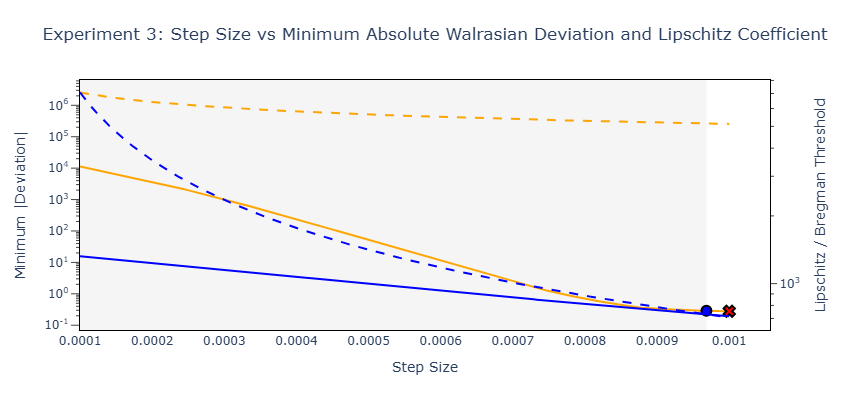}
    \end{subfigure}
    \hfill
    \begin{subfigure}[ht]{0.46\textwidth}
        \centering
        \includegraphics[width=\linewidth, height=60mm, keepaspectratio]{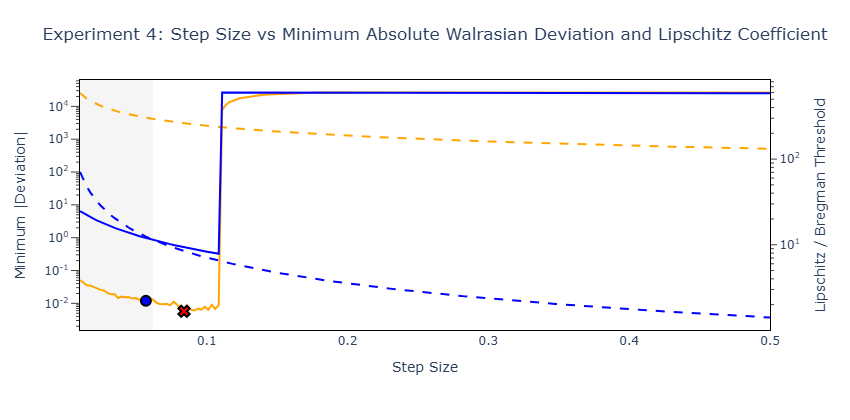}
    \end{subfigure}
       \begin{subfigure}[ht]{0.46\textwidth}
        \centering
        \includegraphics[width=\linewidth, height=60mm, keepaspectratio]{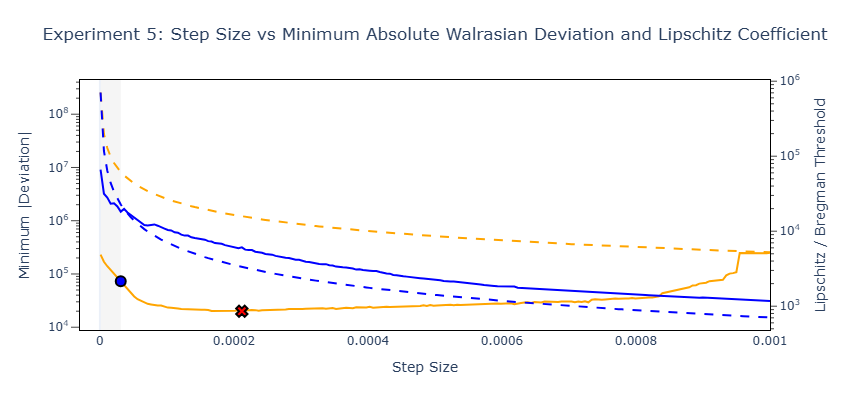}
    \end{subfigure}
    \hfill
    \begin{subfigure}[ht]{0.46\textwidth}
        \centering
        \includegraphics[width=\linewidth, height=60mm, keepaspectratio]{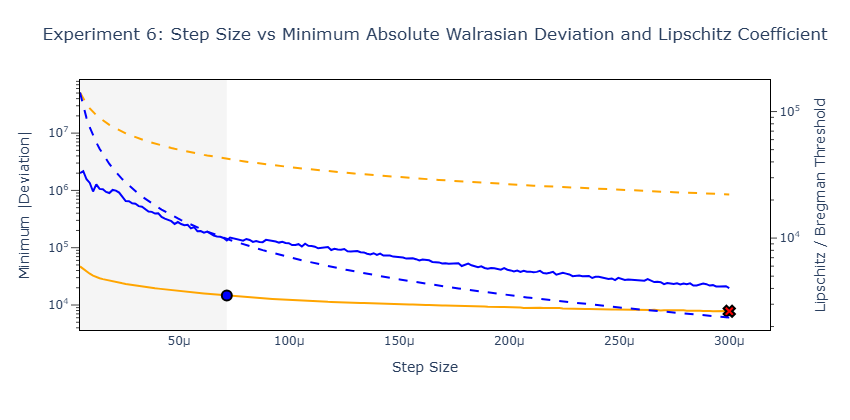}
    \end{subfigure}

    \begin{subfigure}[ht]{0.46\textwidth}
        \centering
        \includegraphics[width=\linewidth, height=60mm, keepaspectratio]{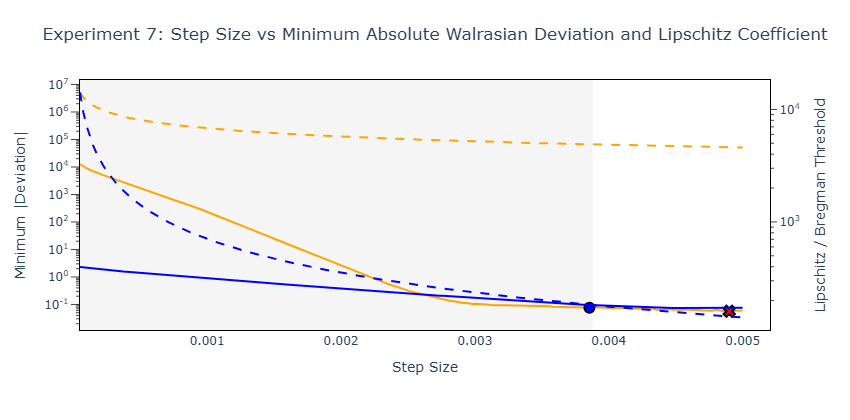}
    \end{subfigure}
    \hfill
    \begin{subfigure}[ht]{0.46\textwidth}
        \centering
        \includegraphics[width=\linewidth, height=60mm, keepaspectratio]{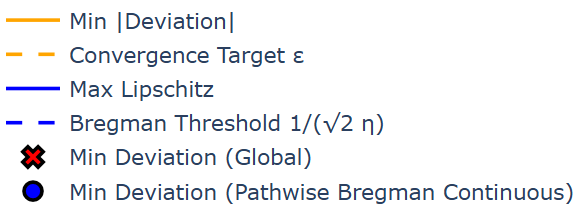}
    \end{subfigure}

    \caption{Grid search over 200 uniformly discretized step sizes. The best-iterate Walrasian deviation and maximum Lipschitz coefficient for each such step size are shown. Among these step sizes, that which achieves the minimum deviation overall and that which achieves the minimum deviation while also satisfying pathwise Bregman continuity are marked. ``Convergence Target'' is a loose upper bound on the best-iterate Walrasian deviation (based on \Cref{thm:mirror_extragradient_global_convergence}), and ``Bregman Threshold'' is the upper bound that the maximum Lipschitz coefficient must fall below for pathwise Bregman continuity to be satisfied (see Appendix~\ref{sec_ap:experiments}). The shaded region in each graph highlights those step sizes at which pathwise Bregman continuity and our loose upper bound are both satisfied, while the red Xs and the blue circles mark our choice step sizes in our experiments.}
    \label{fig:ad_grid_search}
\end{figure}

Finally, going beyond our theory, we run an additional set of experiments on the same seven Arrow-Debreu economies, but this time using step sizes that minimize the best-iterate Walrasian deviation, although pathwise Bregman continuity is not in fact satisfied. 
(I.e., We choose as the step sizes the points indicated by red Xs on the plots in \Cref{fig:ad_grid_search}.)
Although our theory no longer applies, we find that these convergence plots (see Appendix~\ref{sec_ap:ad_experiments}) follow similar trajectories to those shown in \Cref{fig:experiment_plots_alt}.
These experiments suggest that pathwise Bregman continuity, a condition that can only be verified post hoc, may be too restrictive.
We leave as future work the search for a weaker, preferably \emph{ex-ante\/} verifiable, convergence condition.
We also leave as future work the search for theoretical last-iterate guarantees that reflect what we observe in our convergence plots.

\klara{Gabriele Farina has some non-general \href{https://proceedings.mlr.press/v162/anagnostides22a.html}{last-iterate convergence work}.}

\section{Discussion}
A recent paper by \citet{werning2026tatonnement}, revisits the stability of tâtonnement from a complementary but distinct angle. Rather than working within the classical Walrasian framework, they embed a static general equilibrium model into an intertemporal setting with forward-looking monopolistic and monopsonistic price setters subject to pricing frictions, and derive tâtonnement-like price dynamics endogenously from agents' optimizing behavior. Their central stability results rest on two properties of this intertemporal environment: first, forward-looking households use savings to smooth anticipated price changes; 

Our results differ from those of Lorenzoni and Werning along several important dimensions. First, our convergence guarantees are stronger.
They are non-asymptotic---we provide explicit polynomial-time rates of convergence to an approximate Walrasian equilibrium---and global, while their stability analysis is local, relying on a linearization of the excess demand function around equilibrium. Second, our framework requires no additional structure beyond the classical Walrasian setting---we work directly with competitive price-taking agents, without imposing monopolistic market power, pricing frictions, or rational expectations dynamics. Third, our treatment of the Scarf economy does not rely on eliminating income effects through an intertemporal structure: we prove polynomial-time convergence of mirror extratâtonnement to the unique Walrasian equilibrium of the Scarf economy directly, working with its standard Leontief preferences and Walrasian excess demand. More broadly, Lorenzoni and Werning recover stability as a consequence of specific microfoundational assumptions on price setters' behavior. In contrast, our results establish stability as a property of the excess demand structure of the Walrasian economy itself, which can be viewed as a more fundamental statement.
We note that the two approaches are not in conflict: they address adjacent questions and the findings are complementary.

\section{Conclusion}
In this work, we analyze the mirror extragradient method, and show that it is a powerful computational tool for solving variational inequalities, extending existing results to establish polynomial-time convergence under suitable conditions.
Then, by leveraging this framework, we provide the first computationally tractable characterization of Walrasian equilibrium in balanced economies.
Specifically, we introduce the mirror extratâtonnement process, and show that it converges efficiently even in challenging cases like the Scarf economy. 
Our theoretical insights and empirical validation suggest that computational intractability in general equilibrium is largely a consequence of Bregman discontinuities, rather than fundamental hardness, offering a new perspective on Scarf's long-standing challenge in applied general equilibrium theory.

\section{Acknowledgments}
Denizalp Goktas was supported by a JP Morgan AI Fellowship.
This work was further supported by the Office of Naval Research (ONR) grant N00014-24-1-2657.

\bibliographystyle{auxiliary/ACM-Reference-Format}
\bibliography{bib,references}  

\appendix
\section{Related Work}
\label{sec_app:related}
\paragraph{Algorithms for Walrasian Economies}
A detailed inquiry into the computational properties of market equilibria was initiated by \citet{devanur2008market}, who studied a special case of the Arrow-Debreu competitive economy known as the \mydef{Fisher market} \cite{brainard2000compute}.
This model, for which Irving Fisher computed equilibrium prices using a hydraulic machine in the 1890s, is essentially the Arrow-Debreu model of a competitive economy, but there are no firms, and buyers are endowed with only one type of commodity---hereafter good%
\footnote{In the context of Fisher markets, commodities are typically referred to as goods \citep{fisher-tatonnement}, as Fisher markets are often analyzed for a single time period only.
More generally, in Arrow-Debreu markets, where commodities vary by time, location, or state of the world, "an apple today" may be different than "an apple tomorrow". For consistency with the literature, we refer to commodities as goods. }---an artificial currency 
\cite{brainard2000compute, AGT-book}.
\citet{devanur2002market} exploited a connection first made by \citet{eisenberg1961aggregation} between the \mydef{Eisenberg-Gale program} and Walrasian equilibrium to solve Fisher markets assuming buyers with linear utility functions, thereby providing a (centralized) polynomial-time algorithm for equilibrium computation in these markets~\cite{devanur2002market,devanur2008market}.
Their work was built upon by \citet{jain2005market}, who extended the Eisenberg-Gale program to all Fisher markets in which buyers have \mydef{continuous, quasi-concave, and homogeneous} utility functions, and proved that the equilibrium of Fisher markets with such buyers can be computed in polynomial time by interior point methods. 

Concurrent with this line of work on computing Walrasian equilibrium using centralized methods, a line of work on devising and proving 
convergence guarantees for price-adjustment processes (i.e., iterative algorithms that update prices according to a predetermined update rule) developed.
This literature has focused on devising \emph{natural\/} price-adjustment processes, like \emph{t\^atonnement}, which might explain or imitate the movement of prices in real-world markets.
In addition to imitating the law of supply and demand, \emph{t\^atonnement} has been observed to replicate the movement of prices in lab experiments, where participants are given endowments and asked to trade with one another \cite{gillen2020divergence}.
Perhaps more importantly, the main premise of research on the stability of Walrasian equilibrium in computer science 
is that for Walrasian equilibrium to be justified, not only should it be backed by a natural price-adjustment process as economists have long argued, but it should also be computationally efficient \cite{AGT-book}.

Another line of work considers price-adjustment processes in variants of Fisher markets.
\citet{cole2008fast} analyzed \emph{t\^atonnement\/} in a real-world-like model satisfying WGS called the ongoing market model.
In this model, \emph{t\^atonnement\/} once-again converges in polynomial-time \cite{cole2008fast, cole2010discrete}, and it has the advantage that it can be seen as an abstraction for market processes.
\citeauthor{cole2008fast}'s results were later extended by \citet{cheung2012tatonnement} to ongoing markets with \mydef{weak gross complements}, i.e., the excess demand of any commodity weakly increases if the price of any other commodity weakly decreases, fixing all other prices, and ongoing markets with a mix of WGC and WGS commodities.
The ongoing market model these two papers study contains as a special case the Fisher market; however \citet{cole2008fast} assume bounded own-price elasticity of Marshallian demand, and bounded income elasticity of Marshallian demand, while \citet{cheung2012tatonnement} assume, in addition to \citeauthor{cole2008fast}'s assumptions, bounded adversarial market elasticity, which can be seen as a variant of bounded cross-price elasticity of Marshallian demand, from below.
With these assumptions, these results cover Fisher markets with a small range of the well-known CES utilities, including CES Fisher markets with $\rho \in [0, 1)$ and WGC Fisher markets with $\rho \in (- 1, 0]$.%

\citet{fisher-tatonnement} built on this work by establishing the convergence of \emph{t\^atonnement\/} in polynomial time in nested CES Fisher markets, excluding the limiting cases of linear and Leontief markets, but nonetheless extending polynomial-time convergence guarantees for \emph{t\^atonnement\/} to Leontief Fisher markets as well.
More recently, \citet{cheung2018amortized} showed that \citeauthor{fisher-tatonnement}'s [\citeyear{fisher-tatonnement}] result extends to an asynchronous version of \emph{t\^atonnement}, in which good prices are updated during different time periods. 
In a similar vein, \citet{cheung2019tracing} analyzed \emph{t\^atonnement\/} in online Fisher markets, determining that \emph{t\^atonnement\/} tracks Walrasian equilibrium prices closely provided the market changes slowly.

Another price-adjustment process that has been shown to converge to market equilibria in Fisher markets is \mydef{proportional response dynamics}, first introduced by \citet{first-prop-response} for linear utilities; then expanded upon and shown to converge by \cite{proportional-response} for all CES utilities; and very recently shown to converge in Arrow-Debreu exchange economies with linear and CES ($\rho \in [0,1)$) utilities by \citeauthor{branzei2021proportional}. 
The study of the proportional response process was proven fundamental when \citeauthor{fisher-tatonnement} noticed its relationship to gradient descent.
This discovery opened up a new realm of possibilities in analyzing the convergence of market equilibrium processes.
For example, it allowed \citet{cheung2018dynamics} to generalize the convergence results of proportional response dynamics to Fisher markets for buyers with mixed CES utilities.
This same idea was applied by \citet{fisher-tatonnement} to prove the convergence of \emph{t\^atonnement\/} in Leontief Fisher markets, using the equivalence between mirror descent \cite{boyd2004convex}
on the dual of the Eisenberg-Gale program 
and \emph{t\^atonnement}, first observed by \citet{devanur2008market}.
More recently, \citet{gao2020first} developed 
methods to solve the Eisenberg-Gale convex program in the case of linear, quasi-linear, and Leontief Fisher markets.

An alternative to the (global) competitive economy model, in which an agent's trading partners are unconstrained, is the \citet{kakade2004graphical} model of a graphical economies.
This model features local markets, in which each agent can set its own prices for purchase only by neighboring agents, and likewise can purchase only from neighboring agents. 
Auction-like price-adjustment processes have been shown to converge in variants of this model assuming WGS \cite{andrade2021graphical}.

\paragraph{Algorithms for Variational Inequalities}

Variational inequalities \cite{facchinei2003finite} are a mathematical modeling framework whose study dates back to the early 1960s \cite{lions1967variational, hartman1966some, browder1965nonlinear, grioli1973proprieta, brezis2011methodes}. Their utility lies in their very broad mathematical formulation which allows one to solve other mathematical modeling problems using the tools of functional analysis. They have found a great number of applications to problems in engineering and finance \cite{facchinei2003finite} over the years, and have seen an increased interest due to their novel applications in machine learning, to problems ranging from the training of generative adversarial neural networks \cite{goodfellow2014gan} to robust optimization \cite{ben2009robust}.

Historically, the goal of the literature on solution methods for VIs has been to devise algorithms which are asymptotically guaranteed to converge to a strong or weak solution \cite{brezis2011functional}. An overwhelming majority of these works have focused on first-order methods for computing solutions of VIs, with higher order methods having been considered only in recent years (see, for instance, \citet{he2022convergence, huang2022approximation})
While a strong solution of a VI is guaranteed to exist in continuous VIs, most results on the computational complexity of strong solutions, concerns the class of monotone VIs (see, for instance \citet{cai2022tight}) with a few works focusing on VIs that satisfy the Minty condition (see, for instance, \citet{diakonikolas2020halpern}). 

The canonical algorithm to solve VIs is the projected gradient method \cite{cauchy1847methode, nesterov1998introductory} (\sklara{also known under such names as}{which generalizes methods such as} subgradient, gradient descent ascent and Arrow-Hurwicz-Uzawa \cite{arrow-hurwicz, arrow1958studies}). While asymptotic convergence of the projected gradient method to a solution can be shown for a subset of monotone VIs known as strongly monotone VIs\footnote{Recall that for monotone VIs, the set of strong and weak solutions are equal, as such here ``solution'' refers to both strong and weak solutions.}, in general monotone VIs, only ergodic asymptotic convergence (i.e., asymptotic convergence of the averaged iterates) to a strong solution can be guaranteed.
The earliest known algorithm with asymptotic convergence guarantees to a solution of a monotone VI, is the extragradient method, attributed to \citet{korpelevich1976extragradient}. Following this earlier success, \citet{popov1980modification} introduced a closely related algorithm called the optimistic gradient method which he also showed to converge to a solution. These initial extragradient and optimistic gradient algorithms would eventually become much more sophisticated with a large body of work appearing on  asymptotic convergence guarantees for variants of these earlier methods (e.g., \cite{solodov1999hybrid})\sklara{.}{}
of the optimality operator\sklara{}{.}

More recently, the literature has turned its attention to algorithms with non-asymptotic guarantees, and in particular to ones that are guaranteed to compute a $\vepsilon$-strong or $\vepsilon$-weak solution of a VI in polynomial-time, i.e., in a number of evaluations of the optimality operator $\vioperset$ which is polynomial in the inverse of the approximation parameter $\nicefrac{1}{\vepsilon}$, the dimensionality $\spacedim$ of the constraint set, and other relevant assumption specific parameters such as an upper bound on all of the values of the optimality operator. One of the earliest results in this direction was given by \citet{nemirovski2004prox}, who introduced the conceptual mirror-Prox Method, an elegant generalization of the extragradient Method, and established that $\vepsilon$-strong and $\vepsilon$-weak solutions can be computed in $O(\nicefrac{1}{\vepsilon})$ operations by averaging the iterates of the algorithm under the assumption that the the VI is monotone, and the optimality operator is Lipschitz-continuous. \citeauthor{nemirovski2004prox}'s work was subsequently followed by a large body of work on more sophisticated algorithms (e.g., 
\cite{auslender2005interior, diakonikolas2020halpern}) for monotone VIs, and better computational results for the projection method \cite{gidel2018variational}, the extragradient method  \cite{gorbunov2022extragradient, golowich2020tight, cai2022tight} and the optimistic gradient method \cite{gorbunov2022last}.

More recently, a number of works have considered first-order methods to compute a strong solution (e.g., \citet{loizou2021quasiSGDA, he2022convergence, diakonikolas2020halpern}) in VIs or a stationary point of the VI\footnote{A $(\varepsilon, \delta)$ stationary point of a VI $(\set, \vioperset)$ is a point $\vartuple[][*] \in \set$ s.t. for some $\vdelta \geq 0$ there exists $\vartuple \in \ball[\vdelta][{\vartuple[][*]}]$ and $\vartuple$ is a $\vepsilon$-strong solution. Convergence to this weaker solution concept is necessary for VIs in which $\vioperset$ is not singleton-valued for technical reasons, and any future work that seeks to generalize the results in this section should adopt this weaker definition to prove their convergence results.} (e.g., \citet{liu2021first}) that satisfy the Minty condition. The first-order methods considered by this more recent line of work on non-monotone variational inequalities include the extragradient method (e.g., \cite{wang2024extragradient, ofem2023modified}), Tseng's method (e.g., \cite{censor2011strong,thong2020self,uzor2023solving, dung2024convergence,aremu2024modified}), and the optimistic gradient method (e.g., \cite{Lin2022Perseus}) and its variants.

\section{Omitted Examples from Section \ref{section:vis}}
\label{sec_app:vi_examples}

\begin{example}[Non-Convergence of Gradient Method]\label{example:monotone_non_convergence}
    Consider the VI $(\set, \vioper)$ with $\set \doteq \R^2$ and $\vioper(\var, \othervar) = (-\othervar, \var)$. For this VI, we have $\svi(\set, \vioper) = \mvi(\set, \vioper) = \left\{ (0, 0) \right\} $. Suppose that $\left(\vartuple[][(0)], \othervartuple[][(0)]\right) \neq (0, 0)$, then for any $\learnrate[ ][ ] > 0$ the iterates generated by the gradient method are given by:
    \begin{align}
        \left(\var[(\numhorizon)], \othervar[(\numhorizon)]\right) \doteq \left( \var[(0)] - \learnrate[ ][ ] \sum_{k = 1}^{\numhorizon} \othervar[(k-1)],  \othervar[(0)] + \learnrate[ ][ ] \sum_{k = 1}^{\numhorizon} \var[(k-1)] \right) && \forall \numhorizon \in \N_{++}
    \end{align}
    and as such are unbounded, i.e., $\left\|(\vartuple[][(\numhorizon)], \othervartuple[][(\numhorizon)]) \right\| \to \infty$.
\end{example}

\begin{example}[Non-convergence in the absence of the Minty condition]\label{example:non_convergence_non_minty}
    Consider the VI $(\set, \vioper)$ where $\set \doteq \R$ and $\vioper(\var) \doteq 1 - \var[2]$. The set of strong solutions of VI is given by $\svi(\set, \vioper) = \{ -1, 1\}$. Notice that for any any $\var > 1$, $\vioper(\var) < 0$. As such, for the mirror (extra)gradient method, for any choice of step size $\learnrate[ ][ ] > 0$, if the initial iterate is initialized s.t. $\var[(0)] > 1$, then $\var[(t)] \to \infty$.
\end{example}

\section{Omitted Proofs for Section \ref{section:vis}}
\label{sec_app:vis}
Before we start our analysis of the mirror extragradient method (\Cref{alg:VI_mirror_extragrad}), we first prove the following technical lemma on Bregman divergences.

\begin{lemma}[Bregman Triangle Lemma]\label{lemma:bregman_triangle}
    Consider the Bregman divergence $\divergence[\kernel]: \set \times \set \to \R$ associated with a differentiable kernel function $\kernel: \set \to \R$. Let $\x, \y, \z \in \set$, we then have:
    \begin{align}
        \divergence[\kernel][\x][\z] + \divergence[\kernel][\y][\x] - \divergence[\kernel][\y][\z] = \langle \nabla \kernel(\x) - \nabla \kernel(\z), \x - \y \rangle.
    \end{align}
\end{lemma}

\begin{proof}[Proof of \Cref{lemma:bregman_triangle}]

For all $\x, \y, \z \in \set$, we have:
\begin{align*}
& \divergence[\kernel][\x][\z] + \divergence[\kernel][\y][\x] - \divergence[\kernel][\y][\z]\\
&= \left[\kernel(\x) - \kernel(\z) - \langle \nabla \kernel(\z), \x - \z \rangle \right] +  \left[\kernel(\y) - \kernel(\x) - \langle \nabla \kernel(\x), \y - \x \rangle \right]
- \left[ \kernel(\y) - \kernel(\z) - \langle \nabla \kernel(\z), \y - \z \rangle \right]\\
&= - \langle \nabla \kernel(\z), \x - \z \rangle  - \langle \nabla \kernel(\x), \y - \x \rangle
+ \langle \nabla \kernel(\z), \y - \z \rangle\\
&= \langle \nabla \kernel(\z) - \nabla \kernel(\x), \y - \x \rangle.
\end{align*}
\end{proof}

\subsection{Global Convergence of the Mirror Extragradient Algorithm}

With the above technical lemma in hand, we are now ready to prove a progress lemma for the mirror extragradient method, which describes how the algorithm progresses from one iteration to another. Note that under the Minty condition, the following lemma implies  convergence to a weak solution since setting $\vartuple \doteq \vartuple[][*] \in \mvi(\set, \vioper)$, we obtain $\divergence[\kernel][{\vartuple[][*]}][{\vartuple[][][k]}] > \divergence[\kernel][{\vartuple[][*]}][{\vartuple[][][k + 1]}]$ for all $k \in [\numhorizons]$ (i.e., the distance to the weak solution $\vartuple[][*]$ is strictly decreasing). In addition, note that under the assumptions of the lemma, the VI $(\set, \vioper)$ is continuous, hence also implying convergence to a strong solution since any weak solution is also a strong solution in continuous VIs; see \Cref{remark:sol_set_propert}.

\begin{lemma}[Mirror Extragradient Progress]\label{lemma:mirror_extragrad_progress}
Consider the mirror extragradient algorithm (\Cref{alg:VI_mirror_extragrad}), run with a VI $(\set, \vioper)$ where $\set$ is non-empty, compact, and convex, together with a 1-strongly-convex kernel function $\kernel$, step size $\learnrate[ ][ ] > 0$, time horizon $\numhorizons \in \N$, and producing outputs $\left\{\vartuple[][][\numhorizon + 0.5], \vartuple[][][\numhorizon + 1]\right\}_{\numhorizon}$. Suppose that there exists a $\lsmooth \geq 0$ s.t. $\frac{1}{2}\left\|\vioper(\vartuple[][][k+0.5]) - \vioper(\vartuple[][][k])\right\|^2 \leq \lsmooth^2 \divergence[\kernel][{\vartuple[][][k+0.5]}][{\vartuple[][][k]}]$ for all $k \in [\numhorizons]$. Then, for all $k \in [\numhorizons]$ and $\vartuple \in \set$, the following inequality holds for the outputs $\left\{\vartuple[][][\numhorizon + 0.5], \vartuple[][][\numhorizon + 1]\right\}_{\numhorizon}$:
\begin{align}
\divergence[\kernel][{\vartuple}][{\vartuple[][][k]}] - \divergence[\kernel][{\vartuple}][{\vartuple[][][k+1]}] \geq  \learnrate[ ][ ]\langle \vioper(\vartuple[][][k+0.5]),\vartuple[][][k+0.5]  - \vartuple \rangle + \left( 1 - (\learnrate[ ][ ]\lsmooth)^2 \right)\divergence[\kernel][{\vartuple[][][k+0.5]}][{\vartuple[][][k]}].
\end{align}
\end{lemma}

\begin{proof}[Proof of \Cref{lemma:mirror_extragrad_progress}]
By the first order optimality conditions of $\vartuple[][][k+0.5]$, we have for all $\vartuple \in \set$:
\begin{align*}
\left\langle \vioper(\vartuple[][][k]) + \frac{1}{\learnrate[ ][ ]} \left( \grad \kernel(\vartuple[][][k+0.5]) - \grad \kernel(\vartuple[][][k]) \right), \vartuple - \vartuple[][][k+0.5] \right\rangle \geq 0.
\end{align*}

Substituting $\vartuple = \vartuple[][][k+1]$ above, we have:
\begin{align}
\langle \vioper(\vartuple[][][k]), \vartuple[][][k+1] - \vartuple[][][k+0.5] \rangle 
&\geq \frac{1}{\learnrate[ ][ ]} \langle \grad \kernel(\vartuple[][][k]) - \grad \kernel(\vartuple[][][k+0.5]),  \vartuple[][][k+1] - \vartuple[][][k+0.5] \rangle \notag \\
&= \frac{1}{\learnrate[ ][ ]} \left( \divergence[\kernel][{\vartuple[][][k+0.5]}][{\vartuple[][][k]}] + \divergence[\kernel][{\vartuple[][][k+1]}][{\vartuple[][][k+0.5]}] - \divergence[\kernel][{\vartuple[][][k+1]}][{\vartuple[][][k]}] \right), \label{eq:mirror_extra_opt1}
\end{align}
where the last line was obtained by \Cref{lemma:bregman_triangle}.

On the other hand, by the optimality condition at $\vartuple[][][k+1]$, we have for all $\vartuple \in \set$:
\begin{align*}
\left\langle \vioper(\vartuple[][][k+0.5]), \vartuple - \vartuple[][][k+1] \rangle + \frac{1}{\learnrate[ ][ ]} \left( \grad \kernel(\vartuple[][][k+1]) - \grad \kernel(\vartuple[][][k])\right), \vartuple - \vartuple[][][k+1] \right\rangle \geq 0 .
\end{align*}

Hence, for all $\vartuple \in \set$:
\begin{align*}
\langle \vioper(\vartuple[][][k+0.5]), \vartuple - \vartuple[][][k+1] \rangle 
&\geq \frac{1}{\learnrate[ ][ ]} \langle \grad \kernel(\vartuple[][][k])  - \grad \kernel(\vartuple[][][k+1]), \vartuple - \vartuple[][][k+1] \rangle \\
&= \frac{1}{\learnrate[ ][ ]} \left( \divergence[\kernel][{\vartuple[][][k+1]}][{ \vartuple[][][k]}] + \divergence[\kernel][{\vartuple}][{\vartuple[][][k+1 ]}] - \divergence[\kernel][{\vartuple}][{\vartuple[][][k]}] \right) ,
\end{align*}
where the last line was once again obtained by \Cref{lemma:bregman_triangle}.

Continuing with the above inequality, for any given $\vartuple \in \set$, we have:
\begin{align}
    &\frac{1}{\learnrate[ ][ ]} \left( \divergence[\kernel][{\vartuple[][][k+1]}][{ \vartuple[][][k]}] + \divergence[\kernel][{\vartuple}][{\vartuple[][][k+1]}] - \divergence[\kernel][{\vartuple}][{\vartuple[][][k]}] \right) \\
    &\leq \langle \vioper(\vartuple[][][k+0.5]), \vartuple - \vartuple[][][k+1] \rangle \\
    &= \langle \vioper(\vartuple[][][k+0.5]), \vartuple - \vartuple[][][k+0.5] \rangle + \langle \vioper(\vartuple[][][k+0.5]), \vartuple[][][k+0.5] - \vartuple[][][k+1] \rangle \\
    &= \langle \vioper(\vartuple[][][k+0.5]), \vartuple - \vartuple[][][k+0.5] \rangle + \langle \vioper(\vartuple[][][k+0.5]) - \vioper(\vartuple[][][k]), \vartuple[][][k+0.5] - \vartuple[][][k+1] \rangle \notag \\
    &\quad + \langle \vioper(\vartuple[][][k]), \vartuple[][][k+0.5] - \vartuple[][][k+1] \rangle \\
    &\leq \langle \vioper(\vartuple[][][k+0.5]), \vartuple - \vartuple[][][k+0.5] \rangle + \|\vioper(\vartuple[][][k+0.5]) - \vioper(\vartuple[][][k])\|  \|\vartuple[][][k+0.5] - \vartuple[][][k+1]\| \notag \\
    &\quad + \langle \vioper(\vartuple[][][k]), \vartuple[][][k+0.5] - \vartuple[][][k+1] \rangle ,
\end{align}
where the final line follows by the Cauchy-Schwarz inequality \cite{cauchy1821cours, schwarz1884ueber}.

Recall that by the arithmetic mean-geometric mean inequality, $\forall x, y \in \R_+$, $\frac{x+y}{2} \geq \sqrt{xy}$. Hence, applying the inequality with $x = \learnrate[ ][ ] \|\vioper(\vartuple[][][k+0.5]) - \vioper(\vartuple[][][k])\|^2$ and $y = \nicefrac{1}{\learnrate[ ][ ]} \|\vartuple[][][k+0.5] - \vartuple[][][k+1]\|^2$\sklara{}{:}
\begin{align}
     &\frac{1}{\learnrate[ ][ ]} \left( \divergence[\kernel][{\vartuple[][][k+1]}][{ \vartuple[][][k]}] + \divergence[\kernel][{\vartuple}][{\vartuple[][][k+1]}] - \divergence[\kernel][{\vartuple}][{\vartuple[][][k]}] \right) \\
    &\leq \langle \vioper(\vartuple[][][k+0.5]), \vartuple - \vartuple[][][k+0.5] \rangle + \frac{\learnrate[ ][ ] \|\vioper(\vartuple[][][k+0.5]) - \vioper(\vartuple[][][k])\|^2}{2} \notag \\
    &\quad + \frac{\|\vartuple[][][k+0.5] - \vartuple[][][k+1]\|^2}{2\learnrate[ ][ ]} + \langle \vioper(\vartuple[][][k]), \vartuple[][][k+0.5] - \vartuple[][][k+1] \rangle \\
    &\leq \langle \vioper(\vartuple[][][k+0.5]), \vartuple - \vartuple[][][k+0.5] \rangle + \learnrate[ ][ ]\lsmooth^2 \divergence[\kernel][{\vartuple[][][k+0.5]}][{\vartuple[][][k]}] \notag \\
    &\quad + \frac{\|\vartuple[][][k+0.5] - \vartuple[][][k+1]\|^2}{2\learnrate[ ][ ]} + \langle \vioper(\vartuple[][][k]), \vartuple[][][k+0.5] - \vartuple[][][k+1] \rangle 
\end{align}
where the last line was obtained by the initial assumption that there exists a $\lsmooth \geq 0$ s.t. $\frac{1}{2}\|\vioper(\vartuple[][][k+0.5]) - \vioper(\vartuple[][][k])\|^2 \leq \lsmooth^2 \divergence[\kernel][{\vartuple[][][k+0.5]}][{\vartuple[][][k]}]$, for all $k \in [\numhorizons]$.

Additionally note that by the strong convexity of $\kernel$, we have $\forall \vartuple, \othervartuple \in \set$, $\divergence[\kernel][{\vartuple}][{\othervartuple}] \geq \nicefrac{1}{2}\| \vartuple - \othervartuple\|^2$. Hence, continuing we have:
\begin{align}
    &\frac{1}{\learnrate[ ][ ]} \left( \divergence[\kernel][{\vartuple[][][k+1]}][{ \vartuple[][][k]}] + \divergence[\kernel][{\vartuple}][{\vartuple[][][k + 1]}] - \divergence[\kernel][{\vartuple}][{\vartuple[][][k]}] \right) \\
    &\leq \langle \vioper(\vartuple[][][k+0.5]), \vartuple - \vartuple[][][k+0.5] \rangle + \learnrate[ ][ ]\lsmooth^2 \divergence[\kernel][{\vartuple[][][k+0.5]}][{\vartuple[][][k]}] \notag \\
    &\quad + \frac{\divergence[\kernel][{\vartuple[][][k+1]}][{\vartuple[][][k+0.5]}]}{\learnrate[ ][ ]} + \langle \vioper(\vartuple[][][k]), \vartuple[][][k+0.5] - \vartuple[][][k+1] \rangle.
\end{align}

Plugging \Cref{eq:mirror_extra_opt1} into the above, we have:
\begin{align*}
    &\frac{1}{\learnrate[ ][ ]} \left( \divergence[\kernel][{\vartuple[][][k+1]}][{ \vartuple[][][k]}] + \divergence[\kernel][{\vartuple}][{\vartuple[][][k+1]}] - \divergence[\kernel][{\vartuple}][{\vartuple[][][k]}] \right)\\ 
    &\leq \langle \vioper(\vartuple[][][k+0.5]), \vartuple - \vartuple[][][k+0.5] \rangle + \learnrate[ ][ ]\lsmooth^2 \divergence[\kernel][{\vartuple[][][k+0.5]}][{\vartuple[][][k]}] \notag \\
    &\quad + \frac{\divergence[\kernel][{\vartuple[][][k+1]}][{\vartuple[][][k+0.5]}]}{\learnrate[ ][ ]} - \frac{1}{\learnrate[ ][ ]} \left( \divergence[\kernel][{\vartuple[][][k+0.5]}][{\vartuple[][][k]}] + \divergence[\kernel][{\vartuple[][][k+1]}][{\vartuple[][][k+0.5]}] - \divergence[\kernel][{\vartuple[][][k+1]}][{\vartuple[][][k]}] \right)\\
    &\leq \langle \vioper(\vartuple[][][k+0.5]), \vartuple - \vartuple[][][k+0.5] \rangle + \left(\learnrate[ ][ ]\lsmooth^2 - \frac{1}{\learnrate[ ][ ]}\right)\divergence[\kernel][{\vartuple[][][k+0.5]}][{\vartuple[][][k]}] + \frac{1}{\learnrate[ ][ ]} \divergence[\kernel][{\vartuple[][][k+1]}][{\vartuple[][][k]}] 
\end{align*}

Canceling out terms, we simplify the above inequality into:
\begin{align}
    \frac{1}{\learnrate[ ][ ]} \left(\divergence[\kernel][{\vartuple}][{\vartuple[][][k+1]}] - \divergence[\kernel][{\vartuple}][{\vartuple[][][k]}] \right)\leq  \langle \vioper(\vartuple[][][k+0.5]), \vartuple - \vartuple[][][k+0.5] \rangle + \left(\learnrate[ ][ ]\lsmooth^2 - \frac{1}{\learnrate[ ][ ]}\right)\divergence[\kernel][{\vartuple[][][k+0.5]}][{\vartuple[][][k]}] 
\end{align}

Multiplying both sides by $-\learnrate[ ][ ] < 0$, we obtain the lemma statement.
\end{proof}

While as previously noted, the above lemma implies \emph{asymptotic} convergence to a strong solution, to show polynomial-time computation of an $\varepsilon$-strong solution, we have to bound the progress of the intermediate iterates $\divergence[\kernel][{\vartuple[][][k+0.5]}][{\vartuple[][][k]}]$ as a function of the time horizon algorithm. In the proof of the following theorem, we show that we can bound this quantity, assuming that the kernel function is, in addition to being 1-strongly-convex, also $\kernelsmooth$-Lipschitz-smooth.

\thmmirrorextragradglobal*
\begin{proof}[Proof of \Cref{thm:mirror_extragradient_global_convergence}]
    Taking \Cref{lemma:mirror_extragrad_progress} with $\vartuple \doteq \vartuple[][*] \in \mvi(\set, \vioper)$, we have:
    \begin{align*}
        \divergence[\kernel][{\vartuple[][*]}][{\vartuple[][][k]}] - \divergence[\kernel][{\vartuple[][*]}][{\vartuple[][][k+1]}] &\geq  \learnrate[ ][ ] \underbrace{\langle \vioper(\vartuple[][][k+0.5]), \vartuple[][][k+0.5] - \vartuple[][*] \rangle}_{\geq 0} + \left(1 - (\learnrate[ ][ ]\lsmooth)^2 \right)\divergence[\kernel][{\vartuple[][][k+0.5]}][{\vartuple[][][k]}] \\
        &\geq \left(1  - (\learnrate[ ][ ]\lsmooth)^2 \right) \divergence[\kernel][{\vartuple[][][k+0.5]}][{\vartuple[][][k]}]
    \end{align*}

    Multiplying both sides by $\left(1  - (\learnrate[ ][ ]\lsmooth)^2 \right)^{-1} > 0$, we have:
    \begin{align}
        \divergence[\kernel][{\vartuple[][][k+0.5]}][{\vartuple[][][k]}] 
        &\leq \frac{1}{1 - (\learnrate[ ][ ]\lsmooth)^2} \left(\divergence[\kernel][{\vartuple[][*]}][{\vartuple[][][k]}] - \divergence[\kernel][{\vartuple[][*]}][{\vartuple[][][k+1]}] \right)
    \end{align}
    
    Summing up for $k  = 0, \hdots, \numhorizons$:
    \begin{align}
        \sum_{k = 0}^{\numhorizons} \divergence[\kernel][{\vartuple[][][k+0.5]}][{\vartuple[][][k]}]  &\leq \frac{1}{1 - (\learnrate[ ][ ]\lsmooth)^2} \sum_{k = 0}^{\numhorizons} \left(\divergence[\kernel][{\vartuple[][*]}][{\vartuple[][][k]}] - \divergence[\kernel][{\vartuple[][*]}][{\vartuple[][][k+1]}] \right)\\
        &\leq \frac{1}{1 - (\learnrate[ ][ ]\lsmooth)^2} \left(\divergence[\kernel][{\vartuple[][*]}][{\vartuple[][][0]}] - \divergence[\kernel][{\vartuple[][*]}][{\vartuple[][][\numhorizons+1]}] \right)\\
        &\leq \frac{1}{1 - (\learnrate[ ][ ]\lsmooth)^2} \divergence[\kernel][{\vartuple[][*]}][{\vartuple[][][0]}] 
    \end{align}
    Dividing both sides by $\numhorizons$, we have:
    \begin{align}
        \frac{1}{\numhorizons}\sum_{k = 0}^{\numhorizons} \divergence[\kernel][{\vartuple[][][k+0.5]}][{\vartuple[][][k]}] &\leq \frac{1}{\numhorizons\left(1 - (\learnrate[ ][ ]\lsmooth)^2\right)} \left(\divergence[\kernel][{\vartuple[][*]}][{\vartuple[][][0]}] \right)\\
        \min_{k = 0, \hdots, \numhorizons} \divergence[\kernel][{\vartuple[][][k+0.5]}][{\vartuple[][][k]}]  &\leq \frac{1}{\numhorizons\left(1 - (\learnrate[ ][ ]\lsmooth)^2\right)} \left(\divergence[\kernel][{\vartuple[][*]}][{\vartuple[][][0]}] \right)\label{eq:intermediate_progress_bound1}
    \end{align}

    We can transform this convergence into a convergence in terms of the primal gap function\sklara{}{, i.e., the quantity $\langle \vioper(\vartuple[][][k + 0.5]), \vartuple[][][k+0.5] - \vartuple \rangle$ that measures the worst-case violation of the strong solution condition}. Now, recall by the first order optimality conditions of $\vartuple[][][k+0.5]$, we have for all $\vartuple \in \set$:
    \begin{align*}
    \left\langle \vioper(\vartuple[][][k]) + \frac{1}{\learnrate[ ][ ]} \left( \grad \kernel(\vartuple[][][k+0.5]) - \grad \kernel(\vartuple[][][k]) \right), \vartuple - \vartuple[][][k+0.5] \right\rangle \geq 0.
    \end{align*}
Re-organizing, for all $\vartuple \in \set$, and $k \in [\numhorizons]$ we have:
\begin{align}
         \langle \vioper(\vartuple[][][k]), \vartuple[][][k+0.5] - \vartuple \rangle &\leq \frac{1}{\learnrate[ ][ ]} \left\| \grad \kernel(\vartuple[][][k+0.5]) - \grad \kernel(\vartuple[][][k])\right\| \left\|\vartuple[][][k+0.5] - \vartuple \right\|\\
         &\leq \frac{\diam(\set)}{\learnrate[ ][ ]} \left\| \grad \kernel(\vartuple[][][k+0.5]) - \grad \kernel(\vartuple[][][k])\right\|\\
         &\leq \frac{\diam(\set)\kernelsmooth}{\learnrate[ ][ ]} \left\| \vartuple[][][k+0.5] - \vartuple[][][k]\right\|
    \end{align}
    where the last line follow from $h$ being $\kernelsmooth$-Lipschitz-smooth.

Now, with the above inequality in hand, notice that for all $\vartuple \in \set$ and $k \in [\numhorizons]$, we have:
\begin{align*}
    \langle \vioper(\vartuple[][][k + 0.5]), \vartuple[][][k+0.5] - \vartuple \rangle &= \langle \vioper(\vartuple[][][k]), \vartuple[][][k+0.5] - \vartuple \rangle + \langle \vioper(\vartuple[][][k+0.5]) - \vioper(\vartuple[][][k]), \vartuple[][][k+0.5] - \vartuple \rangle \\
    &\leq \frac{\diam(\set)\kernelsmooth}{\learnrate[ ][ ]} \|\vartuple[][][k+0.5] - \vartuple[][][k]\| + \|\vioper(\vartuple[][][k+0.5]) - \vioper(\vartuple[][][k])\|  \|\vartuple[][][k+0.5] - \vartuple\|\\
    &\leq \frac{\diam(\set)\kernelsmooth}{\learnrate[ ][ ]} \|\vartuple[][][k+0.5] - \vartuple[][][k]\| + \lsmooth \sqrt{2 \divergence[\kernel][{\vartuple[][][k+0.5]}][{\vartuple[][][k]}]}  \|\vartuple[][][k+0.5] - \vartuple\|\\
    &\leq \frac{\diam(\set)\kernelsmooth}{\learnrate[ ][ ]} \sqrt{2 \divergence[\kernel][{\vartuple[][][k+0.5]}][{\vartuple[][][k]}]} + \diam(\set) \lsmooth \sqrt{2 \divergence[\kernel][{\vartuple[][][k+0.5]}][{\vartuple[][][k]}]}\\
    &\leq \diam(\set) \left( \frac{\kernelsmooth}{\learnrate[ ][ ]} + \lsmooth \right) \sqrt{2 \divergence[\kernel][{\vartuple[][][k+0.5]}][{\vartuple[][][k]}]},
\end{align*}
where the \sklara{penultimate}{third-to-last} line was obtained by the assumption that there exists $\lsmooth \geq 0$, s.t. $\frac{1}{2}\|\vioper(\vartuple[][][k+0.5]) - \vioper(\vartuple[][][k])\|^2 \leq \lsmooth^2 \divergence[\kernel][{\vartuple[][][k+0.5]}][{\vartuple[][][k]}]$, and the \sklara{last}{penultimate} line from the strong convexity of $\kernel$, which means that we have $\forall \vartuple, \othervartuple \in \set$, $\divergence[\kernel][{\vartuple}][{\othervartuple}] \geq \nicefrac{1}{2}\| \vartuple - \othervartuple\|^2$. 

Now, let $k^* \in \argmin_{k = 0, \hdots, \numhorizons} \divergence[\kernel](\vartuple[][][k+0.5], \vartuple[][][k])$, we then have for all $\vartuple \in \set$:
\begin{align*}
    \langle \vioper(\vartuple[][][k^* + 0.5]), \vartuple[][][k^*+0.5] - \vartuple \rangle 
    &\leq \diam(\set) \left( \frac{\kernelsmooth}{\learnrate[ ][ ]} + \lsmooth \right) \sqrt{2 \divergence[\kernel][{\vartuple[][][{k^*+0.5}]}][{\vartuple[][][{k^*}]}]}\\
    &= \diam(\set) \left( \frac{\kernelsmooth}{\learnrate[ ][ ]} + \lsmooth \right)  \sqrt{2 \min_{k= 0, \hdots, \numhorizons}\divergence[\kernel][{\vartuple[][][k+0.5]}][{\vartuple[][][k]}]}
\end{align*}
%

Now, plugging \Cref{eq:intermediate_progress_bound1} in the above, we have  for all $\vartuple \in \set$:
\begin{align*}
    \langle \vioper(\vartuple[][][k^* + 0.5]), \vartuple[][][k^* +0.5] - \vartuple \rangle 
    &\leq \frac{\sqrt{2} \diam(\set) \left( \frac{\kernelsmooth}{\learnrate[ ][ ]} + \lsmooth \right)}{\sqrt{1 - (\learnrate[ ][ ]\lsmooth)^2}} \frac{\sqrt{\divergence[\kernel][{\vartuple[][*]}][{\vartuple[][][0]}]}}{\sqrt{\numhorizons}}
\end{align*}

Now, by the assumption that $\learnrate[ ][ ] \leq \frac{1}{\sqrt{2}\lsmooth} < \frac{1}{\lsmooth}$, we have:
\begin{align*}
    \langle \vioper(\vartuple[][][k^* + 0.5]), \vartuple[][][k^* +0.5] - \vartuple \rangle &\leq   \frac{\sqrt{2} \diam(\set) \left( \frac{\kernelsmooth}{\learnrate[ ][ ]} + \frac{1}{\learnrate[ ][ ]} \right)}{\sqrt{1 - (\learnrate[ ][ ]\lsmooth)^2}} \frac{\sqrt{\divergence[\kernel][{\vartuple[][*]}][{\vartuple[][][0]}]}}{\sqrt{\numhorizons}}\\
    &= \frac{(1 + \kernelsmooth)\sqrt{2} \diam(\set)}{\learnrate[ ][ ]\sqrt{1 - (\learnrate[ ][ ]\lsmooth)^2}} \frac{\sqrt{\divergence[\kernel][{\vartuple[][*]}][{\vartuple[][][0]}]}}{\sqrt{\numhorizons}}\\
    &\leq   \frac{(1 + \kernelsmooth)\sqrt{2} \diam(\set)}{\learnrate[ ][ ]\sqrt{1 - (\nicefrac{1}{\sqrt{2}})^2}} \frac{\sqrt{\divergence[\kernel][{\vartuple[][*]}][{\vartuple[][][0]}]}}{\sqrt{\numhorizons}}\\
    &= \frac{2(1 + \kernelsmooth)  \diam(\set)}{\learnrate[ ][ ]} \frac{\sqrt{\divergence[\kernel][{\vartuple[][*]}][{\vartuple[][][0]}]}}{\sqrt{\numhorizons}}.
\end{align*}

This concludes the proof of the first part of the statement. We now have:
\begin{align*}
    \min_{k = 0, \hdots, \numhorizons} \max_{\vartuple \in \set} \langle \vioper(\vartuple[][][k+0.5]), \vartuple[][][k+0.5] - \vartuple \rangle 
    &\leq \max_{\vartuple \in \set}  \langle \vioper(\vartuple[][][k^* + 0.5]), \vartuple[][][k^* +0.5] - \vartuple \rangle \\
    &\leq \frac{2 (1 + \kernelsmooth) \diam(\set)}{\learnrate[ ][ ]} \frac{\sqrt{\divergence[\kernel][{\vartuple[][*]}][{\vartuple[][][0]}]}}{\sqrt{\numhorizons}}
\end{align*}

In addition, for any $\varepsilon > 0$, letting $\frac{2 (1 + \kernelsmooth) \diam(\set)}{\learnrate[ ][ ]} \frac{\sqrt{\divergence[\kernel][{\vartuple[][*]}][{\vartuple[][][0]}]}}{\sqrt{\numhorizons}} \leq \varepsilon$, and solving for $\numhorizons$, we have:
\begin{align*}
    \frac{2 (1 + \kernelsmooth) \diam(\set)}{\learnrate[ ][ ]} \frac{\sqrt{\divergence[\kernel][{\vartuple[][*]}][{\vartuple[][][0]}]}}{\sqrt{\numhorizons}} &\leq \varepsilon\\
    \frac{4 (1 + \kernelsmooth)^2 \diam(\set)^2}{\learnrate[ ][ ]^2} \frac{\divergence[\kernel][{\vartuple[][*]}][{\vartuple[][][0]}]}{\varepsilon^2} &\leq \numhorizons
\end{align*}

That is, $\bestiter[\vartuple][\numhorizons] \in \argmin_{\vartuple[][][k+0.5] : k = 0, \hdots, \numhorizons} \divergence[\kernel](\vartuple[][][k+0.5], \vartuple[][][k])$ is a $\varepsilon$-strong solution after $\frac{4 (1 + \kernelsmooth)^2 \diam(\set)^2}{\learnrate[ ][ ]^2} \frac{\divergence[\kernel][{\vartuple[][*]}][{\vartuple[][][0]}]}{\varepsilon^2}$ iterations of the mirror extragradient algorithm.
\end{proof}

\begin{remark}
\sklara{}{The upper-bound on $\numhorizons$ granted by \Cref{thm:mirror_extragradient_global_convergence} is tight, in the sense that for all
\begin{align*}\numhorizons \geq \frac{4 (1 + \kernelsmooth)^2 \diam(\set)^2}{\learnrate[ ][ ]^2} \frac{\divergence[\kernel][{\vartuple[][*]}][{\vartuple[][][0]}]}{\varepsilon^2},
\end{align*}
the corresponding $\bestiter[\vartuple][\numhorizons]$ is an $\varepsilon$-strong solution, which can inform our choice of $\numhorizons$ in practice. This observation applies to all analogous statements in the sequel.}
\end{remark}

\subsection{Local Convergence of the Mirror Extragradient Algorithm}

To understand how a local weak solution can provide us with local convergence, recall by \Cref{lemma:mirror_extragrad_progress} the iterates of the mirror extragradient algorithm satisfy the following for all $\numhorizon \in \N$:
\begin{align*}
\divergence[\kernel][{\vartuple}][{\vartuple[][][k]}] - \divergence[\kernel][{\vartuple}][{\vartuple[][][k+1]}] \geq  \learnrate[ ][ ]\langle \vioper(\vartuple[][][k+0.5]),\vartuple[][][k+0.5]  - \vartuple \rangle + \left( 1 - (\learnrate[ ][ ]\lsmooth)^2 \right)\divergence[\kernel][{\vartuple[][][k+0.5]}][{\vartuple[][][k]}].
\end{align*}

Suppose that the kernel function $\kernel$ is strictly convex, and that the algorithm has not yet converged, i.e., $\vartuple[][][k+0.5] \neq \vartuple[][][k]$, then $\divergence[\kernel][{\vartuple[][][k+0.5]}][{\vartuple[][][k]}] > 0$, and we can drop the term. Re-organizing the expressions, we then have:
\begin{align*}
\divergence[\kernel][{\vartuple}][{\vartuple[][][k]}] > \divergence[\kernel][{\vartuple}][{\vartuple[][][k+1]}] +   \learnrate[ ][ ] \langle \vioper(\vartuple[][][k+0.5]),\vartuple[][][k+0.5]  - \vartuple \rangle 
\end{align*}

Now, notice that if we can ensure that for all $k \in [\numhorizons]_+$, there exists an $\vartuple[][*] \in \svi(\set, \vioper)$ s.t. $\langle \vioper(\vartuple[][][k+0.5]),\vartuple[][][k+0.5]  - \vartuple[][*] \rangle \geq 0$, then we have: 
$\divergence[\kernel][{\vartuple[][*]}][{\vartuple[][][k]}] > \divergence[\kernel][{\vartuple[][*]}][{\vartuple[][][k+1]}]$, implying that $\vartuple[][][k] \to \vartuple[][*]$, i.e., convergence to a strong solution. Since we cannot assume the existence of a weak solution (i.e., the Minty condition), the next best way to ensure that there exists an $\vartuple[][*] \in \svi(\set, \vioper)$ s.t. $\langle \vioper(\vartuple[][][k+0.5]),\vartuple[][][k+0.5]  - \vartuple[][*] \rangle \geq 0$ is to initialize the algorithm with an initial iterate $\vartuple[][][0] \in \set$ which is $O(\delta)$-close to a $ \delta$-local weak solution $\vartuple[][*] \in \lmvi[][\vdelta](\set, \vioper)$\footnote{Note that a local weak solution is guaranteed to be strong solution by Proposition 3.1 of \citet{aussel2024variational}.} for some $\delta \geq 0$, and ensure that all subsequent intermediary iterates $\left\{\vartuple[][][k + 0.5]\right\}_{k \in [\numhorizons]_{+}}$ remain $\delta$-close to $\vartuple[][*]$.

To ensure this, we have to first bound the distance between the intermediary $\left\{\vartuple[][][k + 0.5]\right\}_{k \in [\numhorizons]_{+}}$ and terminal $\left\{\vartuple[][][k]\right\}_{k \in [\numhorizons]_{+}}$ iterates. The following lemma provides us with such a bound.

\begin{lemma}[Distance bound on intermediate iterates]\label{lemma:extragrad_intermediate_iterate_dist}
        Let $(\set, \vioper)$ be a VI satisfying the Minty condition, and $\kernel$ a $1$-strongly-convex and $\kernelsmooth$-Lipschitz-smooth kernel function. Consider the mirror extragradient algorithm (\Cref{alg:VI_mirror_extragrad}) run with the VI $(\set, \vioper)$, the kernel function $\kernel$, step size $\learnrate[ ][ ] \geq 0$, time horizon $\numhorizons \in \N$, and producing outputs $\{\vartuple[][][\numhorizon + 0.5], \vartuple[][][\numhorizon + 1]\}_{\numhorizon}$. We then have:
    \begin{align}
        \norm[{\vartuple[][][k+0.5] - \vartuple[][][k]}] \leq \learnrate[ ][ ] \lipschitz,
    \end{align}
    where $\lipschitz \doteq \max_{\vartuple \in \set} \| \vioper(\vartuple)\|$.
\end{lemma}

\begin{proof}[Proof of \Cref{lemma:extragrad_intermediate_iterate_dist}]
    Note that for all $k \in [\numhorizons]_+$, by the first order optimality conditions of $\vartuple[][][k+0.5]$, we have for all $\vartuple \in \set$:
    \begin{align*}
    \left\langle \vioper(\vartuple[][][k]) + \frac{1}{\learnrate[ ][ ]} \left( \grad \kernel(\vartuple[][][k+0.5]) - \grad \kernel(\vartuple[][][k]) \right), \vartuple - \vartuple[][][k+0.5] \right\rangle \geq 0.
    \end{align*}
    
    Substituting $\vartuple = \vartuple[][][k]$ above, we have:
    \begin{align}
    \langle \vioper(\vartuple[][][k]), \vartuple[][][k] - \vartuple[][][k+0.5] \rangle 
    &\geq \frac{1}{\learnrate[ ][ ]} \langle \grad \kernel(\vartuple[][][k]) - \grad \kernel(\vartuple[][][k+0.5]), \vartuple[][][k] - \vartuple[][][k+0.5] \rangle \notag \\
    &= \frac{1}{\learnrate[ ][ ]} \left( \divergence[\kernel][{\vartuple[][][k+0.5]}][{\vartuple[][][k]}] + \divergence[\kernel][{\vartuple[][][k]}][{\vartuple[][][k+0.5]}] \right),
    \end{align}

    where the last line was obtained by \Cref{lemma:bregman_triangle}. Re-organizing:
    \begin{align}
    \divergence[\kernel][{\vartuple[][][k+0.5]}][{\vartuple[][][k]}] 
    &\leq \learnrate[ ][ ] \langle \vioper(\vartuple[][][k]), \vartuple[][][k] - \vartuple[][][k+0.5] \rangle -\divergence[\kernel][{\vartuple[][][k]}][{\vartuple[][][k+0.5]}] .\\
    &\leq \learnrate[ ][ ] \langle \vioper(\vartuple[][][k]), \vartuple[][][k] - \vartuple[][][k+0.5] \rangle -\nicefrac{1}{2} \norm[{\vartuple[][][k] - \vartuple[][][k+0.5]}]^2\\
    &\leq \learnrate[ ][ ] \norm[{ \vioper(\vartuple[][][k])}] \norm[{\vartuple[][][k] - \vartuple[][][k+0.5]}] - \nicefrac{1}{2}\norm[{\vartuple[][][k] - \vartuple[][][k+0.5]}]^2\\
    &\leq \learnrate[ ][ ] \lipschitz \norm[{\vartuple[][][k] - \vartuple[][][k+0.5]}] - \nicefrac{1}{2} \norm[{\vartuple[][][k] - \vartuple[][][k+0.5]}]^2
    \end{align}

    \noindent
    \sklara{}{where the second inequality follows if we note that by strong convexity of $\kernel$, we have $\forall \vartuple, \othervartuple \in \set$, $\divergence[\kernel][{\vartuple}][{\othervartuple}] \geq \nicefrac{1}{2}\| \vartuple - \othervartuple\|^2$. }Since for all $z \in \R$, $a b \in \R_+$, we have $az - bz^2 \leq \nicefrac{a^2}{4b}$ \klara{How?}:
    \begin{align}
        \divergence[\kernel][{\vartuple[][][k+0.5]}][{\vartuple[][][k]}] 
        &\leq \frac{\learnrate[ ][ ]^2 \lipschitz^2}{2}.
    \end{align}

    \sklara{Additionally, note that by strong convexity of $\kernel$, we have $\forall \vartuple, \othervartuple \in \set$, $\divergence[\kernel][{\vartuple}][{\othervartuple}] \geq \nicefrac{1}{2}\| \vartuple - \othervartuple\|^2$ or equivalently $\sqrt{2\divergence[\kernel][{\vartuple}][{\othervartuple}]} \geq \| \vartuple - \othervartuple\|^2$. Hence, continuing}{Again exploiting the strong convexity of $\kernel$}:
    \begin{align}
        \norm[{\vartuple[][][k+0.5] - \vartuple[][][k]}] \leq \sqrt{2 \divergence[\kernel][{\vartuple[][][k+0.5]}][{\vartuple[][][k]}] }
        &\leq \learnrate[ ][ ] \lipschitz.
    \end{align}
\end{proof}

With \Cref{lemma:extragrad_intermediate_iterate_dist}}in hand, we can show that if the initial iterate starts close enough to some local weak solution, then the intermediate iterates will remain within this $\vdelta$-ball for the remainder of the algorithm for an appropriate choice of step size.

\begin{lemma}[Mirror Extragradient Iterates Remain Local]\label{lemma:mirror_extragrad_iterates_remain_local}
        Let $(\set, \vioper)$ be a $\lsmooth$-Lipschitz-continuous VI  satisfying the Minty condition, and $\kernel$ a $1$-strongly-convex kernel function. Define  $\lipschitz \doteq \max_{\vartuple \in \set} \| \vioper(\vartuple)\|$. Suppose that for some $\vartuple[][*] \in \lmvi[][\delta](\set, \vioper)$ $\delta$-local weak solution, the initial iterate $\vartuple[][][0] \in \set$ is chosen so that $\sqrt{2\divergence[\kernel][{\vartuple[][*]}][{\vartuple[][][0]}]} \leq \delta - \learnrate[ ][ ] \lipschitz$. Consider the mirror extragradient algorithm (\Cref{alg:VI_mirror_extragrad}) run with the VI $(\set, \vioper)$, kernel function $\kernel$, step size $\learnrate[ ][ ] \geq 0$, initial iterate $\vartuple[][][0]$, and time horizon $\numhorizons \in \N$, producing outputs $\{\vartuple[][][\numhorizon + 0.5], \vartuple[][][\numhorizon + 1]\}_{\numhorizon}$. Then for all $\numhorizon \in [\numhorizons]$, we have 
        \begin{align*}
          &\divergence[\kernel][{\vartuple[][*]}][{\vartuple[][][\numhorizon]}] \leq \nicefrac{1}{2}(\delta - \learnrate[ ][ ] \lipschitz)^2  &\text{and} \\
          &\norm[{\vartuple[][][\numhorizon + 0.5] - \vartuple[][*]}] \leq \delta .
        \end{align*}
\end{lemma}

\begin{proof}[Proof of \Cref{lemma:mirror_extragrad_iterates_remain_local}]
    We will prove the claim by induction on $\numhorizon \in \N_+$.

    \paragraph{Base case: $\numhorizon = 0$.}

    \sklara{}{The first part of the claim holds by assumption. For the second part:}
    \begin{align*}
        \norm[{\vartuple[][][0.5] - \vartuple[][*]}]
        &= \norm[{\vartuple[][][0.5] - \vartuple[][][0] + \vartuple[][][0] - \vartuple[][*]}]\\
        &\leq \norm[{\vartuple[][][0.5] - \vartuple[][][0]}] + \norm[{\vartuple[][][0] - \vartuple[][*]}]\\
        &\leq \norm[{\vartuple[][][0.5] - \vartuple[][][0]}] + \sqrt{2\divergence[\kernel][{\vartuple[][*]}][{\vartuple[][][0]}]}\\
        &\leq \learnrate[ ][ ] \lipschitz + (\delta - \learnrate[ ][ ] \lipschitz) && \text{(\Cref{lemma:extragrad_intermediate_iterate_dist})}\\
        &\leq \delta.
    \end{align*}

    \paragraph{Inductive step:} Suppose that for all $\numhorizon = 0, \hdots, \numhorizons$, \sklara{}{we have} $\norm[{\vartuple[][][\numhorizon + 0.5] - \vartuple[][*]}] \leq \delta$ and $\sqrt{2\divergence[\kernel][{\vartuple[][*]}][{\vartuple[][][\numhorizon]}]} \leq \delta - \learnrate[ ][ ] \lipschitz$, or equivalently, $\divergence[\kernel][{\vartuple[][*]}][{\vartuple[][][\numhorizon]}] \leq \nicefrac{1}{2}(\delta - \learnrate[ ][ ] \lipschitz)^2$. We now show that $\norm[{\vartuple[][][\numhorizons + 1.5] - \vartuple[][*]}] \leq \delta$ and $\divergence[\kernel][{\vartuple[][*]}][{\vartuple[][][\numhorizons +1 ]}] \leq \nicefrac{1}{2}(\delta - \learnrate[ ][ ] \lipschitz)^2$.

    By \Cref{lemma:mirror_extragrad_progress}, we have:
    \begin{align}
        \divergence[\kernel][{\vartuple}][{\vartuple[][][\numhorizons]}] - \divergence[\kernel][{\vartuple}][{\vartuple[][][\numhorizons+1]}] \geq  \learnrate[ ][ ]\langle \vioper(\vartuple[][][\numhorizons+0.5]),\vartuple[][][\numhorizons+0.5]  - \vartuple \rangle + \left( 1 - (\learnrate[ ][ ]\lsmooth)^2 \right)\divergence[\kernel][{\vartuple[][][\numhorizons+0.5]}][{\vartuple[][][\numhorizons]}].
    \end{align}

    Substituting in $\vartuple \doteq \vartuple[][*] \in \lmvi[][\delta](\set, \vioper)$, we have:
    \begin{align*}
        \divergence[\kernel][{\vartuple[][*]}][{\vartuple[][][\numhorizons]}] - \divergence[\kernel][{\vartuple[][*]}][{\vartuple[][][\numhorizons+1]}]&\geq  \learnrate[ ][ ] \underbrace{\langle \vioper(\vartuple[][][\numhorizons+0.5]),\vartuple[][][\numhorizons+0.5]  - \vartuple[][*] \rangle}_{\geq 0} + \left( 1 - (\learnrate[ ][ ]\lsmooth)^2 \right) \underbrace{\divergence[\kernel][{\vartuple[][][\numhorizons+0.5]}][{\vartuple[][][\numhorizons]}]}_{\geq 0}\\
        &\geq 0.
    \end{align*}
    Re-organizing, and using the inductive assumption that $\divergence[\kernel][{\vartuple[][*]}][{\vartuple[][][\numhorizons]}] \leq \nicefrac{1}{2}(\delta - \learnrate[ ][ ] \lipschitz)^2$ we have:
    \begin{align}
        \nicefrac{1}{2}(\delta - \learnrate[ ][ ] \lipschitz)^2 \geq \divergence[\kernel][{\vartuple[][*]}][{\vartuple[][][\numhorizons]}] \geq \divergence[\kernel][{\vartuple[][*]}][{\vartuple[][][\numhorizons+1]}]
    \end{align}

    \sklara{}{This gives us the first part of the claim.} Now \sklara{}{for the second part}, notice that we have:
    \begin{align*}
        \norm[{\vartuple[][][\numhorizons + 0.5] - \vartuple[][*]}] 
        &= \norm[{\vartuple[][][\numhorizons + 0.5] - \vartuple[][][\numhorizons] + \vartuple[][][\numhorizons] - \vartuple[][*]}]\\
        &\leq \norm[{\vartuple[][][\numhorizons + 0.5] - \vartuple[][][\numhorizons]}] + \norm[{\vartuple[][][\numhorizons] - \vartuple[][*]}]\\
        &\leq \norm[{\vartuple[][][\numhorizons + 0.5] - \vartuple[][][\numhorizons]}] + \sqrt{2\divergence[\kernel][{\vartuple[][*]}][{\vartuple[][][\numhorizons]}]}\\
        &\leq \learnrate[ ][ ] \lipschitz + (\delta - \learnrate[ ][ ] \lipschitz) && \text{(\Cref{lemma:extragrad_intermediate_iterate_dist})}\\
        &\leq \delta.
    \end{align*}
\end{proof}

With \Cref{lemma:mirror_extragrad_iterates_remain_local} in hand, modifying the proof of \Cref{thm:mirror_extragradient_global_convergence} slightly, we can show local convergence to a strong solution when the initial iterate of the algorithm is initialized close enough to a local \sklara{}{weak} solution.

\thmvimirrorextragradlocal*
\begin{proof}[Proof of \Cref{thm:vi_mirror_extragrad_local}]
Taking \Cref{lemma:mirror_extragrad_progress} with $\vartuple \doteq \vartuple[][*]$, where $\vartuple[][*]$ is given as in the statement, then by \Cref{lemma:mirror_extragrad_iterates_remain_local}, we have for all $k \in [\numhorizons]$:
    \begin{align*}
        \divergence[\kernel][{\vartuple[][*]}][{\vartuple[][][k]}] - \divergence[\kernel][{\vartuple[][*]}][{\vartuple[][][k+1]}] &\geq  \learnrate[ ][ ] \underbrace{\langle \vioper(\vartuple[][][k+0.5]), \vartuple[][][k+0.5] - \vartuple[][*] \rangle}_{\geq 0} + \left(1 - (\learnrate[ ][ ]\lsmooth)^2 \right)\divergence[\kernel][{\vartuple[][][k+0.5]}][{\vartuple[][][k]}] \\
        &\geq \left(1  - (\learnrate[ ][ ]\lsmooth)^2 \right) \divergence[\kernel][{\vartuple[][][k+0.5]}][{\vartuple[][][k]}].
    \end{align*}
    Multiplying both sides by $\left(1  - (\learnrate[ ][ ]\lsmooth)^2 \right)^{-1} > 0$, we have:
    \begin{align}
        \divergence[\kernel][{\vartuple[][][k+0.5]}][{\vartuple[][][k]}] 
        &\leq \frac{1}{1 - (\learnrate[ ][ ]\lsmooth)^2} \left(\divergence[\kernel][{\vartuple[][*]}][{\vartuple[][][k]}] - \divergence[\kernel][{\vartuple[][*]}][{\vartuple[][][k+1]}] \right).
    \end{align}
    Summing up for $k  = 0, \hdots, \numhorizons$:
    \begin{align}
        \sum_{k = 0}^{\numhorizons} \divergence[\kernel][{\vartuple[][][k+0.5]}][{\vartuple[][][k]}]  &\leq \frac{1}{1 - (\learnrate[ ][ ]\lsmooth)^2} \sum_{k = 0}^{\numhorizons} \left(\divergence[\kernel][{\vartuple[][*]}][{\vartuple[][][k]}] - \divergence[\kernel][{\vartuple[][*]}][{\vartuple[][][k+1]}] \right)\\
        &\leq \frac{1}{1 - (\learnrate[ ][ ]\lsmooth)^2} \left(\divergence[\kernel][{\vartuple[][*]}][{\vartuple[][][0]}] - \divergence[\kernel][{\vartuple[][*]}][{\vartuple[][][\numhorizons+1]}] \right)\\
        &\leq \frac{1}{1 - (\learnrate[ ][ ]\lsmooth)^2} \divergence[\kernel][{\vartuple[][*]}][{\vartuple[][][0]}].
    \end{align}
    Dividing both sides by $\numhorizons$, we have:
    \begin{align}
        \frac{1}{\numhorizons}\sum_{k = 0}^{\numhorizons} \divergence[\kernel][{\vartuple[][][k+0.5]}][{\vartuple[][][k]}] &\leq \frac{1}{\numhorizons\left(1 - (\learnrate[ ][ ]\lsmooth)^2\right)} \left(\divergence[\kernel][{\vartuple[][*]}][{\vartuple[][][0]}] \right)\\
        \min_{k = 0, \hdots, \numhorizons} \divergence[\kernel][{\vartuple[][][k+0.5]}][{\vartuple[][][k]}]  &\leq \frac{1}{\numhorizons\left(1 - (\learnrate[ ][ ]\lsmooth)^2\right)} \left(\divergence[\kernel][{\vartuple[][*]}][{\vartuple[][][0]}] \right).\label{eq:intermediate_progress_bound}
    \end{align}

    We can transform this convergence into a convergence in terms of the primal gap function\sklara{}{, i.e., the quantity $\max_{\vartuple \in \set} \langle \vioper(\vartuple[][][k + 0.5]), \vartuple[][][k+0.5] - \vartuple \rangle$ that measures the worst-case violation of the strong solution condition}. Recall by the first order optimality conditions of $\vartuple[][][k+0.5]$, we have for all $\vartuple \in \set$:
    \begin{align*}
    \left\langle \vioper(\vartuple[][][k]) + \frac{1}{\learnrate[ ][ ]} \left( \grad \kernel(\vartuple[][][k+0.5]) - \grad \kernel(\vartuple[][][k]) \right), \vartuple - \vartuple[][][k+0.5] \right\rangle \geq 0.
    \end{align*}
Re-organizing, for all $\vartuple \in \set$, and $k \in [\numhorizons]$ we have:
    \begin{align}
         \langle \vioper(\vartuple[][][k]), \vartuple[][][k+0.5] - \vartuple \rangle &\leq \frac{1}{\learnrate[ ][ ]} \left\| \grad \kernel(\vartuple[][][k+0.5]) - \grad \kernel(\vartuple[][][k])\right\| \left\|\vartuple[][][k+0.5] - \vartuple \right\|\\
         &\leq \frac{\diam(\set)}{\learnrate[ ][ ]} \left\| \grad \kernel(\vartuple[][][k+0.5]) - \grad \kernel(\vartuple[][][k])\right\|\\
         &\leq \frac{\diam(\set)\kernelsmooth}{\learnrate[ ][ ]} \left\| \vartuple[][][k+0.5] - \vartuple[][][k]\right\|
    \end{align}
    where the last line follow from $h$ being $\kernelsmooth$-Lipschitz-smooth.

Now, with the above inequality in hand, notice that for all $\vartuple \in \set$ and $k \in [\numhorizons]$, we have:
\begin{align*}
    \langle \vioper(\vartuple[][][k + 0.5]), \vartuple[][][k+0.5] - \vartuple \rangle &= \langle \vioper(\vartuple[][][k]), \vartuple[][][k+0.5] - \vartuple \rangle + \langle \vioper(\vartuple[][][k+0.5]) - \vioper(\vartuple[][][k]), \vartuple[][][k+0.5] - \vartuple \rangle \\
    &\leq \frac{\diam(\set)\kernelsmooth}{\learnrate[ ][ ]} \|\vartuple[][][k+0.5] - \vartuple[][][k]\| + \|\vioper(\vartuple[][][k+0.5]) - \vioper(\vartuple[][][k])\|  \|\vartuple[][][k+0.5] - \vartuple\|\\
    &\leq \frac{\diam(\set)\kernelsmooth}{\learnrate[ ][ ]} \|\vartuple[][][k+0.5] - \vartuple[][][k]\| + \lsmooth \|\vartuple[][][k+0.5] - \vartuple[][][k]\|  \|\vartuple[][][k+0.5] - \vartuple\|\\
    &\leq \diam(\set) \left( \frac{\kernelsmooth}{\learnrate[ ][ ]} + \lsmooth \right) \|\vartuple[][][k+0.5] - \vartuple[][][k]\|
\end{align*}
where the penultimate line follows from the $\lsmooth$-Lipschitz continuity of $\vioper$, and the strong convexity of $\kernel$, which means that we have $\forall \vartuple, \othervartuple \in \set$, $\divergence[\kernel][{\vartuple}][{\othervartuple}] \geq \nicefrac{1}{2}\| \vartuple - \othervartuple\|^2$.

Now, letting $k^* \in \argmin_{k = 0, \hdots, \numhorizons} \|\vartuple[][][k+0.5] - \vartuple[][][k]\|$, we have:
\begin{align*}
    \langle \vioper(\vartuple[][][k^* + 0.5]), \vartuple[][][k^*+0.5] - \vartuple \rangle &\leq  \diam(\set) \left( \frac{\kernelsmooth}{\learnrate[ ][ ]} + \lsmooth \right) \|\vartuple[][][k^*+0.5] - \vartuple[][][k^*]\|\\
    &=  \diam(\set) \left( \frac{\kernelsmooth}{\learnrate[ ][ ]} + \lsmooth \right) \min_{k= 0, \hdots, \numhorizons} \|\vartuple[][][k^*+0.5] - \vartuple[][][k^*]\|\\
    &\leq \diam(\set) \left( \frac{\kernelsmooth}{\learnrate[ ][ ]} + \lsmooth \right) \min_{k= 0, \hdots, \numhorizons} \sqrt{2 \divergence[\kernel][{\vartuple[][][k+0.5]}][{\vartuple[][][k]}]},
\end{align*}
where the last line follows from \sklara{}{$\kernel$ being $1$-strongly-convex.}

Now, plugging \Cref{eq:intermediate_progress_bound} in the above, we have:
\begin{align*}
    \langle \vioper(\vartuple[][][k^* + 0.5]), \vartuple[][][k^* +0.5] - \vartuple \rangle &\leq  \diam(\set) \left( \frac{\kernelsmooth}{\learnrate[ ][ ]} + \lsmooth \right) \min_{k = 0, \hdots, \numhorizons} \sqrt{2\divergence[\kernel][{\vartuple[][][k+0.5]}][{\vartuple[][][k]}]}\\
    &=  \sqrt{2}\diam(\set) \left( \frac{\kernelsmooth}{\learnrate[ ][ ]} + \lsmooth \right)  \sqrt{\min_{k = 0, \hdots, \numhorizons} \divergence[\kernel][{\vartuple[][][k+0.5]}][{\vartuple[][][k]}]}\\
    &\leq   \frac{\sqrt{2} \diam(\set) \left( \frac{\kernelsmooth}{\learnrate[ ][ ]} + \lsmooth \right)}{\sqrt{1 - (\learnrate[ ][ ]\lsmooth)^2}} \frac{\sqrt{\divergence[\kernel][{\vartuple[][*]}][{\vartuple[][][0]}]}}{\sqrt{\numhorizons}}.
\end{align*}

Now, by the assumption that $\learnrate[ ][ ] \leq \frac{1}{\sqrt{2}\lsmooth} < \frac{1}{\lsmooth}$, we have:
\begin{align*}
    \langle \vioper(\vartuple[][][k^* + 0.5]), \vartuple[][][k^* +0.5] - \vartuple \rangle &\leq   \frac{\sqrt{2} \diam(\set) \left( \frac{\kernelsmooth}{\learnrate[ ][ ]} + \frac{1}{\learnrate[ ][ ]} \right)}{\sqrt{1 - (\learnrate[ ][ ]\lsmooth)^2}} \frac{\sqrt{\divergence[\kernel][{\vartuple[][*]}][{\vartuple[][][0]}]}}{\sqrt{\numhorizons}}\\
    &= \frac{(1 + \kernelsmooth)\sqrt{2} \diam(\set)}{\learnrate[ ][ ]\sqrt{1 - (\learnrate[ ][ ]\lsmooth)^2}} \frac{\sqrt{\divergence[\kernel][{\vartuple[][*]}][{\vartuple[][][0]}]}}{\sqrt{\numhorizons}}\\
    &\leq   \frac{(1 + \kernelsmooth)\sqrt{2} \diam(\set)}{\learnrate[ ][ ]\sqrt{1 - (\nicefrac{1}{\sqrt{2}})^2}} \frac{\sqrt{\divergence[\kernel][{\vartuple[][*]}][{\vartuple[][][0]}]}}{\sqrt{\numhorizons}}\\
    &= \frac{2(1 + \kernelsmooth)  \diam(\set)}{\learnrate[ ][ ]} \frac{\sqrt{\divergence[\kernel][{\vartuple[][*]}][{\vartuple[][][0]}]}}{\sqrt{\numhorizons}}\\
    &= \frac{\sqrt{2} (1 + \kernelsmooth) \diam(\set)}{\learnrate[ ][ ]} \frac{\delta}{\sqrt{\numhorizons}}
\end{align*}
where the last line follows from the assumption that $\sqrt{2\divergence[\kernel][{\vartuple[][*]}][{\vartuple[][][0]}]} \leq \delta - \learnrate[ ][ ] \lipschitz$, which implies $\sqrt{\divergence[\kernel][{\vartuple[][*]}][{\vartuple[][][0]}]} \leq \frac{\delta}{\sqrt{2}}$.

That is, we have:
\begin{align*}
    \min_{k = 0, \hdots, \numhorizons} \max_{\vartuple \in \set} \langle \vioper(\vartuple[][][k+0.5]), \vartuple[][][k+0.5] - \vartuple \rangle
    &\leq \max_{\vartuple \in \set}\langle \vioper(\vartuple[][][k^* + 0.5]), \vartuple[][][k^* +0.5] - \vartuple \rangle \\
    &\leq \frac{\sqrt{2} (1 + \kernelsmooth) \diam(\set)}{\learnrate[ ][ ]} \frac{\delta}{\sqrt{\numhorizons}}.
\end{align*}

In addition, for any $\varepsilon > 0$, letting $\frac{\sqrt{2} (1 + \kernelsmooth) \diam(\set)}{\learnrate[ ][ ]} \frac{\delta}{\sqrt{\numhorizons}} \leq \varepsilon$, and solving for $\numhorizons$, we have:
\begin{align*}
    \frac{\sqrt{2} (1 + \kernelsmooth) \diam(\set)}{\learnrate[ ][ ]} \frac{\delta}{\sqrt{\numhorizons}} &\leq \varepsilon\\
    \frac{2 (1 + \kernelsmooth)^2 \diam(\set)^2}{\learnrate[ ][ ]^2} \frac{\delta^2}{\varepsilon^2} &\leq \numhorizons
\end{align*}

That is, $\bestiter[\vartuple][\numhorizons] \in \argmin_{\vartuple[][][k+0.5] : k = 0, \hdots, \numhorizons} \|\vartuple[][][k+0.5] - \vartuple[][][k]\|$ is a $\varepsilon$-strong solution after $\frac{2 (1 + \kernelsmooth)^2 \diam(\set)^2}{\learnrate[ ][ ]^2} \frac{\delta}{\varepsilon^2}$ iterations of the mirror extragradient algorithm.
\sklara{}{
Finally, notice that we have
\begin{align*}
\lim_{k \to \infty} \max_{\vartuple \in \set} \langle \vioper(\vartuple[][][k + 0.5]), \vartuple[][][k + 0.5] - \vartuple \rangle = \lim_{\numhorizons \to \infty} \min_{k = 0, \hdots, \numhorizons} \max_{\vartuple \in \set} \langle \vioper(\vartuple[][][k + 0.5]), \vartuple[][][k + 0.5] - \vartuple \rangle = 0
\end{align*}
and
\begin{align*}
\lim_{k \to \infty} \divergence[\kernel][{\vartuple[][][k+0.5]}][{\vartuple[][][k]}]= \lim_{\numhorizons \to \infty} \min_{k = 0, \hdots, \numhorizons} \divergence[\kernel][{\vartuple[][][k+0.5]}][{\vartuple[][][k]}] = 0.
\end{align*}
Hence, $\vartuple[][**] = \lim_{\numhorizon \to \infty} \vartuple[][][\numhorizon + 0.5] = \lim_{\numhorizon \to \infty} \vartuple[][][\numhorizon]$ is a strong solution of the VI $(\set, \vioper)$.
}
\end{proof}

\section{Omitted Proofs for Section \ref{section:walrasian_economies}}
\label{sec_app:walrasian}
\subsection{Walrasian Economies and Variational Inequalities}
\thmweequalsvi*
\begin{proof}[Proof of \Cref{thm:we_equal_svi}]
    $(\implies)$
    Let $\price[][][*] \in \we(\numgoods, \excessset)$ be a Walrasian equilibrium. Then, for some $\excess(\price[][][*]) \in \excessset(\price[][][*])$, we have:
    \begin{align*}
        &\innerprod[{\excess(\price[][][*])}][{\price - \price[][][*]}]\\
        &= \innerprod[{\excess(\price[][][*])}][{\price}] - \underbrace{\innerprod[{\excess(\price[][][*])}][{\price[][][*] }]}_{= 0} && \forall \price \in \R^\numgoods_+\\
        &= \underbrace{\innerprod[{\excess(\price[][][*])}][{\price}]}_{\leq 0} && \forall \price \in \R^\numgoods_+\\
        &\leq 0 && \forall \price \in \R^\numgoods_+.
    \end{align*}
    where the last line follows from the feasibility of $\excess(\price[][][*])$, i.e., $\excess(\price[][][*]) \leq 0$, and the positivity of \sklara{}{each entry of} $\price$.

    $(\impliedby)$
    Let $\price[][][*] \in \svi(\R^\numgoods_+, -\excessset)$. Then, for some $\excess(\price[][][*]) \in \excessset(\price[][][*])$, we have:
    \begin{align*}
        0&\geq  \innerprod[{\excess(\price[][][*])}][{\price - \price[][][*]}] && \forall \price \in \R^\numgoods_+.
    \end{align*}
    Substituting $\price \doteq \price[][][*] + \basis[\good]$, we have:
    \begin{align*}
        0 &\geq \innerprod[{\excess(\price[][][*])}][{\price[][][*] + \basis[\good] - \price[][][*]}]\\
        &= \innerprod[{\excess(\price[][][*])}][{\basis[\good]}]\\
        &\geq \excess[\good](\price[][][*]) && \forall \good \in \goods.
    \end{align*}
    That is, $\price[][][*]$ is feasible.
    
    Similarly, substituting in $\price \doteq \zeros$ and $\price \doteq 2\price[][][*]$, we have:
    \begin{align*}
        0&\leq  \innerprod[{\excess(\price[][][*])}][{\price[][][*]}] 
    \end{align*}
    and 
    \begin{align*}
        0&\geq  \innerprod[{\excess(\price[][][*])}][{\price[][][*]}].
    \end{align*}

    That is, $\price[][][*]$ satisfies Walras' law. 
    
    Hence, $\price[][][*]$ is a Walrasian equilibrium.
\end{proof}

\thmwebalancedequalsvi*
\begin{proof}[Proof of \Cref{thm:we_balanced_equal_svi}]
    $(\implies)$
    Let $\price[][][*] \in \we(\numgoods, \excessset)$ be a Walrasian equilibrium. Let $\alpha \doteq \frac{1}{\max\{1, \|\price[][][*]\|_\infty \}}$ so that $\alpha \price[][][*] \in [0, 1]^\numgoods$. Now, for some $\excess(\alpha \price[][][*]) \in \excessset(\alpha \price[][][*])$, we have:
    \begin{align*}
        &\innerprod[{-\excess(\alpha\price[][][*])}][{\alpha\price[][][*] - \price}]\\
        &= \innerprod[{\excess(\price[][][*])}][{\price - \alpha\price[][][*]}] && \forall \price \in [0, 1]^\numgoods && \text{(Homogeneity of $\excess$)}\\
        &= \innerprod[{\excess(\price[][][*])}][{\price}] - \alpha \underbrace{\innerprod[{\excess(\price[][][*])}][{\price[][][*] }]}_{= 0} && \forall \price \in [0, 1]^\numgoods\\
        &= \innerprod[{\excess(\price[][][*])}][{\price}]  && \forall \price \in [0, 1]^\numgoods\\
        &\leq 0 && \forall \price \in [0, 1]^\numgoods.
    \end{align*}
    where the penultimate line follows from Walras' law holding at a Walrasian equilibrium, and the last line follows from the feasibility of $\excess(\price[][][*])$, i.e., $\excess(\price[][][*]) \leq \zeros$ and the positivity of \sklara{}{each entry of} $\price$. Hence, $\alpha \price[][][*]$ is a strong solution of the box VI $([0, 1]^\numgoods, - \excessset)$, which means that $\price[][][*] \in \nicefrac{1}{\alpha} \svi([0, 1]^\numgoods, - \excessset)$.
    
    Now, notice that by the homogeneity of the excess demand in balanced economies---i.e., for all $\lambda > 0$, $\excessset(\lambda \price[][][*]) = \excessset(\price[][][*])$---if $\price[][][*]$ is a Walrasian equilibrium, then so is $\lambda \price[][][*]$. Hence, $\alpha$ takes values in $(0, 1]$, implying $\nicefrac{1}{\alpha} \in [1, \infty)$. As such, we must have $\we(\numgoods, \excessset) \subseteq  \bigcup_{\lambda \geq 1} \lambda \svi([0, 1]^\numgoods, -\excessset)$.
    
    $(\impliedby)$
    Let $\price[][][*] \in \svi([0, 1]^\numgoods, -\excessset)$ and $\lambda \geq 1$. Then, for some $\excess(\price[][][*]) \in \excessset(\price[][][*])$, we have:
    \begin{align}
        0&\geq \innerprod[{-\excess(\price[][][*])}][{\price[][][*] - \price}] && \forall \price \in [0, 1]^\numgoods \notag\\
        &= \innerprod[{\excess(\price[][][*])}][{\price - \price[][][*]}] && \forall \price \in [0, 1]^\numgoods \notag\\
        &= \innerprod[{\excess(\price[][][*])}][{\price}] - \innerprod[{\excess(\price[][][*])}][{\price[][][*] }] && \forall \price \in [0, 1]^\numgoods. \label{eq:we_svi_eq_left}
    \end{align}
    Plugging $\price = \zeros[\numgoods]$ into \Cref{eq:we_svi_eq_left}, we then have: 
    \begin{align*}
        0&\geq \underbrace{\innerprod[{\excess(\price[][][*])}][{\zeros[\numgoods]}]}_{= 0} - \innerprod[{\excess(\price[][][*])}][{\price[][][*] }]\\
        0&\geq -  \innerprod[{\excess(\price[][][*])}][{\price[][][*]}]\\
        0&\leq \innerprod[{\excess(\price[][][*])}][{\price[][][*]}]\\
        0&\leq \innerprod[{\excess(\lambda \price[][][*])}][{\price[][][*]}] && \text{(Homogeneity of $\excess$)}\\
        0&\leq \innerprod[{\excess(\lambda \price[][][*])}][{\lambda \price[][][*]}].
    \end{align*}
    Further, since $(\numgoods, \excessset)$ is balanced, we have $\samy{}{\lambda \price[][][*] \cdot \excess(\lambda \price[][][*]) = } \lambda \price[][][*] \cdot \lambda \excess(\price[][][*]) \leq \price[][][*] \cdot \excess(\price[][][*]) \leq 0$, hence, combining it with the above inequality, $\lambda \price[][][*] \cdot \excess(\lambda \price[][][*]) = 0$, meaning that $\lambda \price[][][*]$ satisfies Walras' law.

    In addition, continuing from \Cref{eq:we_svi_eq_left} again, we have:
    \begin{align}
        0&\geq \innerprod[{\excess(\price[][][*])}][{\price}] - \underbrace{\innerprod[{\excess(\price[][][*])}][{\price[][][*] }]}_{= 0} && \forall \price \in [0, 1]^\numgoods\notag\\
        &= \innerprod[{\excess(\price[][][*])}][{\price}] && \forall \price \in [0, 1]^\numgoods \notag\\
        &= \innerprod[{\excess(\lambda \price[][][*])}][{\price}] && \forall \price \in [0, 1]^\numgoods.\label{eq:weak_walras_law_application}
    \end{align}
     where the \sklara{penultimate line follows from the fact that balanced economies satisfy weak Walras' law, and the}{} last line \sklara{}{follows} from the homogeneity of degree $0$ of the excess demand.

    Now, plugging $\price = \basis[\good]$ into \Cref{eq:weak_walras_law_application}, we have:
    \begin{align*}
        0 &\geq \innerprod[{\excess(\lambda \price[][][*])}][{\basis[\good]}] && \forall \good \in \goods\\
        &\geq \excess[\good](\lambda \price[][][*]) && \forall \good \in \goods.
    \end{align*}

   That is, $\lambda \price[][][*]$ is feasible. Putting it all together, $\lambda \price[][][*]$ is a Walrasian equilibrium. As such, $ \bigcup_{\lambda \geq 1} \lambda \svi([0, 1]^\numgoods, -\excessset) \subseteq \we(\numgoods, \excessset)$.
\end{proof}

\lemmaapproxsvieqapproxwe*
\begin{proof}[Proof of \Cref{lemma:approx_svi_eq_approx_we}]
    For any $\varepsilon \geq 0$, let $\price[][][*] \in \svi[\varepsilon]([0, 1]^\numgoods, -\excessset)$. Then, for some $\excess(\price[][][*]) \in \excessset(\price[][][*])$, we have:
    \begin{align}
        \varepsilon &\geq \innerprod[{-\excess(\price[][][*])}][{\price[][][*] - \price}] && \forall \price \in [0, 1]^\numgoods \notag\\
        &= \innerprod[{\excess(\price[][][*])}][{\price - \price[][][*]}] && \forall \price \in [0, 1]^\numgoods \notag\\
        &= \innerprod[{\excess(\price[][][*])}][{\price}] - \innerprod[{\excess(\price[][][*])}][{\price[][][*] }] && \forall \price \in [0, 1]^\numgoods. \label{eq:we_svi_eq_left1}
    \end{align}
    Plugging $\price = \zeros[\numgoods]$ into \Cref{eq:we_svi_eq_left1}, we then have: 
    \begin{align}
        \varepsilon &\geq \underbrace{\innerprod[{\excess(\price[][][*])}][{\zeros[\numgoods]}]}_{= 0} - \innerprod[{\excess(\price[][][*])}][{\price[][][*] }] \notag\\
        \varepsilon &\geq - \innerprod[{\excess(\price[][][*])}][{\price[][][*]}] \notag\\
        -\varepsilon &\leq \innerprod[{\excess(\price[][][*])}][{\price[][][*]}]\label{eq:i_need_a_name}
    \end{align}
    Further, since $(\numgoods, \excessset)$ is balanced, it follows that $\price[][][*] \cdot \excess(\price[][][*]) \leq 0 \leq \varepsilon$. Combining this conclusion with \Cref{eq:i_need_a_name}, we see that $\price[][][*]$ satisfies $\varepsilon$-Walras' law.

    In addition, continuing from \Cref{eq:we_svi_eq_left1} again, we have:
    \begin{align}
        \varepsilon &\geq \innerprod[{\excess(\price[][][*])}][{\price}] - \underbrace{\innerprod[{\excess(\price[][][*])}][{\price[][][*] }]}_{\leq 0} && \forall \price \in [0, 1]^\numgoods \notag\\
        &\geq \innerprod[{\excess(\price[][][*])}][{\price}] && \forall \price \in [0, 1]^\numgoods, \label{eq:weak_walras_law_application2}
    \end{align}
     where the last line follows from the fact that balanced economies satisfy weak Walras' law.

    Now, plugging $\price = \basis[\good]$ into \Cref{eq:weak_walras_law_application2}, we have:
    \begin{align*}
        \varepsilon &\geq \innerprod[{\excess(\price[][][*])}][{\basis[\good]}] && \forall \good \in \goods \notag\\
        &\geq \excess[\good](\price[][][*]) && \forall \good \in \goods \notag .
    \end{align*}

   That is, $\price[][][*]$ is $\varepsilon$-feasible. Putting it all together, $\price[][][*]$ is an $\varepsilon$-Walrasian equilibrium.

\end{proof}

\thmwecompequalsvi*
\begin{proof}[Proof of \Cref{thm:we_comp_equal_svi}]
    $(\implies)$
    Let $\price[][][*] \in \we(\numgoods, \excessset)$ be a Walrasian equilibrium. Let $\alpha \doteq \frac{1}{\|\price[][][*]\|_1}$. Then, we have $\alpha \price[][][*] \in \simplex[\numgoods]$. Further, for some $\excess(\alpha \price[][][*]) \in \excessset(\alpha \price[][][*])$, we have:
    \begin{align*}
        &\innerprod[{-\excess(\alpha\price[][][*])}][{\alpha\price[][][*] - \price}]\\
        &= \innerprod[{\excess(\price[][][*])}][{\price - \alpha\price[][][*]}] && \forall \price \in \simplex[\numgoods] && \text{(Homogeneity of $\excess$)}\\
        &= \innerprod[{\excess(\price[][][*])}][{\price}] - \alpha \underbrace{\innerprod[{\excess(\price[][][*])}][{\price[][][*] }]}_{= 0} && \forall \price \in \simplex[\numgoods]\\
        &= \innerprod[{\excess(\price[][][*])}][{\price}]  && \forall \price \in \simplex[\numgoods]\\
        &\leq 0 && \forall \price \in \simplex[\numgoods].
    \end{align*}
    where the penultimate line follows from Walras' law holding at a Walrasian equilibrium, and the last line follows from the feasibility of $\excess(\price[][][*])$, i.e., $\excess(\price[][][*]) \leq \zeros$, and the positivity of \sklara{}{each entry of} $\price$. Hence, $\alpha \price[][][*]$ is a strong solution of the simplex VI $(\simplex[\numgoods], - \excessset)$, which means that $\price[][][*] \in \frac{1}{\alpha} \svi(\simplex[\numgoods], - \excessset)$.
    
    Now, notice that, by homogeneity of the excess demand in competitive economies, since for all $\lambda > 0$, $\excessset(\lambda \price[][][*]) = \excessset(\price[][][*])$ holds, if $\price[][][*]$ is a Walrasian equilibrium, then so is $\lambda \price[][][*]$. Hence, $\alpha$ takes values in $(0, \infty)$, implying $\frac{1}{\alpha} \in (0, \infty)$, and as such we must have $\we(\numgoods, \excessset) \subseteq  \bigcup_{\lambda > 0} \lambda \svi(\simplex[\numgoods], -\excessset)$.
    
    $(\impliedby)$
    Let $\price[][][*] \in \svi(\simplex[\numgoods], -\excessset)$ and $\lambda > 0$. Then, for some $\excess(\price[][][*]) \in \excessset(\price[][][*])$, we have:
    \begin{align*}
        0&\geq \innerprod[{-\excess(\price[][][*])}][{\price[][][*] - \price}] && \forall \price \in \simplex[\numgoods] \notag\\
        &= \innerprod[{\excess(\price[][][*])}][{\price - \price[][][*]}] && \forall \price \in \simplex[\numgoods] \notag\\
        &= \innerprod[{\excess(\price[][][*])}][{\price}] - \underbrace{\innerprod[{\excess(\price[][][*])}][{\price[][][*] }]}_{\leq 0} && \forall \price \in \simplex[\numgoods] \notag\\
        &\geq \innerprod[{\excess(\price[][][*])}][{\price}] && \forall \price \in \simplex[\numgoods]\\
        &\geq \innerprod[{\excess(\lambda \price[][][*])}][{\price}] && \forall \price \in \simplex[\numgoods].
    \end{align*}
     where the penultimate line follows from the fact that competitive economies satisfy weak Walras' law, and the last line from homogeneity of degree $0$ of the excess demand.

    Now, plugging $\price = \basis[\good]$ for all $\good \in \goods$ in the above, we have:
    \begin{align*}
        0 &\geq \innerprod[{\excess(\lambda \price[][][*])}][{\basis[\good]}] && \forall \good \in \goods \notag\\
        &\geq \excess[\good](\lambda \price[][][*]) && \forall \good \in \goods \notag .
    \end{align*}

   That is, $\lambda \price[][][*]$ is feasible. Now by non-satiation, since $\excess(\lambda \price[][][*]) \leq \zeros[\numgoods]$, we\sklara{must have $\lambda \price[][][*] \cdot \excess(\lambda \price[][][*]) \geq 0$. As by weak Walras' law $\lambda \price[][][*] \cdot \excess(\lambda \price[][][*]) \leq 0$, we must}{} have $\lambda \price[][][*] \cdot \excess(\lambda \price[][][*]) = 0$, meaning that $\lambda \price[][][*]$ satisfies Walras' law. Putting it all together, $\lambda \price[][][*]$ is a Walrasian equilibrium. As such we have $ \bigcup_{\lambda > 0} \lambda \svi(\simplex[\numgoods], -\excessset) \subseteq \we(\numgoods, \excessset)$.
\end{proof}

\thmexistencewe*
\begin{proof}[Proof of \Cref{thm:existence_we}]
    By \Cref{thm:we_comp_equal_svi}, we know that the set of strong solutions $\svi(\simplex[\numgoods], -\excessset)$ of the simplex VI $(\simplex[\numgoods], -\excessset)$ is a subset of the set of Walrasian equilibria \samy{}{of a competitive economy}.

    Now, notice that for a continuous \samy{Walrasian}{competitive} economy $(\numgoods, \excessset)$, the corresponding \sklara{}{simplex} VI $(\simplex[\numgoods], -\excessset)$ is continuous. \sklara{}{This is because $\simplex[\numgoods]$ is non-empty, compact, and convex, and $-\excessset$ is upper-hemicontinuous, non-empty-, compact-, and convex-valued by assumption.} Hence, by Theorem 2.2.1 of \citet{facchinei2003finite}, a strong solution to $(\simplex[\numgoods], -\excessset)$ is guaranteed to exist, which in turn implies the existence of a Walrasian equilibrium in continuous competitive economies.
\end{proof}

\subsection{Results for Balanced Economies}

\lemmabalancedisminty*
\begin{proof}[Proof of \Cref{lemma:balanced_is_minty}]
    Let $(\numgoods, \excessset)$ be a balanced economy.
    Setting $\price[][][*] \doteq \zeros[\numgoods]$, we have for all $\price \in [0, 1]^\numgoods$:
    \begin{align*}
        \innerprod[{ \excess(\price)}][{\price[][][*] - \price}] &= \innerprod[{ \excess(\price)}][{\zeros[\numgoods] - \price}]\\
        &= \underbrace{\innerprod[{ \excess(\price)}][{\zeros[\numgoods]}]}_{= 0} - \innerprod[{ \excess(\price)}][{\price}]\\
        &= - \underbrace{\innerprod[{ \excess(\price)}][{\price}]}_{\leq 0}\\
        &\geq 0 \enspace ,
    \end{align*}
    where the last line follow from weak Walras' law, which is assumed to hold in balanced economies. 
\end{proof}

\amy{there is a bug here, but i cannot figure out why latex is complaining}

\thmmirrorextratatonnconvergence*
\begin{proof}[Proof of \Cref{thm:mirror_extra_tatonn_convergence}]
    Since $(\numgoods, \excess)$ is a balanced economy, by \Cref{lemma:balanced_is_minty}, $(\numgoods, \excess)$ is variationally stable on $[0, 1]^\numgoods$, and hence the box VI $([0, 1]^\numgoods, -\excess)$ satisfies the Minty condition. Hence, as the mirror \emph{extrat\^atonnement} process is simply the mirror extragradient method run on the box VI $([0, 1]^\numgoods, -\excess)$, the assumptions of \Cref{thm:mirror_extragradient_global_convergence} are satisfied, and we obtain the result. 

\end{proof}

\lemmabregmancontelastic*
\begin{proof}[Proof of \Cref{lemma:bregman_cont_elastic}]
    \klara{I've commented out the old proof beginning since it was long.}
    By the $\lelastic$-elasticity assumption, we have, for all distinct $\price, \otherprice \in \pricespace$,
    \begin{align*}
        \frac{\norm[\demandfunc(\otherprice)-\demandfunc(\price)]}{\norm[\demandfunc(\price)]}
        &\leq \lelastic \frac{\norm[\otherprice-\price]}{\norm[\price]_\infty}.
    \end{align*}    

    Isolating the change in demand and exploiting the $1$-convexity of $\kernel$, we obtain:
    \begin{align*}
        \norm[\demandfunc(\otherprice)-\demandfunc(\price)]
        &\leq \frac{\lelastic \norm[\demandfunc(\price)]}{\norm[\price]_\infty}\norm[\otherprice-\price] \\
        &\leq \frac{\lelastic \norm[\demandfunc(\price)]}{\norm[\price]_\infty}\sqrt{2\divergence[\kernel](\otherprice,\price)}.
    \end{align*}

    By a similar argument, we also have:
    \begin{align*}
        \left\|\supplyfunc(\otherprice) - \supplyfunc(\price)  \right\| &\leq \frac{\lelastic \|\supplyfunc(\price)\|}{\|\price\|_\infty} \sqrt{2\divergence[\kernel](\otherprice,\price)}.
    \end{align*}

    Combining the two bounds, we then have:
    \begin{align*}
        \| \excess(\otherprice) - \excess(\price) \| &=  \| \demandfunc(\otherprice)  - \supplyfunc(\otherprice) - \demandfunc(\price) + \supplyfunc(\price) \|\\
        &\leq \|\demandfunc(\otherprice) - \demandfunc(\price) \| + \| \supplyfunc(\otherprice) - \supplyfunc(\price) \|\\
        &\leq \frac{\lelastic \|\demandfunc(\price)\|}{\|\price\|_\infty} \sqrt{2 \divergence[\kernel](\otherprice,\price)}  + \frac{\lelastic \|\supplyfunc(\price)\|}{\|\price\|_\infty} \sqrt{2 \divergence[\kernel](\otherprice,\price)}\\
        &\leq \frac{\lelastic \left( \|\demandfunc(\price)\| + \|\supplyfunc(\price)\| \right)}{\|\price\|_\infty} \sqrt{2 \divergence[\kernel](\otherprice,\price)}.
    \end{align*}
    
    Squaring both sides and reorganizing yields:
    \begin{align*}
        \frac{1}{2}\| \excess(\otherprice) - \excess(\price) \|^2 \leq \left(\frac{\lelastic \left(\|\demandfunc(\price)\| + \|\supplyfunc(\price)\| \right)}{\|\price\|_\infty} \right)^2 \divergence[\kernel](\otherprice,\price).
    \end{align*}
\end{proof}

\subsection{Results for Variationally Stable Economies}

\thmmirrorextratatonnvarstable*
\begin{proof}[Proof of \Cref{thm:mirror_extratatonn_var_stable}]
       Since $(\numgoods, \excess)$ is variationally stable on $\simplex[\numgoods]$, the simplex VI $(\simplex[\numgoods], -\excess)$ satisfies the Minty condition. In addition, since by the assumption of the theorem the economy is $\lelastic$-elastic and $\lbounded$-bounded, by \Cref{lemma:bregman_cont_elastic}, $\excess$ is $\left(2\numgoods \lelastic \lbounded \right)$-Bregman-continuous on $\simplex[\numgoods]$. That is, we have \begin{align*}
        \frac{1}{2}\| \excess(\otherprice) - \excess(\price) \|^2 &\leq \left(\frac{\lelastic \left(\|\demandfunc(\price)\| + \|\supplyfunc(\price)\| \right)}{\|\price\|_\infty} \right)^2 \divergence[\kernel](\otherprice,\price)\\
        &\leq \max_{\price \in \simplex[\numgoods]} \left\{\left(\frac{\lelastic \left(\|\demandfunc(\price)\| + \|\supplyfunc(\price)\| \right)}{\|\price\|_\infty} \right)^2\right\} \divergence[\kernel](\otherprice,\price)\\
        &\leq  \left(\frac{\lelastic \left(\|\demandfunc\|_\infty + \|\supplyfunc\|_\infty \right)}{\min_{\price \in \simplex[\numgoods]} \|\price\|_\infty} \right)^2 \divergence[\kernel](\otherprice,\price)\\
        &\leq \left(\frac{2\lelastic \lbounded }{\frac{1}{\numgoods}} \right)^2 \divergence[\kernel](\otherprice,\price)\\
        &\leq \left(2\numgoods \lelastic \lbounded  \right)^2 \divergence[\kernel](\otherprice,\price) . 
       \end{align*}

        Suppose that under the assumptions of the theorem the mirror generates the sequence of prices $\left\{\price[][\numhorizon], \price[][\numhorizon + 0.5] \right\}_{\numhorizon}$. Let $\bestiter[{\price}][\numhorizons] \in \argmin_{\vartuple[][][k+0.5] : k = 0, \hdots, \numhorizons} \divergence[\kernel](\price[][k+0.5], \price[][k])$. As the mirror \emph{extrat\^atonnement} process is simply the mirror extragradient method run on the simplex VI $(\simplex[\numgoods], -\excess)$, and the assumptions of \Cref{thm:mirror_extragradient_global_convergence} are satisfied and hence we have the following bound:
        \begin{align*}
             \min_{k = 0, \hdots, \numhorizons} \max_{\price \in \simplex} \langle -\excess(\price[][k+0.5]),  \price[][k+0.5] - \price \rangle &\leq  \frac{2 (1 + \kernelsmooth)\diam(\simplex[\numgoods])}{\learnrate[ ][ ]} \frac{\sqrt{\divergence[\kernel][{\price[*][]}][{\price[][0]}]}}{\sqrt{\numhorizons}}\\
             \min_{k = 0, \hdots, \numhorizons} \max_{\price \in \simplex} \langle \excess(\price[][k+0.5]),  \price - \price[][k+0.5] \rangle &\leq  \frac{2 \sqrt{2}(1 + \kernelsmooth)}{\learnrate[ ][ ]} \frac{\sqrt{\divergence[\kernel][{\price[*]}][{\price[][0]}]}}{\sqrt{\numhorizons}}.
        \end{align*}

        Further, by \Cref{lemma:approx_svi_eq_approx_we}, \sklara{}{we have that $\lim_{\numhorizon} \price[][\numhorizon] = \lim_{\numhorizon \to \infty} \price[][][(\numhorizon + 0.5)] = \price[][][**]$ is a Walrasian equilibrium.}

\end{proof}

\subsection{The Scarf Economy}

\klara{Commented out much of the previous stuff because the strike-throughs were getting illegible.}

The following lemma states that any Scarf economy is a balanced economy.
\begin{lemma}[Scarf Economies are Balanced]
\label{lemma:Scarf_balanced}
    The Scarf economy $(3, \excess^{\mathrm{scarf}})$ is a balanced economy that satisfies Walras' law, i.e., for all $\price \in \R^\numgoods_+$, it holds that $\price \cdot \excess^{\mathrm{scarf}}(\price) = 0$. 
    Furthermore, the set of Walrasian equilibria of the Scarf economy is given by $\we(3, \excess^{\mathrm{scarf}}) \doteq \{\lambda \ones[3] \mid \lambda > 0\}$.
\end{lemma}

\begin{proof}
    First, notice that the Scarf economy is homogeneous of degree $0$. That is, for all $\lambda \geq 0$,
    \begin{align*}
        \excess^{\mathrm{scarf}}(\lambda \price) \doteq \begin{pmatrix}
            \frac{\lambda\price[1]}{\lambda\price[1] + \lambda\price[2]} + \frac{\lambda\price[3]}{\lambda\price[1] + \lambda\price[3]} - 1\\
            \frac{\lambda\price[1]}{\lambda \price[1] + \lambda \price[2]} + \frac{\lambda \price[2]}{\lambda \price[2] + \lambda \price[3]} - 1\\
            \frac{\lambda \price[2]}{\lambda \price[2] + \lambda \price[3]} + \frac{\lambda \price[3]}{\lambda \price[1] + \lambda \price[3]} - 1
            \end{pmatrix} =  \begin{pmatrix}
            \frac{\price[1]}{\price[1] + \price[2]} + \frac{\price[3]}{\price[1] + \price[3]} - 1\\
            \frac{\price[1]}{\price[1] + \price[2]} + \frac{\price[2]}{\price[2] + \price[3]} - 1\\
            \frac{\price[2]}{\price[2] + \price[3]} + \frac{\price[3]}{\price[1] + \price[3]} - 1
            \end{pmatrix} 
            = \excess^{\mathrm{scarf}}(\price).
    \end{align*}

    Second, for all $\price \in \R^\numgoods$, notice that
    \begin{align*}
        \price \cdot \excess^{\mathrm{scarf}}(\price) &= \frac{\price[1][][2]}{\price[1] + \price[2]} + \frac{\price[1]\price[3]}{\price[1] + \price[3]} - \price[1] + 
            \frac{\price[1]\price[2]}{\price[1] + \price[2]} + \frac{\price[2][][2]}{\price[2] + \price[3]} - \price[2] +
            \frac{\price[2]\price[3]}{\price[2] + \price[3]} + \frac{\price[3][][2]}{\price[1] + \price[3]} - \price[3]\\
            &= \frac{\price[1][][2] + \price[1]\price[2]}{\price[1] + \price[2]} + \frac{\price[2][][2] + \price[2]\price[3]}{\price[2] + \price[3]} + \frac{\price[3][][2] + \price[1]\price[3]}{\price[1] + \price[3]} - \price[1] - \price[2] - \price[3]\\
            &= \frac{\price[1] (\price[1] + \price[2])}{\price[1] + \price[2]} + \frac{\price[2] (\price[2] + \price[3])}{\price[2] + \price[3]} + \frac{\price[3](\price[3] + \price[1])}{\price[1] + \price[3]} - \price[1] - \price[2] - \price[3]\\
            &= 0.
    \end{align*}

    Finally, observe that for $\price[][][*] = \ones[\numgoods]$, we have that $\excess^{\mathrm{scarf}}(\price[][][*]) = \zeros[\numgoods]$, and thus, $\price[][][*] \cdot \excess^{\mathrm{scarf}}(\price[][][*]) = 0$. 
    This equilibrium is unique up to positive scaling.
\end{proof}


\begin{restatable}[\sklara{Variational stability and}{}Bregman continuity of the Scarf Economy]{lemma}{lemmaScarfvarstablebregcont}
\label{lemma:scarf_var_stable_breg_cont}
        For any $\underline{\price[ ][]} \in (0, \nicefrac{1}{3})$ and any $1$-strongly-convex kernel function $\kernel: \R^3_+ \to \R$, the Scarf economy is $(\nicefrac{3}{\underline{\price[ ][]}^2}, \kernel)$-Bregman-continuous on $[\underline{\price[ ][]}, 1]^3$.
\end{restatable}

\begin{proof}[Proof of \Cref{lemma:scarf_var_stable_breg_cont}]
    

    \textbf{Part 1: Variational instability on $[\underline{\price[ ][]},1]^3$.}

    We claim that no strong solution $\price[][][*]=a\mathbf{1}_3$ with $a>0$ of the truncated-box VI $\left( [\underline{\price[ ][]},1]^3,-\excess^{\mathrm{scarf}} \right)$ is a Minty solution. That is, there exists for each such $\price[][][*]$ a price $\price \in [\underline{\price[ ][]},1]^3$ such that $\innerprod[{ \excess^{\mathrm{scarf}}(\price)}][{\price[][][*] - \price}] < 0$. In fact, this is true for any $\price \in [\underline{\price[ ][]},1]^3$ such that $\price[1] < \price[2] < \price[3]$ since then
    \begin{align*}
        \innerprod[{ \excess^{\mathrm{scarf}}(\price)}][{\price[][][*] - \price}] &= \innerprod[{ \excess^{\mathrm{scarf}}(\price)}][{\price[][][*] }] - \underbrace{\innerprod[{ \excess^{\mathrm{scarf}}(\price)}][{\price}]}_{= 0} \tag{\Cref{lemma:Scarf_balanced}} \\
        &= \innerprod[{ \excess^{\mathrm{scarf}}(\price)}][{\price[][][*] }]\\
        &=a\left( \frac{\price[1]}{\price[1] + \price[2]} + \frac{\price[3]}{\price[1] + \price[3]}-1 \right)+a\left( \frac{\price[1]}{\price[1] + \price[2]} + \frac{\price[2]}{\price[2] + \price[3]}-1 \right)+a\left( \frac{\price[2]}{\price[2] + \price[3]} + \frac{\price[3]}{\price[1] + \price[3]}-1 \right) \\
        &=2a\frac{\price[1]}{\price[1] + \price[2]} + 2a\frac{\price[2] }{\price[2] + \price[3]} + 2a\frac{\price[3]}{\price[1] + \price[3]}-3a\\
        &=a\frac{\price[1]-\price[2]}{\price[1] + \price[2]} + a\frac{\price[2]-\price[3]}{\price[2] + \price[3]} + a\frac{\price[3]-\price[1]}{\price[1] + \price[3]}, \\
        &=a\frac{{\price[1]}^2\price[2]-{\price[1]}^2\price[3]+{\price[3]}^2\price[1]-{\price[3]}^2\price[2]+{\price[2]}^2\price[3]-{\price[2]}^2\price[1]}
        {\left(\price[1] + \price[2]\right)\left(\price[2] + \price[3]\right)\left(\price[1] + \price[3]\right)} \\
        &=a\frac{\left(\price[1]\price[2]\price[1]-\price[1]\price[2]\price[3]\right)+\left(\price[1]\price[3]\price[3]-\price[1]\price[3]\price[1]\right)+\left(\price[2]\price[3]\price[2]-\price[2]\price[3]\price[3]\right)+\left(\price[2]\price[1]\price[3]-\price[2]\price[1]\price[2]\right)}{\left(\price[1] + \price[2]\right)\left(\price[2] + \price[3]\right)\left(\price[1] + \price[3]\right)} \\
        &=a\frac{\price[1]\price[2]\left(\price[1]-\price[3]\right)+\price[1]\price[3]\left(\price[3]-\price[1]\right)+\price[2]\price[3]\left(\price[2]-\price[3]\right)+\price[2]\price[1]\left(\price[3]-\price[2]\right)}{\left(\price[1] + \price[2]\right)\left(\price[2] + \price[3]\right)\left(\price[1] + \price[3]\right)} \\
        &=a\frac{\price[1]\left(\price[2]-\price[3]\right)\left(\price[1]-\price[3]\right)+\price[2]\left(\price[2]-\price[3]\right)\left(\price[3]-\price[1]\right)}{\left(\price[1] + \price[2]\right)\left(\price[2] + \price[3]\right)\left(\price[1] + \price[3]\right)} \\
        &=a\frac{\left(\price[1]-\price[2]\right)\left(\price[2]-\price[3]\right)\left(\price[1]-\price[3]\right)}{\left(\price[1] + \price[2]\right)\left(\price[2] + \price[3]\right)\left(\price[1] + \price[3]\right)},
    \end{align*}
    and the numerator is negative. Hence, the Scarf economy is variationally unstable on $[\underline{\price[ ][]},1]^3$.

    \textbf{Part 2: Bregman continuity on $[\underline{\price[ ][]}, 1]^3$.}
    
    Notice that the excess demand is differentiable, with its Jacobian given by:
    %
    \begin{align*}
        \grad \excess(\price) = 
        \begin{bmatrix}
            -\frac{\price[1]}{(\price[1] + \price[2])^2} - \frac{\price[3]}{(\price[1] + \price[3])^2} 
            & -\frac{\price[1]}{(\price[1] + \price[2])^2} 
            & -\frac{\price[3]}{(\price[1] + \price[3])^2} \\[10pt]
            -\frac{\price[1]}{(\price[1] + \price[2])^2} 
            & -\frac{\price[1]}{(\price[1] + \price[2])^2} - \frac{\price[3]}{(\price[2] + \price[3])^2} 
            & -\frac{\price[2]}{(\price[2] + \price[3])^2} \\[10pt]
            -\frac{\price[3]}{(\price[1] + \price[3])^2} 
            & -\frac{\price[2]}{(\price[2] + \price[3])^2}
            & -\frac{\price[2]}{(\price[2] + \price[3])^2} - \frac{\price[3]}{(\price[1] + \price[3])^2}
        \end{bmatrix}.
    \end{align*}
    
    Thus, the Jacobian consists of entries of the form of $f(x,y) \doteq \frac{x}{(x+y)^2}$. 
    For $x, y \in [\underline{\price[ ][]}, 1]$, it follows that $|f(x,y)| \leq \frac{1}{4 \underline{\price[ ][]}^2}$, which means that the absolute value of the off diagonal entries of $\grad \excess(\price)$ are bounded by $\frac{1}{4 \underline{\price[ ][]}^2}$, while the diagonal entries are bounded by $\frac{1}{2\underline{\price[ ][]}^2}$.
    Hence, for all $\price \in [\underline{\price[ ][]}, 1]^3$, it holds that $\norm[\grad \excess(\price)]_1 \leq \frac{3}{2\underline{\price[ ][]}^2} + \frac{6}{4 \underline{\price[ ][]}^2} = \frac{3}{\underline{\price[ ][]}^2}$.
    Then, by the mean value theorem, $\excess^{\mathrm{scarf}}$ is $\nicefrac{3}{\underline{\price[ ][]}^2}$-Lipschitz-continuous on $[\underline{\price[ ][]}, 1]^3$, i.e., for all $\price, \otherprice \in [\underline{\price[ ][]}, 1]^3$, $\norm[\excess(\price) - \excess(\otherprice)] \leq \nicefrac{3}{\underline{\price[ ][]}^2} \norm[\otherprice - \price]$.
    Now, 
    since $\kernel$ is $1$-strongly-convex, we have, for all $\price, \otherprice \in \R^3_+$, $\frac{1}{2} \norm[\price - \otherprice]^2 \leq \divergence[\kernel](\price, \otherprice)$. 
    Finally, for all $\price, \otherprice \in [\underline{\price[ ][]}, 1]^3$, 
    \begin{align*}
        \frac{1}{2} \norm[\excess(\price) - \excess(\otherprice)]^2 &\leq \frac{1}{2} \left( \frac{3}{\underline{\price[ ][]}^2} \right)^2 \norm[\price - \otherprice]^2\\ 
        &\leq \left( \frac{3}{\underline{\price[ ][]}^2} \right)^2 \divergence[\kernel](\price, \otherprice).
    \end{align*}
\end{proof}


\section{Arrow-Debreu Competitive Economies as Competitive Walrasian Economies}
\label{sec_app:ad_comp}
An \mydef{Arrow-Debreu economy} $(\numbuyers, \numcommods, \consumptions, \consendow, \util)$, denoted $(\consumptions, \consendow, \util)$ when clear from context, comprises a finite set of $\numcommods \in \N_+$ divisible \mydef{commodities} and $\numconsumers \in \N_+$ \mydef{consumers}.
Each consumer $\consumer \in \consumers$ is characterized by a \mydef{set of consumptions} $\consumptions[\consumer] \subseteq \R^{\numcommods}$, an \mydef{endowment} of commodities $\consendow[\consumer] = \left(\consendow[\consumer][1], \dots, \consendow[\consumer][\numcommods] \right) \in \R^\numconsumers$, and a \mydef{utility function} $\util[\consumer]: \R^{\numcommods} \to \mathbb{R}$ which for any \mydef{consumption} $\consumption[\consumer] \in \consumptions[\consumer]$ describes the utility $\util[\consumer](\consumption[\consumer])$ consumer $\consumer$ derives.\footnote{In line with the literature (see, for instance, \cite{debreu1954representation}), the  value of this utility function should not be interpreted to have any meaning, and the utility function $\util[\consumer]$ should be understood to represent a preference relation $\prefer[\consumer]$ on the space of consumptions $\consumptions[\consumer]$ so that for any two consumptions $\consumption[\consumer], \consumption[\consumer][][][\prime] \in \consumptions$, $\util[\consumer](\consumption[\consumer]) \geq \util[\consumer](\consumption[\consumer][][][\prime]) \implies \consumption[\consumer] \prefer[\consumer] \consumption[\consumer][][][\prime]$. }
We define any collection of per-consumer consumptions $\consumption \doteq (\consumption[1], \hdots, \consumption[\numconsumers]) \in \consumptions$ a \mydef{consumption profile}, where $\consumptions \doteq \bigtimes_{\consumer \in \consumers} \consumptions[\consumer]$ is the \mydef{set of consumption profiles}, and any collection of per-consumer endowments an \mydef{endowment profile} $\consendow \doteq \left(\consendow[1], \hdots, \consendow[\numbuyers] \right) \in \R^{\numconsumers \numcommods}$. 

\begin{remark}
    For ease of exposition, without loss of generality, we restrict ourselves to Arrow-Debreu exchange economies and opt to not present Arrow-Debreu competitive economies (see \citet{arrow-debreu}) which in addition to consumers also contain firms. Nevertheless, our focus on Arrow-Debreu exchange economies is without loss of generality since any firm can be represented as a consumer in an Arrow-Debreu exchange economy by adding an additional commodity into the economy which represents ownership of the firm, setting the consumption space of the new consumer to be equal to the production space of the firm, and its utility function so that it seeks to maximize its consumption of the commodity associated with the firm's ownership. The commodity associated with ownership of the firm should further appear in the endowments of consumers that are supposed to have a contractual claim over the profits of the firms. A similar, albeit much more complicated reduction than described here was proposed earlier by \citet{garg2015markets}, to which we refer the reader for additional details.
\end{remark}

\begin{definition}
    An \mydef{Arrow-Debreu equilibrium} $(\consumption[][][][*], \price[][][*])$ is a tuple comprising consumptions $\consumption[][][][*] \in \R_+^{\numconsumers \times \numbuyers}$ and prices $\price[][][*] \in \simplex[\numcommods]$ s.t.
        \begin{enumerate}
        \item ({\mydef{Utility maximization}})  all consumers $\consumer \in \consumers$, maximize their utility constrained by the value of their endowment: $ \max\limits_{\consumption[\consumer] \in \consumptions[\consumer]: \consumption[\consumer] \cdot \price[][][*] \leq \consendow[\consumer] \cdot \price[][][*]} \util[\consumer](\consumption[\consumer]) \leq \util[\consumer](\consumption[\consumer][][][*]) $;
        \item ({\mydef{Feasibility}}) the consumptions are feasible, i.e., $\sum_{\consumer \in \consumers} \consumption[\consumer][][][*]  \leq \sum_{\consumer \in \consumers} \consendow[\consumer]$;
        \item ({\mydef{Walras' law}}) the value of the demand and the supply are equal, i.e., $\price[][][*] \cdot \left( \sum_{\consumer \in \consumers} \consumption[\consumer][][][*]  - \sum_{\consumer \in \consumers} \consendow[\consumer] \right) = 0$,
    \end{enumerate}
\end{definition}

\begin{assumption}\label{assum:ad_economy}
    Any Arrow-Debreu economy $(\consumptions, \consendow, \util)$ satisfies the following conditions for all consumers $\consumer \in \consumers$:
    \begin{enumerate}
        \item (Closed consumption set) $\consumptions[\consumer]$ is non-empty, bounded from below, closed, and convex;
        \item (Feasible budget set)  There exists a consumption that is strictly less than the consumer's endowment, i.e., for all $\consumer \in \consumers$, there exists $\consumption[\consumer]  \in \consumptions[\consumer]$, s.t. $\consumption[\consumer] < \consendow[\consumer]$;
        \item (Continuity) $\util[\consumer]$ is continuous;
        \item (\sklara{}{Explicit} Quasiconcavity) $\util[\consumer]$ is \sklara{}{explicitly} quasi-concave, i.e., for all $\consumption[\consumer], \consumption[\consumer][][][\prime] \in \R^\numcommods, \lambda \in(0,1)$:
        \begin{align*}
        \util[\consumer](\lambda \consumption[\consumer] + (1-\lambda) \consumption[\consumer][][][\prime]) \geq \min \left\{\util[\consumer](\consumption[\consumer]), \util[\consumer](\consumption[\consumer][][][\prime]) \right\}
        \end{align*}
        \sklara{}{and}
        \begin{align*}
        \util[\consumer](\consumption[\consumer]) > \util[\consumer](\consumption[\consumer][][][\prime]) \implies \util[\consumer](\lambda \consumption[\consumer] + (1-\lambda) \consumption[\consumer][][][\prime]) > \util[\consumer](\consumption[\consumer][][][\prime]);
        \end{align*}
        \item (Non-satiation) $\util[\consumer]$ is non-satiated, i.e.,  $\forall \consumption[\consumer] \in \consumptions[\consumer]$, there exists $\consumption[\consumer][][][\prime] \in \consumptions[\consumer]$ s.t. $\util[\consumer](\consumption[\consumer][][][\prime]) > \util[\consumer](\consumption[\consumer][][][])$.
    \end{enumerate}
\end{assumption}

\begin{definition}[Walrasian Arrow-Debreu Competitive Economy]
    Given an Arrow-Debreu economy $(\numconsumers, \numcommods, \consumptions, \consendow, \util)$, the \mydef{Walrasian Arrow-Debreu competitive economy} $(\numcommods, \excessset)$ is a Walrasian economy with the excess demand correspondence given as:
    \begin{align*}
        \excessset(\price) &= \sum_{\player \in \players} \left[\argmax\limits_{\consumption[\consumer] \in \consumptions[\consumer][\prime]: \consumption[\consumer] \cdot \price \leq \consendow[\consumer] \cdot \price}  \util[\consumer](\consumption[\consumer]) \right] - \sum_{\consumer \in \consumers} \consendow[\consumer] \enspace ,
    \end{align*}
    where $\consumptions[\consumer][\prime] \doteq \left\{\consumption[\consumer] \mid \sum_{k \in \consumers} \consumption[k] \leq \sum_{k \in \consumers} \consendow[k], \consumption[{k}] \in \consumptions[{k}] \right\}$.
\end{definition}

From the proof of Theorem~1 of \citet{arrow-debreu}, we can infer that any Walrasian equilibrium $\price[][][*] \in \simplex[\numgoods]$ of the  Walrasian Arrow-Debreu competitive economy $(\numcommods, \excessset)$ is an Arrow-Debreu equilibrium price of $(\numconsumers, \numcommods, \consumptions, \consendow, \util)$. Further, as shown in the following lemma, the \samy{Walrasian}{} \sklara{}{Arrow-Debreu} economy $(\numcommods, \excessset)$, as the name suggests, gives rise to a Walrasian competitive economy. 

\begin{lemma}[Arrow-Debreu Economies are Walrasian competitive Economies]\label{lemma:ad_economies_are_comp_bounded}
    Consider the Walrasian Arrow-Debreu competitive economy $(\numcommods, \excessset)$ associated with the Arrow-Debreu economy $(\numconsumers, \numcommods, \consumptions, \consendow, \util)$. Then, $\excessset$ satisfies the following:
    \begin{enumerate}
        \item (Homogeneity of degree $0$) For all $\lambda >0$, $\excessset(\lambda \price) = \excessset(\price)$;
        \item (Weak Walras' law) For all $\price \in \R^\numgoods_+$ and $\excess(\price) \in \excessset(\price)$,  $\price \cdot \excess(\price) \leq 0$;
        \item (Non-Satiation) for all $\price \in \R^\numgoods_+$, and $\excess(\price) \in \excessset(\price)$, $\excess(\price) \leq \zeros[\numgoods] \implies \price \cdot \excess(\price) = 0$;
        \item (Continuity) The excess demand correspondence $\excessset$ is upper hemicontinuous on $\simplex[\numgoods]$, non-empty-, compact-, and convex-valued;
        \item (Boundedness) For all $\price \in \R^\numgoods_+$, and $\excess(\price) \in \excessset(\price)$, $\| \excess(\price)\|_\infty < \infty$.
    \end{enumerate}
    
    That is, the Walrasian Arrow-Debreu competitive economy $(\numcommods, \excessset)$ associated with the Arrow-Debreu economy $(\numconsumers, \numcommods, \consumptions, \consendow, \util)$, is a continuous competitive economy which is bounded.
    
\end{lemma}

\begin{proof}[Proof of \Cref{lemma:ad_economies_are_comp_bounded}]

    \textbf{Homogeneity.}
    For all $\lambda >0$, we have:
    \begin{align*}
        \excessset(\lambda \price) &= \sum_{\player \in \players} \left[\argmax\limits_{\consumption[\consumer] \in \consumptions[\consumer][\prime]: \consumption[\consumer] \cdot (\lambda \price) \leq \consendow[\consumer] \cdot (\lambda \price)}  \util[\consumer](\consumption[\consumer]) \right] - \sum_{\consumer \in \consumers} \consendow[\consumer]\\
        &= \sum_{\player \in \players} \left[\argmax\limits_{\consumption[\consumer] \in \consumptions[\consumer][\prime]: \lambda \consumption[\consumer] \cdot \price \leq \lambda \consendow[\consumer] \cdot  \price}  \util[\consumer](\consumption[\consumer]) \right] - \sum_{\consumer \in \consumers} \consendow[\consumer]\\
        &= \sum_{\player \in \players} \left[\argmax\limits_{\consumption[\consumer] \in \consumptions[\consumer][\prime]: \consumption[\consumer] \cdot \price \leq \consendow[\consumer] \cdot  \price}  \util[\consumer](\consumption[\consumer]) \right] - \sum_{\consumer \in \consumers} \consendow[\consumer] = \excessset(\price).
    \end{align*}

    \textbf{\sklara{}{Weak} Walras' law.}
    Fix any $\price \in \R^\numcommods_+$, and let for all consumers $\consumer \in \consumers$, $\consumption[\consumer][][][*] \in \argmax\limits_{\consumption[\consumer] \in \consumptions[\consumer][\prime]: \consumption[\consumer] \cdot \price \leq \consendow[\consumer] \cdot  \price}  \util[\consumer](\consumption[\consumer])$. Then, we have:
    \begin{align*}
        \consumption[\consumer][][][*] \cdot \price  \leq \consendow[\consumer][][][]  \cdot \price.
    \end{align*}
    Summing up across all consumers, and re-organizing, we have:
    \begin{align*}
        \price \cdot \left( \sum_{\consumer \in \consumers} \consumption[\consumer][][][*]   - \sum_{\consumer \in \consumers} \consendow[\consumer][][][]  \right)  \leq 0.
    \end{align*}
    Hence, we have for all $\price \in \R^\numgoods_+$ and $\excess(\price) \in \excessset(\price)$,  $\price \cdot \excess(\price) \leq 0$.
    
    \textbf{Non-Satiation}
    \sklara{}{We now transform \emph{local} non-satiation into \emph{global} non-satiation.} Fix any $\price \in \simplex[\numcommods]$, and let for all consumers $\consumer \in \consumers$, $\consumption[\consumer][][][*] \in \argmax\limits_{\consumption[\consumer] \in \consumptions[\consumer][\prime]: \consumption[\consumer] \cdot \price \leq \consendow[\consumer] \cdot  \price}  \util[\consumer](\consumption[\consumer])$. Suppose by contradiction that $\excess(\price) \leq \zeros[\numgoods]$, but there exists some consumer $\consumer \in \consumers$ s.t.:
    \begin{align*}
        \consumption[\consumer][][][*] \cdot \price  < \consendow[\consumer][][][]  \cdot \price.
    \end{align*}

    Now, by non-satiation \sklara{}{as defined in \Cref{assum:ad_economy}}, there exists \sklara{}{an} $\consumption[\consumer][][][\prime] \in \consumptions[\consumer]$ s.t. $\util[\consumer](\consumption[\consumer][][][\prime]) > \util[\consumer](\consumption[\consumer][][][*])$. As a result, there must also exist \sklara{}{a} $\lambda \in (0, 1)$ s.t. for the consumption $\consumption[\consumer][][][\dagger] \doteq \lambda \consumption[\consumer][][][\prime] + (1- \lambda) \consumption[\consumer][][][*]$, we have: 
    \begin{enumerate}
        \item $\consumption[\consumer][][][\dagger] \in \consumptions[\consumer][\prime]$ since $\consumption[\consumer][][][*] \in \interior(\consumptions[\buyer][\prime])$;
        \item $\util[\consumer](\consumption[\consumer][][][\dagger]) > \util[\consumer](\consumption[\consumer][][][*])$ \deni{$\util[\consumer](\consumption[\consumer][][][\dagger]) = \util[\consumer](\lambda \consumption[\consumer][][][\prime] + (1- \lambda) \consumption[\consumer][][][*]) \geq \min\{\util[\consumer]( \consumption[\consumer][][][\prime]), \util[\consumer](\consumption[\consumer][][][*])  \} > \min\{\util[\consumer]( \consumption[\consumer][][][\prime]), \util[\consumer](\consumption[\consumer][][][*])  \} > \util[\consumer](\consumption[\consumer][][][*])$} since $\util[\consumer]$ is \sklara{}{explicitly} quasiconcave;
        \item $\consumption[\consumer][][][\dagger] \cdot \price  \leq \consendow[\consumer][][][]  \cdot \price$ since the function $\consumption[\consumer] \mapsto \consumption[\consumer] \cdot \price$ is \sklara{continuous}{monotone}.
    \end{enumerate}
    However, this is a contradiction since $\consumption[\consumer][][][*] \in \argmax\limits_{\consumption[\consumer] \in \consumptions[\consumer][\prime]: \consumption[\consumer] \cdot \price \leq \consendow[\consumer] \cdot \price} \util[\consumer](\consumption[\consumer])$.
    
    Hence, for all consumers $\consumer \in \consumers$ we must have:
    \begin{align}
        \consumption[\consumer][][][*] \cdot \price  = \consendow[\consumer][][][*]  \cdot \price.
    \end{align}

    Summing the above across $\consumer \in \consumers$, and re-organizing the expression, we have for all $\excess(\price) \in \excessset(\price)$:
    \begin{align*}
        0 = \price[][][*] \cdot \left( \sum_{\consumer \in \consumers} \consumption[\consumer][][][*]  - \sum_{\consumer \in \consumers} \consendow[\consumer] \right) = \price[][][*] \cdot \excess(\price).
    \end{align*}

    \textbf{Continuity.}

    Since $\consumptions[][\prime]$ is non-empty, compact, and convex, and for all consumers $\consumer \in \consumers$, $\util[\consumer]$ is continuous and quasiconcave, and $\exists \consumption[\consumer]  \in \consumptions[\consumer]$ s.t.\ $\consumption[\consumer] < \consendow[\consumer ]$, the assumptions of Berge's maximum theorem \cite{berge1997topological} hold, and the excess demand $\excessset$ is upper hemicontinuous, non-empty, compact, and convex-valued over $\simplex[\numcommods]$.

    \textbf{Boundedness}
    Since for all consumers $\consumer \in \consumers$, $\consumptions[\consumer]$ is bounded from below, $\consumptions[\consumer][\prime]$ must be bounded as it is bounded from above by $\sum_{\consumer \in \consumers} \consendow[\consumer]$. Hence, for all consumers $\consumer \in \consumers$, $\consumptions[\consumer][\prime]$ is compact. Hence, we must have for all $\price \in \R^\numgoods_+$, and $\excess(\price) \in \excessset(\price)$, \sklara{$\| \excess(\price)\|_\infty < \diam(\consumptions[\consumer][\prime])$}{$| \excess(\price)| \leq \diam(\consumptions[\consumer][\prime])$}.
\end{proof}

\begin{remark}[Linear Utilities in Practice]
\label{rmk:linear_practice}
\sklara{}{The linear utility function class in Arrow-Debreu exchange economies is equivalent to the CES utility function class when $\rho \to 1$. As $\rho \to 1$, while the excess demand correspondence remains upper hemicontinuous by \Cref{lemma:ad_economies_are_comp_bounded}, the Bregman continuity coefficient of the excess demand $\lsmooth \to \infty$. This affects the price space in a way that even pathwise Bregman continuity with a more meaningful coefficient may not be recoverable for a mirror extrat\^atonnement path from an initial iterate $\price[][][(0)]$ to an equilibrium $\price[][][*]$.

To see why, let $A$ and $B$ be two goods, $v_A$ and $v_B$ a buyer's valuations for these goods, and $p_A$ and $p_B$ the corresponding prices. If $\frac{v_A}{p_A}=\frac{v_B}{p_B}$, the buyer is indifferent between the goods. The excess demand for this price, buyer, and goods is thus the line between all of $A$ and all of $B$. If $\numgoods$ is the number of goods, the preimage of these jumps on the $(\numgoods - 1)$-dimensional price simplex $\simplex[\numgoods]$ is an $(\numgoods - 2)$-dimensional ``cobweb'' that intersects the interior of $\simplex[\numgoods]$ and is ``rooted'' in each vertex. See \Cref{fig:linear_excess} for the excess demand $\excess$ in this case.}

\begin{figure}[h]
    \centering
    \includegraphics[width=0.75\linewidth]{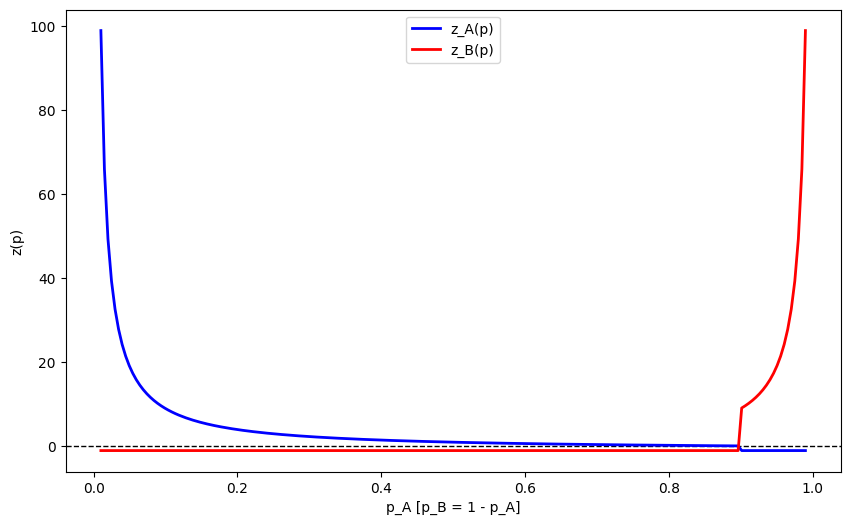}
    \caption{\sklara{}{Excess demand $\excess$ (unbounded for demonstrative purposes) on the price space $\simplex[2]$ projected onto the unit segment. A buyer is endowed with $1.0$ units of goods $A$ and $B$ and has valuations $v_A=0.9$ and $v_B=0.1$.}}
    \label{fig:linear_excess}
\end{figure}
\end{remark}

\section{Classes of Variationally Stable Economies on $\simplex[\numgoods]$}\label{sec_app:var_stable_classes}

We now discuss some important classes of Walrasian economies which are variationally stable on $\simplex[\numgoods]$. The most basic class of Walrasian economies which are variationally stable on $\simplex[\numgoods]$ are those which satisfy the law of supply and demand. Intuitively, these Walrasian economies are those for which the excess demand is downward sloping.

\begin{definition}[Law of supply and demand economies]\label{def:law_of_supply_and_demand}
    Given a Walrasian economy $(\numgoods, \excessset)$, an excess demand correspondence is said to satisfy the \mydef{law of supply and demand} iff     
    \begin{align}
        \innerprod[{\excess(\otherprice) - \excess(\price)}][{\otherprice - \price}] \leq 0 && \text{ for all $\excess(\price) \in \excessset(\price), \excess(\otherprice) \in \excessset(\otherprice)$}.
    \end{align}
\end{definition}

We note that the excess demand of a Walrasian economy satisfies the law of supply and demand iff $-\excessset$ is monotone. This implies that $-\excessset$ is quasimonotone, and hence for any non-empty and compact price space $\pricespace \subseteq \R^\numgoods_+$ the VI $(\pricespace, -\excessset)$ satisfies the Minty condition (see Lemma 3.1 of \citet{he2017solvability}), meaning that any Walrasian economy which satisfies the law of supply and demand is variationally stable on $\pricespace$. 

Another important class of Walrasian economies which are variationally stable on $\simplex[\numgoods]$ is the class of Walrasian economies which satisfy the weak gross substitutes condition. Intuitively, these are Walrasian economies for which the excess demand for a given good only increases when the price of some other good increases. 
While we omit the proof as it is involved, we note that any continuous balanced weak gross substitutes Walrasian economy $(\numgoods, \excessset)$ which satisfies Walras' law, i.e., (for all $\price \in \R^\numgoods_+, \excess(\price) \in \excessset(\price)$, $\price \cdot \excess(\price)=0$) is a subset of the class of variationally stable economies on $\pricespace \subseteq \R^\numgoods_+$ for any non-empty and compact price space $\pricespace$ (see, for instance Lemma~5 of \citet{arrow1959stability}).

\begin{definition}[Weak Gross Substitutes economies]\label{def:wgs}
    Given a Walrasian economy $(\numgoods, \excessset)$, an excess demand correspondence is said to satisfy the \mydef{weak gross substitutes condition} iff for all $\price, \otherprice \in \R^\numgoods_+$ s.t. for some $k \in \goods$, $\otherprice[k] > \price[k]$ and for all $\good \neq k, \otherprice[\good] = \price[\good]$, we have:
    \begin{align}
        \excess[\good](\otherprice) \geq \excess[\good](\price) 
        && \text{ for all $\excess(\price) \in \excessset(\price), \excess(\otherprice) \in \excessset(\otherprice)$}.
    \end{align}
\end{definition}

Going further, we can show that any Walrasian economy which satisfies the well-known weak axiom of revealed preferences \cite{afriat1967construction, arrow-hurwicz} is variationally stable on $\simplex[\numgoods]$ (and more generally on any non-empty and compact price space $\pricespace \subseteq \R^\numgoods$). To this end, let us first define the weak axiom of revealed preferences for balanced economies. 

\begin{definition}[WARP excess demand]\label{def:warp}
    Given a Walrasian economy $(\numgoods, \excessset)$, an excess demand correspondence is said to satisfy the \mydef{weak axiom of revealed preferences} (\mydef{WARP}) iff for all $\excess(\price) \in \excessset(\price), \excess(\otherprice) \in \excessset(\otherprice)$:     
    \begin{align*}
        \innerprod[{\excess(\otherprice)}][{\price}] \leq \innerprod[{\excess(\otherprice)}][{\otherprice}] \text{ and } \excess(\price) \neq \excess(\otherprice)  \implies \innerprod[{\excess(\price)}][{\otherprice}] > \innerprod[{\excess(\price)}][{\price}].
    \end{align*}
\end{definition}

\begin{remark}
    This definition of (WARP) is adapted to arbitrary Walrasian economies and as such is a generalization of the usual definition for economies which satisfy Walras' law (i.e., for all $\price \in \R^\numgoods_+$, $\price \cdot \excess(\price)=0$), which requires that $\excessset$ is singleton-valued, and $\innerprod[{\excess(\otherprice)}][{\price}] \leq 0 
    \text{ and } \excess(\price) \neq \excess(\otherprice)  
    \implies \innerprod[{\excess(\price)}][{\otherprice}] > 0$ (i.e., for all $\price \in \R^\numgoods_+$, $\price \cdot \excess(\price)=0$). 
\end{remark}

As we show next, WARP implies that $-\excessset$ is pseudomonotone in balanced economies.\footnote{To be more precise, we note that an excess demand function \sklara{$-\excessset$}{$-\excess$} satisfies WARP iff \sklara{$-\excessset$}{$-\excess$} is strictly pseudomonotone\sklara{}{, i.e., for all distinct $\price, \otherprice \in \pricespace$, $\left< -\excess(\price), \otherprice - \price \right> \geq 0 \, \text{ implies } \left< -\excess(\otherprice), \otherprice - \price\right> > 0$.}. However, as this result will not be used we present the more general result.}

\begin{lemma}[WARP $\implies$ pseudomonotone ]\label{lemma:warp_implies_pseudomonotone}
    If the excess demand correspondence $\excessset$ of a Walrasian economy $(\numgoods, \excessset)$ satisfies WARP, then $-\excessset$ is pseudomonotone.
\end{lemma}
\begin{proof}
    Suppose that $\excessset$ satisfies WARP, and that $\innerprod[{-\excess(\otherprice)}][{\otherprice - \price}] = \innerprod[{\excess(\otherprice)}][{\price - \otherprice}] \leq 0$\sklara{, then we have $\innerprod[{-\excess(\price)}][{\otherprice - \price}] = \innerprod[{\excess(\price)}][{\price - \otherprice}]$}{.}
    If $\excess(\price) \neq \excess(\otherprice)$, then, by WARP, we have $\innerprod[{\excess(\price)}][{\price - \otherprice}] < 0$. 
    
    Otherwise, if $\excess(\price) = \excess(\otherprice)$, then we have:
    \begin{align*}
        \innerprod[{\excess(\price)}][{\price - \otherprice}] =  \innerprod[{\excess(\otherprice)}][{\price - \otherprice}] \leq 0.
    \end{align*}

    That is, if $\excessset$ satisfies WARP, we have:
    \begin{align*}
        \innerprod[{-\excess(\otherprice)}][{\otherprice - \price}] \leq 0 \implies \innerprod[{-\excess(\price)}][{\otherprice - \price}] \leq 0.
    \end{align*}
    
    Hence, $-\excessset$ is pseudomonotone.
\end{proof}

An important consequence of \Cref{lemma:warp_implies_pseudomonotone} is that since $-\excessset$ is pseudomonotone, for any non-empty and compact price space $\pricespace \subseteq \R^\numgoods_+$ the VI $(\pricespace, -\excessset)$ satisfies the Minty condition (see Lemma 3.1 of \citet{he2017solvability}). As such, we have the following corollary of \Cref{lemma:warp_implies_pseudomonotone}. 

\begin{corollary}[WARP $\implies$ Minty's condition]\label{cor:warp_iff_minty}
    Any Walrasian economy which satisfies WARP is variationally stable on any non-empty and compact price space $\pricespace \subseteq \R^\numgoods_+$. 
\end{corollary}

\section{Experiment Details}
\label{sec_ap:experiments}
\paragraph{Computational Resources}
\salp{Our experiments were run on MacOS machine with 8GB RAM and an Apple M1 chip, and took about 10 minutes to run. Only CPU resources were used.}{Our experiments (excluding grid search experiments) were run on Google Colaboratory using Python 3 Google Compute Engine, and took around 300 seconds to run with only CPU resources. Relevant specifications can be found in Table~\ref{table:exp_specs}}.

\begin{table}[H]
   \caption{Specifications of the Python 3 Google Compute Engine instance used for experiments\sklara{}{.}}\label{table:exp_specs}
   \begin{center}
    \begin{tabular}{c@{\hspace{2em}} c@{\hspace{2em}}}
    \textbf{Component} & \textbf{Specification} \\ \hline
    CPU & 2x Intel\textsuperscript{\textregistered{}} Xeon\textsuperscript{\textregistered{}} CPU @ 2.20GHz \\ \hline
    RAM & 12.7 GB \\ \hline
    Disk Storage & 107.7 GB \\ \hline
    Operating System & Ubuntu 22.0.4.5 LTS\\ \hline
    \end{tabular}
   \end{center}
\end{table}

\paragraph{Programming Languages, Packages, and Licensing}
We ran our experiments in Python 3.12.12 \cite{python2025}
, using NumPy \cite{numpy},  Jax \cite{jax2018github}, and  JaxOPT \cite{jaxopt_implicit_diff}.
All figures were graphed using Matplotlib \cite{matplotlib}. 

Python software and documentation are licensed under the PSF License Agreement. Numpy is distributed under a liberal BSD license. Pandas is distributed under a new BSD license. Matplotlib only uses BSD compatible code, and its license is based on the PSF license.

\paragraph{Experimental Setup Details}

Each economy is initialized using a random seed to ensure reproducibility. Each consumer is assigned an initial endowment, drawn from a uniform distribution:
$
\consendow[][][][\prime] \sim \mathrm{Unif}(10^{-6}, 1), \quad \forall \consumer \in [\numconsumers], \good \in [\numgoods].
$
For numerical stability, we restrict the total economy-wide aggregate supply of each commodity to remain fixed at $10$\footnote{This is without loss of generality since commodities are divisible.}, to this end we normalize the endowments of consumers for all $\good \in \goods$,  $\consumer \in \consumers$ to obtain their final endowment:
\[
\consendow[\consumer][\good] \doteq \frac{10 \consendow[\consumer][\good][][\prime]}{\sum_{\consumer \in \consumers} \consendow[\consumer][\good][][\prime]}.
\]

Each consumer’s valuation of each commodity is drawn from a uniform distribution, i.e., for all $\good \in \goods$, $\consumer \in \consumers$:
\[
\valuation[\consumer][\good] \sim \mathrm{Unif}(0, 1) .
\]

For any CES consumer $\consumer \in \consumers$, the elasticity of substitution parameter $\rho_\consumer$, is drawn as follows from the uniform distribution for substitutes and complements consumers respectively:
\begin{align*}
&\rho_\consumer^{\text{substitutes}} \sim \mathrm{Unif}(0.6, 0.9) \
&\rho_\consumer^{\text{complements}} &\sim \mathrm{Unif}(-1000, -1)
\end{align*}

The initial price vector $\price[][0]$ for the algorithms is drawn from a uniform distribution s.t. for all $\good \in \goods$:
\[
\price[\good][0] \sim \mathrm{Unif}(1, 10) .
\]
We note that while we initialize the prices between $1$ and $10$ for numerical stability, this choice is without loss of generality since the excess demand is homogeneous of degree $0$.

To summarize. Given a random seed, the initialization process consists of:
1) Sampling endowments from a uniform distribution and normalizing them to ensure total supply constraints; 2) sampling valuations from a uniform distribution; 3) sampling substitution parameters for CES consumers, 4) generating an initial price vector.

\alp{Should we add the following section? If so, should we use $z$ instead of $f$ to make it clear that pathwise Bregman continuity is assumed for excess demand? }
\paragraph{Pathwise Bregman Continuity}

The pathwise Bregman continuity is verified for every experiment. For the choice of kernel function $\kernel(\price) \doteq \frac{1}{2}\|\price\|^2$, Bregman divergence corresponds to half of Euclidean distance, therefore Bregman continuity becomes: 
\begin{align*}
    \frac{1}{2} \norm[\vioper({\vartuple[][][k+0.5]}) - \vioper({\vartuple[][][k])}]^2 &\leq \lsmooth^2 \frac{1}{2}\norm[{\vartuple[][][k+0.5]} - {\vartuple[][][k]}]^2 \\
    \frac{\norm[\vioper({\vartuple[][][k+0.5]}) - \vioper({\vartuple[][][k])}]^2}{\norm[{\vartuple[][][k+0.5]} - {\vartuple[][][k]}]^2}  &\leq \lsmooth^2 \\
    \frac{\norm[\vioper({\vartuple[][][k+0.5]}) - \vioper({\vartuple[][][k])}]}{\norm[{\vartuple[][][k+0.5]} - {\vartuple[][][k]}]}  &\leq \lsmooth.
\end{align*}

Because Theorem~\ref{thm:mirror_extragradient_global_convergence} relies on the existence of a $\lsmooth \in (0, \frac{1}{\sqrt{2}\learnrate[ ][ ]}]$, we have:
\begin{align*}
    \frac{\norm[\vioper({\vartuple[][][k+0.5]}) - \vioper({\vartuple[][][k])}]}{\norm[{\vartuple[][][k+0.5]} - {\vartuple[][][k]}]} &\leq \lsmooth \leq \frac{1}{\sqrt{2}\learnrate[ ][ ]} \\
    \frac{\norm[\vioper({\vartuple[][][k+0.5]}) - \vioper({\vartuple[][][k])}]}{\norm[{\vartuple[][][k+0.5]} - {\vartuple[][][k]}]} &\leq \frac{1}{\sqrt{2}\learnrate[ ][ ]}.
\end{align*}

For every experiment, the value of the left hand-side term (the Lipschitz coefficient) is recorded at each iteration of the \emph{extrat\^atonnement} algorithm. The maximum value among these recorded Lipschitz coefficients is then checked to be less than $\frac{1}{\sqrt{2}\learnrate[ ][ ]}$, which is sufficient to show that the Bregman continuity is satisfied. This empirically verifies that excess demand is pathwise Bregman-continuous during the experiments. 

\paragraph{Grid Search Experiments}

We performed a grid search for a suitable step size for our experiments with Arrow-Debreu economies.
To do so, we discretized an interval into 200 equally-spaced step sizes.
Table~\ref{table:grid_intervals} lists the intervals that were used for each economy.

\begin{table}[hbtp]
   \caption{The intervals $[\eta_{\text{min}}, \eta_{\text{max}}]$ used for grid search on each Arrow-Debreu economy.}
   \label{table:grid_intervals}
    \begin{center}
        \renewcommand{\arraystretch}{1.1}
        \begin{tabular}{| m{1.25cm} | m{1.75cm} m{1.75cm}|}
            \hline
            Exp No. & $\eta_{\text{min}}$ & $\eta_{\text{min}}$ \\ 
            \hline \hline
            1 & 0.1 & 20 \\
            2 & 0.001 & 10 \\
            3 & 0.0001 & 0.001 \\
            4 & 0.01 & 0.5 \\
            5 & 0.000001 & 0.001 \\
            6 & 0.000005 & 0.0003 \\
            7 & 0.00005 & 0.005 \\ \hline
        \end{tabular}
        \renewcommand{\arraystretch}{1.0}
    \end{center}
\end{table}

\begin{figure}[h!]
    \centering
    
    \begin{subfigure}{0.46\textwidth}
        \centering
        \includegraphics[width=\linewidth, height=60mm, keepaspectratio]{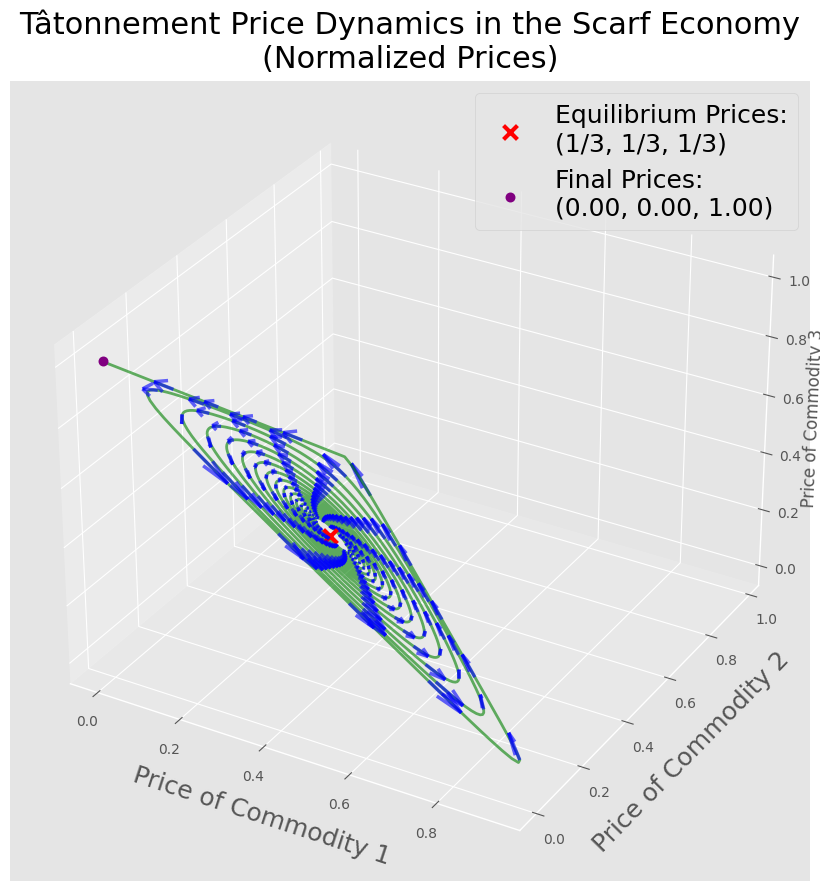}
    \end{subfigure}
    \hfill
    \begin{subfigure}{0.46\textwidth}
        \centering
        \includegraphics[width=\linewidth, height=60mm, keepaspectratio]{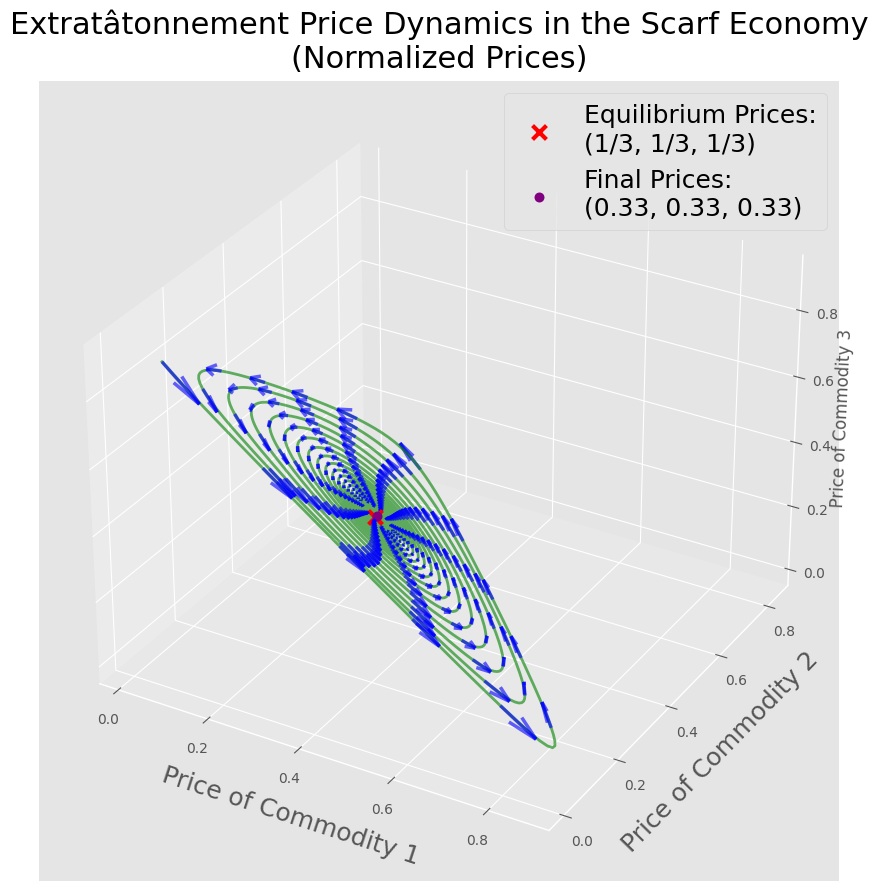}
    \end{subfigure}

    \caption{Phase portraits of t\^atonnement and extrat\^atonnement for the Scarf economy. All experiments were run on the unit box with step size $\eta = \nicefrac{7}{180}$. (Normalized) price plots show the price trajectories generated by t\^atonnement and extrat\^atonnement in the Scarf economy. They are more triangular, but otherwise not qualitatively different.}
\end{figure}

\section{Additional Experiments with Arrow-Debreu Economies}
\label{sec_ap:ad_experiments}
In \Cref{fig:experiment_plots_opt}, we plot the Walrasian deviations at prices generated by each step of extrat\^atonnement in our seven Arrow-Debreu exchange economies, using step sizes that correspond to the (global) minimum best-iterate Walrasian deviation without satisfying pathwise Bregman continuity (See~\Cref{table:ad_grid_search_step_sizes}). Similar to the convergence plots in \Cref{fig:experiment_plots_alt}, which use the step sizes that correspond to the minimum best-iterate Walrasian deviation while satisfying pathwise Bregman continuity, we observe fast convergence in all markets that are not inhabited by any linear consumers (economies 1--4 and 7); otherwise, convergence is slower.

\begin{figure}[htbp]
    \centering
    \begin{subfigure}{0.3\textwidth}
        \centering
        \includegraphics[height=30mm, width=45mm]{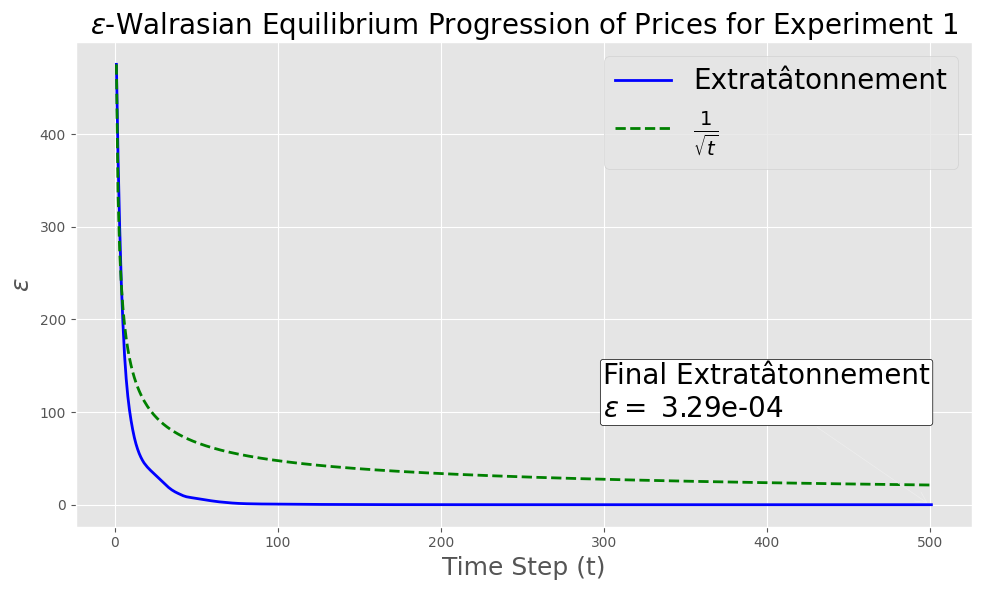}
    \end{subfigure}
    \hfill
    \begin{subfigure}{0.3\textwidth}
        \centering
        \includegraphics[height=30mm, width=45mm]{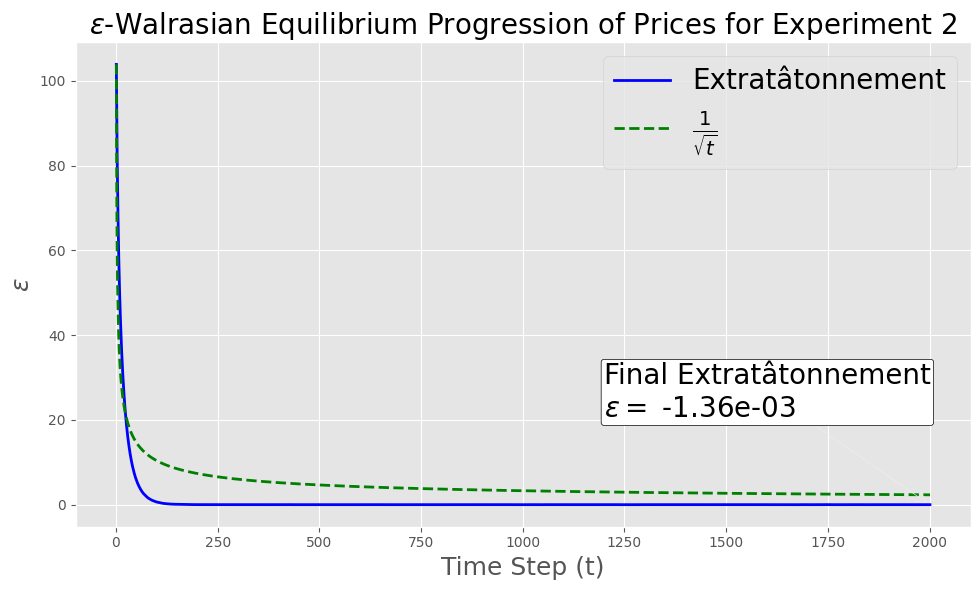}
    \end{subfigure}
    \hfill
    \begin{subfigure}{0.3\textwidth}
        \centering
        \includegraphics[height=30mm, width=45mm]{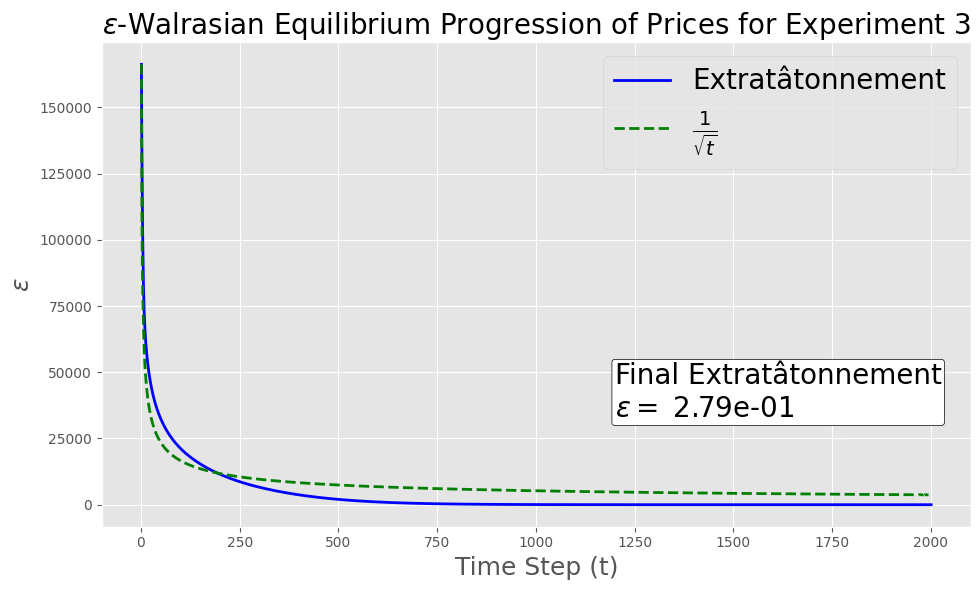}
    \end{subfigure}

    \vspace{0.3cm}

    \begin{subfigure}{0.3\textwidth}
        \centering
        \includegraphics[height=30mm, width=45mm]{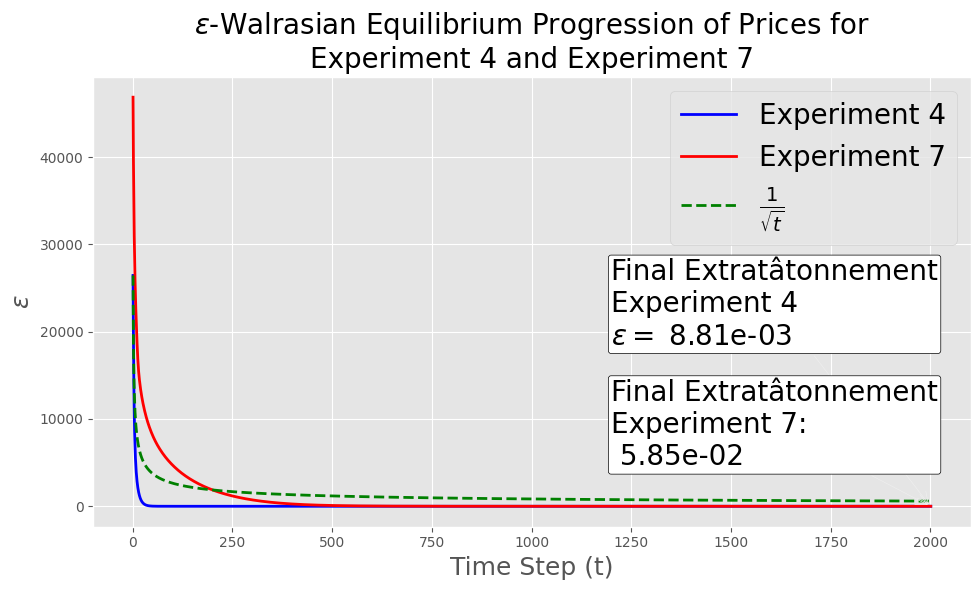}
    \end{subfigure}
    \hfill
    \begin{subfigure}{0.3\textwidth}
        \centering
        \includegraphics[height=30mm, width=45mm]{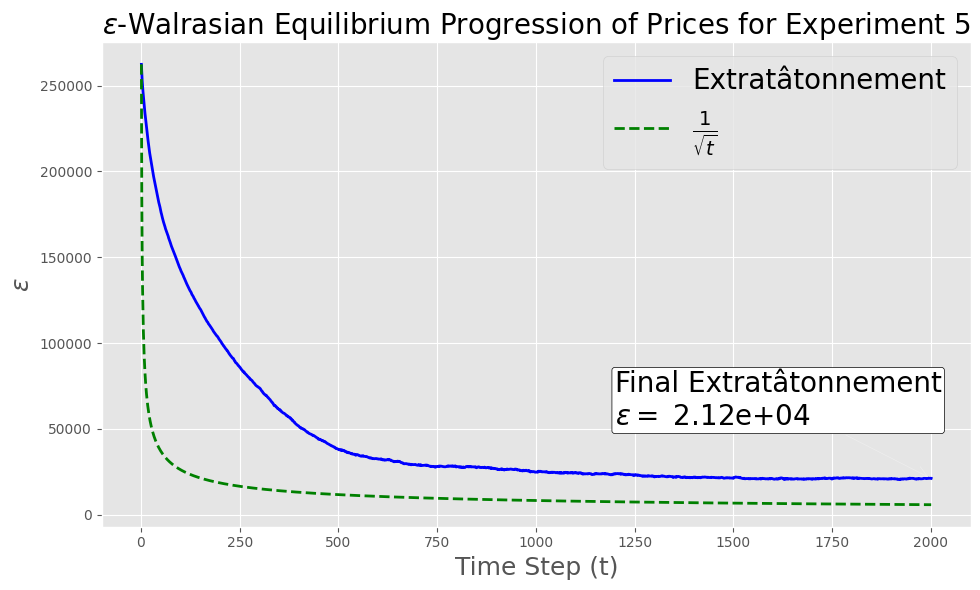}
    \end{subfigure}
    \hfill
    \begin{subfigure}{0.3\textwidth}
        \centering
        \includegraphics[height=30mm, width=45mm]{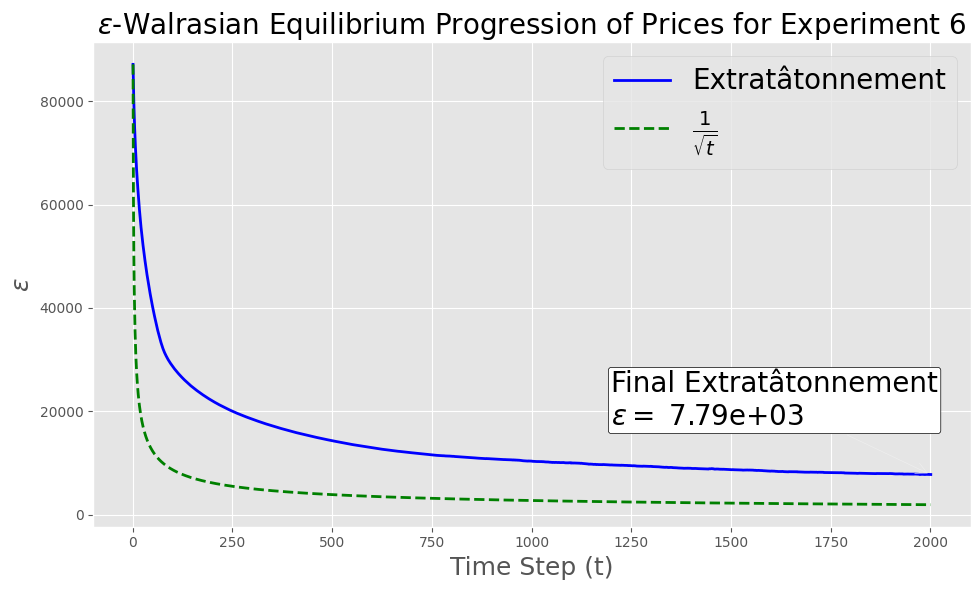}
    \end{subfigure}

    \caption{\sklara{$\epsilon$}{$\varepsilon$}-Walrasian deviations in Experiments 1--7 using step sizes that correspond to the minimum best-iterate Walrasian deviation without satisfying pathwise Bregman continuity. 
    These trajectories are similar to those generated using step sizes that minimize Walrasian  deviation while satisfying pathwise Bregman continuity
    (see \Cref{fig:experiment_plots_alt}).}
    \label{fig:experiment_plots_opt}
\end{figure}

\begin{table}[hbtp]
   \caption{Grid-searched step sizes that minimize Walrasian deviations in economies 1--7. \emph{Global} (\emph{Bregman Continuous}) represents the step sizes that minimize Walrasian deviations without (while) satisfying pathwise Bregman continuity. \amy{are there any interesting conclusions to be drawn from these numbers? maybe not. maybe we don't actually need this table? \@Alp? if we want to keep this table, then we should add some insights about the Global column to the text.}}
   \label{table:ad_grid_search_step_sizes}
    \begin{center}
        \renewcommand{\arraystretch}{1.1}
        \begin{tabular}{|c| c  c |}
            \hline
            No. & Step Size (Global)  & Step Size (Bregman Continuous) \\ 
            \hline 
            \hline
            1 & $4.00 \times 10^{0}$ & $2.30 \times 10^{0}$ \\
            2 & $4.87 \times 10^{0}$ & $6.54 \times 10^{-1}$ \\
            3 & $1.00 \times 10^{-3}$ & $9.68 \times 10^{-4}$ \\
            4 & $8.39 \times 10^{-2}$ & $5.68 \times 10^{-2}$ \\
            5 & $2.12 \times 10^{-4}$ & $3.11 \times 10^{-5}$ \\
            6 & $3.00 \times 10^{-4}$ & $7.17 \times 10^{-5}$ \\
            7 & $4.90 \times 10^{-3}$ & $3.86 \times 10^{-3}$ \\
            \hline
        \end{tabular}
        \renewcommand{\arraystretch}{1.0}
    \end{center}
\end{table}

\end{document}